\documentclass{llncs}

\usepackage{tabto}
\usepackage{amssymb}
\usepackage{setspace}
\usepackage{smartref} 
\usepackage{fullpage}
\usepackage{listings}
\usepackage[lowtilde]{url}
\usepackage{wrapfig}
\usepackage{paralist}
\usepackage{pgf}

\usepackage{tabularx}
\newcolumntype{Y}{>{\centering\arraybackslash} X}

\newcommand{\comnospace}{\mbox{$\triangleright$}}
\newcommand{\com}{\mbox{\comnospace\ }}

\newcommand{\func}[1]{\mbox{\sc #1}}

\newcommand{\cas}{\mbox{CAS}}

\newcommand{\true}{\textsc{True}}
\newcommand{\false}{\textsc{False}}

\newcommand{\node}{\mbox{Node-record}}
\newcommand{\op}{\mbox{\sct -record}}
\newcommand{\rec}{\mbox{Data-record}}

\newcommand{\info}{\textit{info}}

\newcommand{\llt}{\func{LLX}}

\newcommand{\sct}{\func{SCX}}
\newcommand{\vlt}{\func{VLX}}
\newcommand{\del}{\func{Delete}}
\newcommand{\ins}{\func{Insert}}

\newcommand{\tryins}{\func{TryInsert}}
\newcommand{\trydel}{\func{TryDelete}}
\newcommand{\tryrebalance}{\func{TryRebalance}}

\newcommand{\cleanup}{\func{Cleanup}}
\newcommand{\search}{\func{Search}}

\newcommand{\fail}{\textsc{Fail}}
\newcommand{\finalized}{\textsc{Finalized}}

\newcommand{\nil}{\textsc{Nil}}

\newcommand{\presctlinked}{1}
\newcommand{\presctabainit}{2}
\newcommand{\presctaba}{3}

\newcommand{\conmarkallremovedrecs}{3}

\newcommand{\ritalic}[1]{r_{\mbox{\tiny\textsc{#1}}}}
\newcommand{\rx}{\ritalic x}
\newcommand{\rl}{\ritalic l}
\newcommand{\rr}{\ritalic r}
\newcommand{\rxx}{\ritalic{xx}}
\newcommand{\rxl}{\ritalic{xl}}
\newcommand{\rxr}{\ritalic{xr}}
\newcommand{\rxll}{\ritalic{xll}}
\newcommand{\rxlr}{\ritalic{xlr}}
\newcommand{\rxrl}{\ritalic{xrl}}
\newcommand{\rxrr}{\ritalic{xrr}}

\newcommand{\rxxl}{\ritalic{xxl}}
\newcommand{\rxxr}{\ritalic{xxr}}

\newcommand{\rxxrl}{\ritalic{xxrl}}
\newcommand{\rxxrr}{\ritalic{xxrr}}

\newcommand{\rxxrll}{\ritalic{xxrll}}
\newcommand{\rxxrlr}{\ritalic{xxrlr}}
\newcommand{\rxxrrl}{\ritalic{xxrrl}}
\newcommand{\rxxrrr}{\ritalic{xxrrr}}

\newcommand{\rxxrlll}{\ritalic{xxrlll}}
\newcommand{\rxxrllr}{\ritalic{xxrllr}}
\newcommand{\rxxrlrl}{\ritalic{xxrlrl}}
\newcommand{\rxxrlrr}{\ritalic{xxrlrr}}

\newcommand{\uitalic}[1]{u_{\mbox{\tiny\textsc{#1}}}}
\newcommand{\ux}{\uitalic x}
\newcommand{\ul}{\uitalic l}

\newcommand{\uxx}{\uitalic{xx}}
\newcommand{\uxl}{\uitalic{xl}}
\newcommand{\uxr}{\uitalic{xr}}
\newcommand{\uxll}{\uitalic{xll}}
\newcommand{\uxlr}{\uitalic{xlr}}
\newcommand{\uxrl}{\uitalic{xrl}}
\newcommand{\uxrr}{\uitalic{xrr}}

\newcommand{\uxlrl}{\uitalic{xlrl}}
\newcommand{\uxlrr}{\uitalic{xlrr}}

\newcommand{\nitalic}[1]{n_{\mbox{\tiny\textsc{#1}}}}
\newcommand{\nL}{\nitalic{l}}
\newcommand{\nR}{\nitalic{r}}
\newcommand{\nLL}{\nitalic{ll}}
\newcommand{\nLLL}{\nitalic{lll}}
\newcommand{\nLR}{\nitalic{lr}}

\def\pwidth{2.5cm}
\def\indentamt{0mm}

\newcommand{\dline}[2]{#1\\ \mbox{\hspace{\indentamt}#2}}
\newcommand{\tline}[3]{\dline{#1}{#2}\\ \mbox{\hspace{\indentamt}#3}}

\lstdefinestyle{nonumbers}{numbers=none}
\newcommand{\preplisting}{\lstset{gobble=1, numbers=left, numberstyle=\tiny, numberblanklines=false, numbersep=-8.5pt, firstnumber=last, escapeinside={//}{\^^M}, breaklines=true, basicstyle=\footnotesize, keywordstyle=\bfseries, morekeywords={type,subtype,break,if,else,end,loop,while,do,done,exit,goto, when,then,return,read,and,or,not,,for,each,boolean,procedure,invoke,next,iteration,until}}}
\newcommand{\prepnewlisting}{\lstset{gobble=1, numbers=left, numberstyle=\tiny, numberblanklines=false, numbersep=-8.5pt, firstnumber=1, escapeinside={//}{\^^M}, breaklines=true, basicstyle=\footnotesize, keywordstyle=\bfseries, morekeywords={type,subtype,break,if,else,end,loop,while,do,done,exit,goto, when,then,return,read,and,or,not,for,each,boolean,procedure,invoke,next,iteration,until}}}

\newtheorem{thm}{Theorem}
\newtheorem{obs}[thm]{Observation}
\newtheorem{lem}[thm]{Lemma}
\newtheorem{cor}[thm]{Corollary}
\newtheorem{con}[thm]{Constraint}
\newtheorem{defn}[thm]{Definition}

\newcommand{\ourcomments}[2]{}

\newcommand{\edited}{}

\newcommand{\eric}[1]{\ourcomments{#1}{Eric}}

\newcommand{\strike}[1]{}

\newcommand{\vone}[1]{}

\newcommand{\after}[1]{}

\begin{document}


\title{A General Technique for Non-blocking Trees} 
\author{
Trevor Brown$^1$, Faith Ellen$^1$ and Eric Ruppert$^2$}
\institute{
University of Toronto$^1$, York University$^2$}

\maketitle

\begin{abstract}

We describe a general technique for obtaining provably correct, non-blocking implementations of a large class of tree data structures where pointers are directed from parents to children. Updates are permitted to modify any contiguous portion of the tree atomically. Our non-blocking algorithms make use of the \llt, \sct\ and \vlt\ primitives, which are multi-word generalizations of the standard LL, SC and VL primitives and have been implemented from single-word CAS~\cite{paper1}.

To illustrate our technique, we describe how it can be used in a fairly straightforward way to obtain a non-blocking implementation of a chromatic tree, which is a relaxed variant of a red-black tree. The height of the tree at any time is $O(c+ \log n)$, where $n$ is the number of keys and $c$ is the number of updates in progress. We provide an experimental performance analysis which demonstrates that our Java implementation of a chromatic tree rivals, and often significantly outperforms, other leading concurrent dictionaries.
\end{abstract}

\section{Introduction}\label{sec-intro}

The binary search tree (BST) is among the most important data structures.
Previous concurrent implementations of balanced BSTs without locks
either used coarse-grained transactions, which limit concurrency, or lacked rigorous proofs of correctness.
In this paper, we describe a general technique for implementing 
\textit{any} data structure based on a down-tree (a directed acyclic graph 
of indegree one),
with updates that modify any connected subgraph of the tree atomically.
The resulting implementations are non-blocking, which means that some process is always guaranteed to make progress, even if processes crash.
Our approach drastically simplifies the task of proving correctness.
This makes it feasible to develop provably correct implementations of non-blocking balanced BSTs with fine-grained synchronization
(i.e., with updates that synchronize on a small constant number of nodes).

As with all concurrent implementations, the implementations obtained using our technique are more efficient if each update to the data structure involves a small number of nodes near one another.
We call such an update {\em localized}.
We use \emph{operation} to denote
an operation of the abstract data type (ADT) being implemented by the data structure.
Operations that cannot modify the data structure are called \emph{queries}.
For some data structures, such as Patricia tries and leaf-oriented BSTs,
operations modify the data structure using a single localized update.
In some other data structures, operations that modify the data structure
can be split into several
localized updates that can be freely interleaved.

A particularly interesting application of our technique is to implement \textit{relaxed-balance} versions of sequential data structures efficiently.
Relaxed-balance data structures decouple
updates that rebalance the data structure from operations,
and allow
updates that accomplish rebalancing to be delayed and freely interleaved with
other updates.
For example, a chromatic tree is a relaxed-balance version of a red-black tree (RBT) which splits up the insertion or deletion
of a key and any subsequent rotations into a sequence of 
localized updates.
There is a rich literature of relaxed-balance versions of sequential
data structures \cite{DBLP:journals/acta/Larsen98},
and several papers (e.g., \cite{LOS01}) have described general techniques
that can be used to easily produce
them from large classes of
existing sequential data structures.
The small number of nodes involved in each update makes relaxed-balance data structures perfect candidates for efficient implementation using our technique.

\subsubsection*{Our Contributions}

\begin{itemize}
\item We provide a simple template that can be filled in to obtain an implementation of any update for a data structure based on a down-tree.
We prove that any data structure that follows our template for all of its updates will automatically be linearizable and non-blocking.
The template takes care of all process coordination, so the data structure designer is able to think of updates as atomic steps.

\item  To demonstrate the use of our template, we provide a complete, provably correct, non-blocking linearizable implementation of a chromatic tree \cite{NS96}, which is a
relaxed-balanced version of a RBT.
To our
knowledge, this is the first provably correct, non-blocking balanced BST implemented using fine-grained synchronization.
Our chromatic trees always have height $O(c+\log n)$, where $n$ is the number of keys stored in the tree and $c$ is the number of insertions and deletions that are in progress (Section~\ref{sec-height}).

\item 
We show that sequential implementations of some queries
are linearizable,
even though they completely ignore concurrent updates.
For example, an ordinary BST search (that works when there is no concurrency) also works in our chromatic tree.
Ignoring updates makes searches very fast.
We also describe how to perform
successor queries in our chromatic tree, which
interact properly with updates that follow our template
(Section \ref{sec-chromatic-succ}).

\item We show experimentally that our Java implementation of a chromatic tree rivals, and often significantly outperforms, known highly-tuned concurrent dictionaries,
over a variety of workloads, contention levels and thread counts.
For example, with 128 threads, our algorithm outperforms Java's non-blocking skip-list by 13\% to 156\%, the lock-based AVL tree of Bronson et~al. by 63\% to 224\%, and a RBT that uses software transactional memory (STM) by 13 to 134 times (Section~\ref{sec-exp}).
\end{itemize}

\section{Related Work} \label{sec-related}

There are many lock-based implementations of search tree data structures.
(See \cite{AKKMT12,BCCO10:ppopp} for state-of-the-art examples.)
Here, we focus on implementations that do not use locks.
Valois \cite{Val95} sketched an implementation of
non-blocking
node-oriented
BSTs from CAS.
Fraser \cite{Fra03} gave a non-blocking BST 
using 8-word CAS, but did not provide 
a full proof of correctness.  He also described how
multi-word CAS can be implemented 
from single-word CAS instructions.
Ellen et al.~\cite{EFRB10:podc} gave a
provably correct, non-blocking
implementation of leaf-oriented
BSTs directly from single-word CAS.
A similar approach was used for $k$-ary search trees \cite{BH11}
and Patricia tries \cite{Shafiei}.
All three
used the cooperative technique originated by Turek, Shasha and 
Prakash~\cite{TSP92} and Barnes~\cite{Barnes}.
Howley and Jones \cite{HJ12} used a similar approach to build
node-oriented BSTs.  They tested
their implementation using a model checker, but did not prove it correct.
Natarajan and Mittal \cite{NM14-incompletecitation} give another leaf-oriented
BST implementation, together with a sketch of
correctness.
Instead of marking nodes, it marks edges.
This enables insertions to be accomplished by a single CAS, so 
they do not need to be helped.
It also combines deletions that would otherwise conflict.
All of these trees are not balanced, so the height of a tree with $n$ keys can be $\Theta(n)$.

Tsay and Li \cite{TL94} gave a general approach for implementing trees in a wait-free manner
using LL and SC operations (which can, in turn be implemented from CAS, e.g., \cite{AM95}).  However,
their technique requires every process accessing the tree (even for read-only operations such
as searches) to copy an entire path
of the tree starting from the root.
Concurrency is severely limited, since every operation
must change the root pointer.
Moreover, an extra level of indirection is required
for
every child pointer.

Red-black trees  \cite{Bay72,GS78} are
well known BSTs
that have height $\Theta(\log n)$.
Some
attempts have been made to implement RBTs without using locks.
It was observed that the approach of Tsay and Li could be used to 
implement wait-free RBTs
\cite{NSM13}
and, furthermore, this could be done so
that only updates must copy a path;
searches may simply read the path.
However, the concurrency of updates is still very limited.
Herlihy et al.~\cite{HLMS03:podc} and Fraser and Harris \cite{FH07:tocs} 
experimented on RBTs implemented using
software transactional memory (STM),
which only satisfied obstruction-freedom, a weaker progress property.
Each insertion or deletion, together with necessary rebalancing is enclosed in a single
large transaction, which can touch all nodes on a path from the root to a leaf.

Some researchers have attempted fine-grained approaches to
build non-blocking balanced search trees, but they all
use extremely complicated process coordination schemes.
Spiegel and Reynolds \cite{SR10} described a non-blocking 
data structure that 
combines elements of B-trees and skip lists.
Prior to this paper, it was the leading implementation of an ordered dictionary.
However, the authors provided only
a brief justification of correctness.
Braginsky and Petrank~\cite{BP12} described
a B+tree implementation.
Although they
have posted
a correctness proof, it is very long and
complex.

\medskip

In a balanced search tree, a process is typically responsible
for restoring balance after an insertion or deletion
by performing a series of rebalancing steps along the path
from the root to the location
where the insertion or deletion occurred.
Chromatic trees, introduced by Nurmi and Soisalon-Soininen \cite{NS96},
decouple the updates that perform the insertion or deletion from the updates that perform the rebalancing steps.
Rather than treating an insertion or deletion and its associated rebalancing steps as a single, large
update, it is broken into smaller, localized updates that can be interleaved, allowing more concurrency.
This decoupling originated in the work
of Guibas and Sedgewick \cite{GS78} and Kung and Lehman \cite{KL80}.
We use the leaf-oriented chromatic trees by
Boyar, Fagerberg and Larsen \cite{Boyar97amortizationresults}.
They provide a family of local rebalancing steps which can be executed in any
order, interspersed with insertions and deletions.
Moreover, an amortized \textit{constant} number of rebalancing steps
per \ins\ or \del\ is sufficient to restore balance for any sequence of operations.
We have also used our template to implement a non-blocking version of Larsen's
leaf-oriented relaxed AVL tree \cite{Lar00}.
In such a tree, an amortized \textit{logarithmic} number of rebalancing steps
per \ins\ or \del\ is sufficient to restore balance.

There is also a node-oriented relaxed AVL tree by
Boug\'{e} et~al.~\cite{BGMS98}, in which
an amortized \textit{linear} number of rebalancing steps per \ins\ or \del\ is sufficient to restore balance.
Bronson et~al.~\cite{BCCO10:ppopp} developed a highly optimized fine-grained locking implementation
of this data structure
using optimistic concurrency techniques to improve search performance.
Deletion of a key stored in an internal node with two children is done
by simply marking the node and a later insertion of the same key can
reuse the node by removing the mark.
If all internal nodes are marked, the tree is essentially leaf-oriented.
Crain et~al.
gave a different implementation using lock-based STM
\cite{Crain:2012:SBS:2145816.2145837}
and locks \cite{Crain:2013},
in which \emph{all} deletions are done by marking the node containing the key.
Physical removal of nodes and rotations are performed
by one separate thread. Consequently, the tree can become very unbalanced.
Drachsler et~al.~\cite{DVY14-incompletecitation}
give another fine-grained lock-based implementation,
in which deletion physically removes the node containing the key
and searches are non-blocking.
Each node also contains predecessor and successor pointers,
so when a search ends at an incorrect leaf, sequential search can be performed
to find the correct leaf.
A non-blocking implementation of Boug\'{e}'s tree has not appeared,
but our template would make it easy to produce one.

\section{\llt, \sct\ and \vlt\ Primitives} \label{sec-primitives}

The
load-link extended (\llt), store-conditional extended (\sct) and validate-extended (\vlt)
primitives
are multi-word generalizations of the well-known load-link (LL), store-conditional (SC) and validate (VL) primitives
and they
have been implemented from single-word \cas\ \cite{paper1}.
The benefit of using \llt, \sct\ and \vlt\ to implement our template is two-fold:
the template can be described quite simply, and much of the complexity
of its correctness proof is encapsulated in that of
\llt, \sct\ and \vlt.

Instead of operating on single words, \llt, \sct\ and \vlt\ operate on
\rec s, each of which consists of a fixed number of mutable fields (which can change), and a fixed number of immutable fields (which cannot). 
\llt$(r)$ attempts to take a snapshot of the mutable fields of a \rec\ $r$.
If it is concurrent with an \sct\ involving~$r$, it may return \fail, instead.
Individual fields of a \rec\ can also be read directly.
An \sct$(V,R,fld,new)$ takes as arguments a sequence $V$ of \rec s, a subsequence $R$ of $V$,
a pointer $fld$ to a mutable field of one \rec\ in~$V$, and a new value $new$ for that field.
The \sct\ tries to atomically store the value $new$ in the field that $fld$
points to and {\it finalize} each \rec\ in $R$.
Once a \rec\ is finalized, its mutable fields cannot be changed by any subsequent \sct, 
and any \llt\ of the \rec\ will return \finalized\ instead of a snapshot.

Before a process invokes \sct\ 
or \vlt($V$),
it must perform an \llt$(r)$ on each \rec\ $r$ in $V$.
The last such \llt\ by the process
is said to be {\it linked} to the \sct\ or \vlt, and the linked \llt\ must return a snapshot of $r$ (not \fail\ or \finalized).
An \sct($V, R, fld, new$) by a process modifies the data structure only if each \rec\ $r$ in
$V$ has not been changed since its linked \llt($r$); otherwise the \sct\ fails.
Similarly, a \vlt$(V)$ returns \true\ only if each \rec\ $r$ in $V$ has not been changed since its linked \llt($r$) by the same process; otherwise the \vlt\ fails.
\vlt\ can be used
to obtain a snapshot of a set of \rec s.
Although \llt, \sct\ and \vlt\ can fail, their failures are limited in such a way that we can use them to build non-blocking data structures.
See \cite{paper1} for a more formal specification of these primitives.

These new primitives were designed to balance ease of use and efficient 
implementability using single-word \cas.
The implementation of the primitives from CAS in \cite{paper1} is more efficient 
if the user of the primitives can guarantee that two constraints, which we describe next, are satisfied.
The first constraint prevents the ABA problem for the \cas\ steps that actually perform the updates.

{\bf Constraint 1}:
Each invocation of  \sct$(V, R, fld, new)$ tries to change 
$fld$ to a value
$new$ that it never previously contained.

The implementation of \sct\ does something akin to locking the elements of $V$ in the order they are given. 
Livelock can be easily avoided by requiring all $V$ sequences to be sorted according to some total order on \rec s.
However, this ordering is necessary only to guarantee that \sct s continue to succeed.
Therefore, as long as \sct s are still succeeding in an execution, it does not matter how $V$ sequences are ordered.
This observation leads to the following constraint, which is much weaker.

{\bf Constraint 2}:
Consider each execution that contains
a configuration $C$ after which the value of no
field of any \rec\ changes.
There
is
a total order
of all \rec s created during this execution such that,
for every \sct\ whose linked \llt s begin after $C$,
the $V$ sequence passed to the \sct\ is sorted according to the total order.

It is easy to satisfy these two constraints using standard approaches, e.g., by attaching a version number to each field, and sorting $V$ sequences by any total order, respectively.
However, we shall see that Constraints 1 and 2 are \textit{automatically}
satisfied in a natural way when \llt\ and \sct\ are used according
to our tree update template.

Under these constraints, the implementation of \llt,  \sct, and \vlt\ in \cite{paper1} guarantees that there is a linearization of all \sct s that modify the data structure (which may include \sct s that do not terminate because a process crashed, but \textit{not} any \sct s that fail), and all \llt s and \vlt s that return,
but do not fail.

We assume there is a \rec\ $entry$ which acts as the entry point to the data structure and is never deleted.
This \rec\ points to the root of a down-tree.
We represent an empty down-tree by a pointer to an empty \rec.
A \rec\ is {\em in the tree} if it can be reached by following pointers
from $entry$.
A \rec\ $r$ is {\em removed from the tree} by an \sct\ if $r$
is in the tree immediately prior to the linearization point of the \sct\ and
is not in the tree immediately afterwards.
Data structures produced using our template \textit{automatically} satisfy one additional constraint:
\after{note: there is something strange with the definition of ``removed from the tree'' if multiple \sct s can be linearized at the same time.}

{\bf Constraint \conmarkallremovedrecs}: A \rec\ is finalized when (and only when) 
it is removed from the tree.

\noindent
Under this additional constraint, the implementation of \llt\ and \sct\ in \cite{paper1}
also guarantees the following three properties.
\begin{compactitem}
\item
If \llt$(r)$ returns a snapshot, 
then $r$ is in the tree
just before the \llt\ is linearized.
\item
If an \sct$(V,R,fld,new)$ is linearized and $new$ is (a pointer to) a \rec, then this \rec\ 
is in the tree
immediately after
the \sct\ is linearized.
\item
If an operation reaches a \rec\ $r$ by following pointers read from other \rec s, starting from $entry$, then $r$ was in the tree at some
earlier time during the operation.
\end{compactitem}
\noindent These properties are useful for proving the correctness of our template.
In the following, we sometimes abuse notation by treating the
sequences $V$ and $R$ as sets, in which case
we mean the set of all \rec s in the sequence.

The memory overhead introduced by the implementation of \llt\ and \sct\ is fairly low.
Each node in the tree is augmented with a pointer to a descriptor and a bit.
Every node that has had one of its child pointers changed by an \sct\ points to a descriptor.
(Other nodes have a \nil\ pointer.)
A descriptor can be implemented to use only three machine words after the update it describes has finished.
The implementation of \llt\ and \sct\ in \cite{paper1} assumes garbage collection, and we do the same in this work.
This assumption can be eliminated by using, for example, the new efficient memory reclamation scheme of Aghazadeh et~al. \cite{AGW13}.

\section{Tree Update Template} \label{sec-dotreeupdate}

Our tree update template implements updates that
atomically replace an old connected subgraph in a down-tree by 
a new connected subgraph.
Such an 
update can implement any change to
the tree, such as an insertion into a BST or a 
rotation used to rebalance a RBT.
The old subgraph includes all nodes
with a field (including a child pointer) to be modified.
The new subgraph
may have
pointers to
nodes in the old tree.
Since every node in a down-tree has indegree one, the update can be
performed by changing
a single child pointer of some node $parent$.  (See Figure~\ref{fig-replace-subtree}.)
However, problems could arise if a concurrent operation changes 
the part of the tree being updated. 
For example, nodes in the old subgraph, or even $parent$, could be
removed from the tree before $parent$'s child pointer is changed.
Our template takes care of the process coordination required to
prevent such problems.
\begin{figure}[tb]
	\centering
	\input{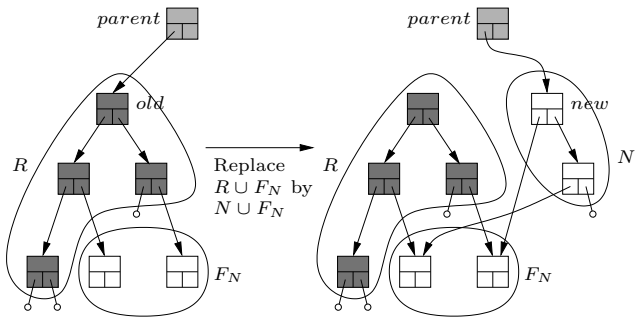}
	\caption{Example of the tree update template.
			$R$ is the set of nodes to be removed, 
			$N$ is a tree of new nodes that have never before appeared in the tree, and
			$F_N$ is the set of children of $N$ (and of $R$). 
			Nodes in $F_N$ may have children.  
			The shaded nodes (and possibly others) are in the sequence $V$  of the \sct\ that performs the update.
			The darkly shaded nodes are finalized by the \sct.
			}
	\label{fig-replace-subtree}
\end{figure}

\begin{figure}[tb]
	\centering
	\input{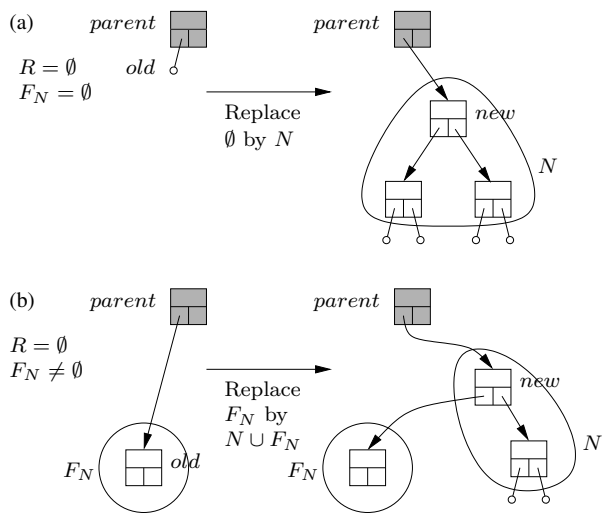}
	\caption{Examples of two special cases of the tree update template when no nodes are removed from the tree.  
	(a) Replacing a \nil\ child pointer:  In this case, $R=F_N=\emptyset$.  
	(b) Inserting new nodes in the middle of the tree:  In this case, $R=\emptyset$ and $F_N$ consists of a single node.}
	\label{fig-replace-subtree2}
\end{figure}

Each tree node is represented by a \rec\ with a fixed number of child pointers 
as its mutable fields (but different nodes may have different numbers of child fields).  
Each child pointer points to a \rec\ or contains \nil\ (denoted by $\multimap$ in our figures).
For simplicity, we assume that any other data in the node is stored 
in immutable fields. 
Thus, if an update must change some of this data, it makes a new copy of the node with
the updated data.

\begin{figure}[tbp]
\def\pwidth{4cm}
\prepnewlisting
\vspace{-2mm}
\begin{lstlisting}[mathescape=true]
    //\func{Template}$(args)$
      //follow zero or more pointers from $entry$ to reach a node $n_0$
      $i := 0$
      loop
        $s_i := \llt(n_i)$
        if $s_i \in \{\fail, \finalized\}$ then return $\fail$
        $s_i' :=$// immutable fields of $n_i$
        exit loop when $\func{Condition}(s_0, s_0', \ldots, s_i, s_i', args)$ // \hfill $\\ \com \func{Condition}$ must eventually return $\true$
        $n_{i+1} := \func{NextNode}(s_0, s_0', \ldots, s_i, s_i', args)$  // \hfill $\\ \com$ returns a non-\nil\ child pointer from one of $s_0, \ldots, s_i$
        $i := i+1$
      end loop
      if $\sct(\func{\sct-Arguments}(s_0, s_0', \ldots, s_i, s_i', args))$ then return $\func{Result}(s_0, s_0', \ldots, s_i, s_i', args)$
      else return $\fail$
\end{lstlisting}
	\caption{Tree update template. \func{Condition}, \func{NextNode}, \func{SCX-Arguments} and \func{Result} can be filled in with any locally computable functions, provided that \func{SCX-Arguments} satisfies postconditions PC1 to PC\ref{con-V-sequences-ordered-consistently}.}
	\label{code-dotreeupdate}
\end{figure}

Our template for performing an update to the tree 
is fairly simple:
An update first performs \llt s on nodes in a contiguous portion of the tree,
including $parent$ and the set $R$ of nodes to be removed from the tree.  
Then, it performs an \sct\ that atomically
changes the child pointer 
as shown in Figure~\ref{fig-replace-subtree}
and finalizes nodes in $R$.  Figure~\ref{fig-replace-subtree2} shows two special cases where $R$ is empty.
An update that performs this sequence of steps is said to
\textit{follow} the template.

We now describe the tree update template in more detail.
An update \func{UP}$(args)$ that follows the template shown in Figure~\ref{code-dotreeupdate} takes any arguments, $args$, that are needed to perform the update.
\func{UP} first reads a sequence of child pointers starting from $entry$ to reach some node $n_0$.
Then, \func{UP} performs \llt s on a sequence $\sigma = \langle n_0, n_1, \ldots\rangle$ of nodes starting with $n_0$.
%
For maximal flexibility of the template, 
the sequence  $\sigma$ can be constructed on-the-fly, as \llt s are performed.
Thus, \func{UP} chooses a  non-\nil\ child of one of the previous nodes to be the next node of $\sigma$
by performing some deterministic 
local computation (denoted by \func{NextNode} in Figure~\ref{code-dotreeupdate}) using any information
that is available locally, namely, the snapshots of mutable fields returned by \llt s on the previous elements of $\sigma$, values read from immutable fields of previous elements of $\sigma$, and $args$.
(This flexibility can be used, for example, to avoid unnecessary \llt s when deciding how to rebalance a BST.) 
\func{UP} performs another local computation (denoted by \func{Condition} in Figure~\ref{code-dotreeupdate}) 
to decide whether  more \llt s should be performed.
To avoid infinite loops, this function must eventually return \true\ in any execution of \func{UP}.
If any 
\llt\ in the sequence
returns \fail\ or \finalized, \func{UP} also
returns \fail, to indicate
that the attempted update has been aborted because of a concurrent
update on an overlapping portion of the tree.
If all of the \llt s successfully return snapshots, \func{UP} invokes \sct\
and returns a result calculated
locally by the \func{Result} function (or \fail\ if the \sct\ fails).

\func{UP} applies the function \func{\sct-Arguments} to use locally available information to construct the arguments $V$, $R$, $fld$ and $new$ for the \sct.
The postconditions that 
must be satisfied by \func{\sct-Arguments}
are somewhat technical, but intuitively, they are
meant to ensure that the arguments produced describe an update
as shown in Figure~\ref{fig-replace-subtree} 
or 
Figure~\ref{fig-replace-subtree2}.
The update must 
remove
a connected set $R$ of nodes from the tree and replace it
by a connected set $N$ of newly-created nodes that is rooted at $new$ by
changing the
child pointer stored in $fld$
to point to $new$.
%
In order for this change to occur atomically, we include $R$ and the
node containing $fld$
in $V$.  This ensures
that if any of these nodes has changed since it was last accessed by
one of \func{UP}'s \llt s,
the \sct\ will fail.
The sequence $V$ may also include any other nodes in $\sigma$. 

More formally, we require \func{SCX-Arguments} to satisfy nine postconditions.
The first three are basic requirements of \sct.
\begin{compactenum}[\hspace{3.4mm}\bfseries PC1:]
\item $V$ is a subsequence of $\sigma$.
\label{con-llt-on-all-nodes-in-V}
\item $R$ is a subsequence of $V$.
\label{con-R-subsequence-of-V}
\item The node $parent$ containing the mutable field $fld$ is in $V$.%
\label{con-parent-in-V}%
\end{compactenum}

\noindent
Let $G_N$
be the directed graph $(N \cup F_N,E_N)$,
where $E_N$ is the set of all child pointers of nodes in $N$
when they are initialized, and
$F_N = \{ y : y\not\in N \mbox{ and }(x,y) \in E_N$ for some $x\in N$\}.
Let $old$ be the value read from $fld$ by the \llt\ on $parent$.
\begin{compactenum}[\hspace{4.1mm}\bfseries PC1:]
\setcounter{enumi}{3}
\item 
$G_{N}$ is a non-empty down-tree rooted at $new$.
\label{con-GN-non-empty-tree}
\item 
If $old=\nil$ then $R=\emptyset$ and $F_N = \emptyset$. 
\label{con-old-nil-then-R-empty}
\item 
If $R = \emptyset$ and $old \neq \nil$, then $F_N = \{ old \}$.
\label{con-fringe-of-GN-is-old}
\item
\func{UP} allocates memory for all nodes in $N,$ including $new$.%
\label{con-new-nodes}%
\end{compactenum}

\noindent
Postcondition PC\ref{con-new-nodes} requires $new$ to be a newly-created node, in order to satisfy Constraint 1.
There is no loss of generality in using this approach:
If we wish to change a child $y$ of node $x$ to \nil\ (to chop off the entire subtree rooted at $y$)
or to a descendant of $y$ (to splice out a portion of the tree),
then, instead, we can replace $x$ by a new copy
of $x$ with an updated child pointer.
Likewise, if we want to delete the entire tree, then $entry$ can be changed to point to a new, empty \rec.

The next postcondition is used to guarantee Constraint 2, which is used to prove progress.
\begin{compactenum}[\hspace{4.1mm}\bfseries PC1:]
\setcounter{enumi}{7}
\item 
The sequences $V$ constructed by all updates 
that take place entirely 
during a period of time when no \sct s change the tree structure
must be ordered consistently according to a fixed tree traversal algorithm (for example, an in-order traversal or a breadth-first traversal).
\label{con-V-sequences-ordered-consistently}
\end{compactenum}

Stating the remaining postcondition formally requires some care, since the tree may be changing while \func{UP} performs its \llt s.
If $R \neq \emptyset$, let $G_R$ be the directed graph $(R \cup F_R,E_R)$,
where $E_R$ is the union of the sets of edges representing child pointers read from each $r \in R$
when it was last accessed by one of \func{UP}'s \llt s and
$F_R = \{ y : y\not\in R \mbox{ and }(x,y) \in E_R$ for some $x\in R$\}.
$G_R$ represents \func{UP}'s view of the nodes in $R$ according to its \llt s, and $F_R$ is the \textit{fringe} of $G_R$.
If other processes do not change the tree while \func{UP} is being performed,
then $F_R$ contains the nodes
that should remain in the tree, but
whose parents will be removed and replaced.
Therefore, we must ensure that the nodes in $F_R$ are reachable from nodes in $N$ (so they are not accidentally removed from the tree).
Let $G_\sigma$ be the directed graph $(\sigma \cup F_\sigma,E_\sigma)$,
where $E_\sigma$ is the union of the sets of edges representing child pointers read from each $r \in \sigma$ when it was last accessed by one of \func{UP}'s \llt s and $F_\sigma = \{ y : y\not\in \sigma \mbox{ and }(x,y) \in E_\sigma$ for some $x\in \sigma$\}.
Since $G_\sigma$, $G_R$ and $G_N$ are not affected by concurrent updates, the following postcondition can be proved using purely sequential reasoning, ignoring the possibility that concurrent updates could modify the tree during \func{UP}.
\begin{compactenum}[\hspace{4.1mm}\bfseries PC1:]
\setcounter{enumi}{8}
\item 
If $G_\sigma$ is a down-tree and $R \neq \emptyset$,
then $G_{R}$ is a non-empty down-tree rooted at $old$
and $F_N = F_R$.
\label{con-R-non-empty-then-GR-a-non-empty-tree}
\end{compactenum}%

\subsection{Correctness and Progress}

For brevity, we only sketch the main ideas of the proof here.
The full proof appears in Appendix~\ref{app-tree-proof}.
%
Consider a data structure in which all updates 
follow the tree update template and \func{\sct-Arguments} satisfies postconditions PC1 
to PC\ref{con-R-non-empty-then-GR-a-non-empty-tree}.
We prove, by induction on the sequence of steps in an execution,
that the data structure is always a tree, each call to \llt\ and \sct\ satisfies its preconditions, Constraints 1 to 3 are satisfied, and each successful \sct\ 
atomically replaces a connected subgraph containing nodes $R \cup F_N$
with another connected subgraph containing nodes $N \cup F_N,$
finalizing and removing the nodes in $R$ from the tree, and
adding the new nodes in $N$ to the tree.
We also prove no node in the tree is finalized, every removed node is finalized, and removed nodes are never reinserted. 

We linearize each update UP that follows the template and performs an \sct\ that modifies the data structure at the linearization point of its \sct.
We prove the following correctness properties. 
\begin{compactenum}[\hspace{4.1mm}\bfseries C1:]
\setcounter{enumi}{0}
\item If UP were performed atomically at its linearization point, then it would perform \llt s on the same nodes, and these \llt s would return the same values.
\label{prop-corr-lin-same-llt}%
\end{compactenum}%
This implies that UP's \func{SCX-Arguments} and \func{Result} computations must be the same as they would be if UP were performed atomically at its linearization point, so we obtain the following.
\begin{compactenum}[\hspace{4.1mm}\bfseries C1:]
\setcounter{enumi}{1}
\item If UP were performed atomically at its linearization point, then it would perform the same \sct\ (with the same arguments) and return the same value. 
\label{prop-corr-lin-same-sct}%
\end{compactenum}%
Additionally, a property is proved in \cite{paper1} that allows some query operations to be performed very efficiently using only \func{read}s,
for example, \func{Get} in Section \ref{sec-chromatic}.
\begin{compactenum}[\hspace{4.1mm}\bfseries C1:]
\setcounter{enumi}{2}
\item If a process $p$ follows child pointers starting from a node in the tree at time $t$ and reaches a node $r$ at time $t' \geq t$,
then $r$ {\it was} in the tree at some time between $t$ and $t'$.
Furthermore, if $p$ reads $v$ from a mutable field of $r$ at time
$t'' \ge t'$ then, at some time between $t$ and $t''$, node $r$ was in the tree and this field contained $v$.
\label{prop-corr-queries}%
\end{compactenum}%

The following properties, which come from \cite{paper1}, can be used to prove non-blocking progress of queries.

\begin{compactenum}[\hspace{4.1mm}\bfseries P1:]
\setcounter{enumi}{0}
\item If \llt s are performed infinitely often, then they 
return snapshots or \finalized\ infinitely often.
\label{prop-progress-llx}%
\item If \vlt s are performed infinitely often, and \sct s are not performed infinitely often, then \vlt s  return \true\  infinitely often.
\label{prop-progress-vlx}%
\end{compactenum}%

\noindent Each update that follows the template is wait-free.
Since updates can fail, we also prove the following progress property. %

\begin{compactenum}[\hspace{4.1mm}\bfseries P1:]
\setcounter{enumi}{2}
\item If updates that follow the template are performed infinitely often, then updates succeed infinitely often.
\label{prop-progress}%
\end{compactenum}
\noindent
A successful update performs an \sct\ that modifies the tree.
Thus, it is necessary to show that \sct s succeed infinitely often.
Before an invocation of \sct$(V, R, fld, new)$ can succeed, it must perform an \llt$(r)$ that returns a snapshot, for each $r \in V$.
Even if P\ref{prop-progress-llx} is satisfied, it is possible for \llt s to always return \finalized, preventing any \sct s from being performed.
We prove that any algorithm whose updates follow the template automatically guarantees that, for each \rec\ $r$, each process performs at most one invocation of \llt$(r)$ that returns \finalized.
We use this fact to prove P\ref{prop-progress}.

\section{Application: 
Chromatic Trees} \label{sec-chromatic}

Here, we show how the tree update template can be used to implement an ordered dictionary ADT using
chromatic trees.
Due to space restrictions, we only sketch the algorithm and its correctness proof.
All details of the implementation and its correctness proof appear in Appendix~\ref{chromatic}. 
The ordered dictionary stores a set of keys, each with an associated value, 
where the keys are drawn from a totally ordered universe.
The dictionary supports five operations.
If $key$ is in the dictionary, \func{Get}$(key)$ returns its associated value.
Otherwise, \func{Get}$(key)$ returns $\bot$.
\func{Successor}$(key)$ returns the smallest key in the dictionary that is larger than $key$ (and its associated value), or $\bot$ if no key in the dictionary is larger than $key$.
\func{Predecessor}$(key)$ is analogous.
\ins$(key, value)$ replaces the value associated with $key$ by $value$ and returns the previously associated value, or $\bot$ if $key$ was not in the dictionary.
If the dictionary contains $key$, \del$(key)$ removes it and returns the value that was associated immediately beforehand.
Otherwise, \del($key$) simply returns $\bot$.

A RBT is a BST in which the root and all leaves are coloured black, and every other node is coloured either red or black, subject to the constraints that no red node has a red parent, and the number of black nodes on a path from the root to a leaf is the same for all leaves.
These properties guarantee that the height of a RBT is logarithmic in the number of nodes it contains.
We consider search trees that are leaf-oriented, meaning the dictionary keys are stored in the leaves, and internal nodes store keys that are used only to direct searches towards the correct leaf.
In this context, the BST property says that, for each node $x$, all descendants of $x$'s left child have keys less than $x$'s key and all descendants of $x$'s right child have keys that are greater than {\it or equal to} $x$'s key.

\begin{figure}[tb]
\centering
\raisebox{10.9mm}{
\includegraphics[scale=.7]{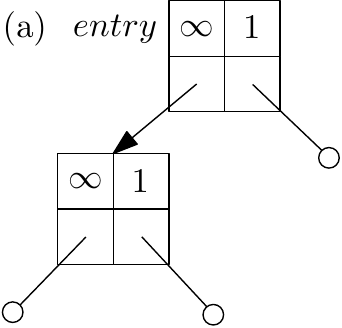}}\hspace*{3mm}
\includegraphics[scale=.7]{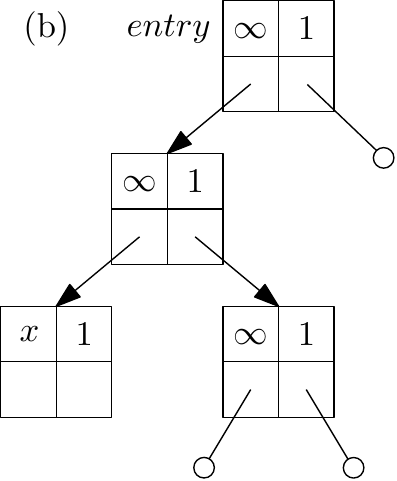}
\caption{(a) empty tree, (b) non-empty tree.}
\label{fig-treetop}
\vspace{-2mm}
\end{figure}

To decouple rebalancing steps from insertions and deletions, so that each is localized, and rebalancing steps can be interleaved with insertions and deletions, it is necessary to relax the balance properties of RBTs.
A {\it chromatic tree} \cite{NS96} is a relaxed-balance RBT in which colours are replaced by non-negative integer weights, where weight zero corresponds to red and weight one corresponds to black.  As in RBTs, 
the sum of the weights on each path from the root to a leaf is the same.  However, 
RBT properties can be violated in the following two ways.
First, a red child node may have a red parent, in which case we say that a
\textit{red-red violation} occurs at this child.
Second, a node may have weight $w>1$, in which case we say that
$w-1$ \textit{overweight violations} occur at this node.
The root always has weight one,
so no violation can occur at the root.

To avoid special cases when the chromatic tree is empty, we add sentinel nodes 
at the top of the tree (see Figure \ref{fig-treetop}).
The sentinel nodes and $entry$ have key $\infty$ to avoid special cases for \search, \ins\ and \del, and weight one to avoid special cases for rebalancing steps.
Without having a special case for \ins, we automatically get the two sentinel nodes in Figure~\ref{fig-treetop}(b), which also eliminate special cases for \del.
The chromatic tree is rooted at the leftmost grandchild of $entry$.
The sum of weights is the same for all paths from the root of the chromatic tree to its leaves, but not for paths that include $entry$ or the sentinel nodes.

{\it Rebalancing steps} are localized updates to a chromatic tree that are performed at the location of a violation.
Their goal is to eventually eliminate all red-red and overweight violations, while maintaining the invariant that the tree is a chromatic tree.
If no rebalancing step can be applied to a chromatic tree (or, equivalently, the 
chromatic tree contains no violations), then it is a RBT.
We use the set of rebalancing steps of Boyar, Fagerberg and 
Larsen~\cite{Boyar97amortizationresults},
which have a number of desirable properties:
No rebalancing step increases the number of violations in the tree,
rebalancing steps can be performed in any order, and,
after sufficiently many rebalancing steps, the tree will always become a RBT.
Furthermore, in any sequence of insertions, deletions and rebalancing steps starting from an empty chromatic tree, the amortized number of rebalancing steps
is at most three per insertion and one per deletion.

\subsection{Implementation}

We represent each node by a \rec\ with two mutable child pointers, and immutable fields $k$, $v$ and $w$ that contain the node's key, associated value, and weight, respectively.
The child pointers of a leaf are always \nil, and the value field of an internal node is always \nil.

\func{Get}, \ins\ and \del\ each execute an auxiliary procedure, \func{Search}($key$), which starts at $entry$ and traverses nodes as in an ordinary BST
search, using \func{Read}s of child pointers until reaching a leaf,
which it then returns (along with the leaf's parent and grandparent).
Because of the sentinel nodes shown in Figure~\ref{fig-treetop}, the leaf's parent always exists, and the grandparent exists whenever the chromatic tree is non-empty.
If it is empty, \func{Search} returns \nil\ instead of the grandparent.
We define the \textit{search path} for $key$ at any time to be the path that \func{Search}($key$) would follow, if it were done instantaneously.
The \func{Get}($key$) operation simply executes a \func{Search}($key$) and then returns
the value found in the leaf if the leaf's key is $key$, or $\bot$ otherwise.

\begin{figure}
\prepnewlisting
\vspace{-5mm}
\hrule
\vspace{-2mm}
\begin{lstlisting}[mathescape=true]
    //\func{Get}$(key)$
      $\langle -, -, l \rangle := \func{Search}(key)$
      return $(key = l.k)\ ?\ l.v : \nil$// \\ \vspace{-2mm} \hrule \vspace{1mm} %
      
    //\func{Search}$(key)$ %
      %//\com Returns the parent and grandparent of the leaf found by doing a BST search for $key$
      
      $n_0 := \nil; n_1 := entry; n_2 := entry.\mbox{\textit{left}}$
      while $n_2$// is internal
        $n_0 := n_1; n_1 := n_2$
        $n_2 := (key < n_1.k)\ ?\ n_1.\mbox{\textit{left}} : n_1.\mbox{\textit{right}}$
      return $\langle n_0, n_1, n_2 \rangle$// \\ \vspace{-2mm} \hrule \vspace{1mm} %
    
    //\del$(key)$ %
      %//\com Deletes $key$ and returns its associated value, or returns $\bot$ if $key$ was not in the dictionary
      
      do
        $result := \trydel(key)$
      while $result = \fail$
      $\langle value, violation \rangle := result$
      if $violation$ then $\cleanup(key)$
      return $value$// \\ \vspace{-2mm} \hrule \vspace{1mm} %
    
    //\cleanup$(key)$
      //\com Eliminates the violation created by an \ins\ or \del\ of $key$
      while $\true$
        //\com Save four last nodes traversed
        $n_0 := \nil; n_1:=\nil; n_2:=entry; n_3 := entry.\mbox{\textit{left}}$ //\label{cleanup-start}
        while //\true 
          if $n_3.w > 1$ or ($n_2.w = 0$ and $n_3.w = 0$) then
            //\com Found a violation at node $n_3$ \label{find-violation}
            $\tryrebalance(n_0, n_1, n_2, n_3)$ //\hfill \com Try to fix it
            exit loop //\hfill \com Go back to $entry$ and search again\label{cleanup-end}
          else if $n_3\mbox{ is a leaf}$ then return//\label{cleanup-terminate}
            //\com Arrived at a leaf without finding a violation %\vspace{1.5mm}%
          
          if $key < n_3.k$ then
            $n_0 := n_1; n_1 := n_2; n_2 := n_3; n_3:=n_3.\mbox{\textit{left}}$ //\label{move-l-left}
          else $n_0 := n_1; n_1 := n_2; n_2 := n_3; n_3:=n_3.\mbox{\textit{right}}$ //\label{move-l-right} \vspace{-2mm}
\end{lstlisting}
	\caption{\func{Get}, \func{Search}, \func{Delete} and \func{Cleanup}.}
	\label{code-chromatic1}
\end{figure}

\begin{figure}
\prepnewlisting
\vspace{-5mm}
\hrule
\vspace{-2mm}
\begin{lstlisting}[mathescape=true]
    //\trydel$(key)$
      //\com If successful, returns $\langle value, violation \rangle$, where $value$ is the value associated with $key$, or $\nil$ if $key$ was not in the dictionary, and $violation$ indicates whether the deletion created a violation.  Otherwise, \fail\ is returned.%
      
      $\langle n_0, -, - \rangle := \func{Search}(key)$//\label{del-search-line} %\vspace{1.5mm}%
      
      //\com Special case: there is no grandparent of the leaf reached
      if $n_0 = \nil$ then return $\langle \nil, \false \rangle$//\label{del-no-gp}%\vspace{1.5mm}%
      
      //\com Template iteration 0 (grandparent of leaf)
      $s_0 := \llt(n_0)$
      if $s_0 \in \{\fail, \finalized\}$ then return $\fail$
      $n_1 := (key < s_0.\mbox{\textit{left}}.k)\ ?\ s_0.\mbox{\textit{left}} : s_0.\mbox{\textit{right}}$//%\vspace{1.5mm}%

      //\com Template iteration 1 (parent of leaf)
      $s_1 := \llt(n_1)$
      if $s_1 \in \{\fail, \finalized\}$ then return $\fail$
      $n_2 := (key < s_1.\mbox{\textit{left}}.k)\ ?\ s_1.\mbox{\textit{left}} : s_1.\mbox{\textit{right}}$//%\vspace{1.5mm}%
      
      //\com Special case: $key$ is not in the dictionary
      if $n_2.k \neq key$ then return $\langle \bot, \false \rangle$ //\label{del-notin}%\vspace{1.5mm}%

      //\com Template iteration 2 (leaf)
      $s_2 := \llt(n_2)$
      if $s_2 \in \{\fail, \finalized\}$ then return $\fail$
      $n_3 := (key < s_1.\mbox{\textit{left}}.k)\ ?\ s_1.\mbox{\textit{right}} : s_1.\mbox{\textit{left}}$//\label{del-getsibling}%\vspace{2mm}%
      
      //\com Template iteration 3 (sibling of leaf)
      $s_3 := \llt(n_3)$
      if $s_3 \in \{\fail, \finalized\}$ then return $\fail$//%\vspace{2mm}%
    
      //\com Computing \func{\sct-Arguments} from locally stored values
      $w := (n_1.k = \infty$ or $n_0.k = \infty)\ ?\ 1 : n_1.w + n_3.w$//\label{del-weight}
      $new :=$// new node with weight $w$, key $n_3.k$, value $n_3.v$, and children $s_3.\mbox{\textit{left}}, s_3.\mbox{\textit{right}}$\label{del-create-new}
      $V := (key < s_1.\mbox{\textit{left}}.k)\ ?\ \langle n_0, n_1, n_2, n_3 \rangle : \langle n_0, n_1, n_3, n_2 \rangle$
      $R := (key < s_1.\mbox{\textit{left}}.k)\ ?\ \langle n_1, n_2, n_3 \rangle : \langle n_1, n_3, n_2 \rangle$
      $fld := (key < s_0.\mbox{\textit{left}}.k)\ ?\ \&n_0.\mbox{\textit{left}} : \&n_0.\mbox{\textit{right}}$//%\vspace{1.5mm}%
      
      if $\sct(V, R, fld, new)$ then return $\langle n_2.v, (w > 1) \rangle$
      else return $\fail$// \vspace{-2mm}
\end{lstlisting}
	\caption{\func{TryDelete}.}
	\label{code-chromatic2}
\end{figure}

At a high level, \ins\ and \del\ are quite similar to each other.
\ins($key$, $value$) and \del($key$) each perform \func{Search}($key$) 
and then make the required update at the leaf reached,
in accordance with the tree update template.
If the modification fails, then the operation restarts from scratch. 
If it succeeds, it may increase the number of violations in the tree by one, 
and the new violation occurs on the search path to $key$.
If a new violation is created, then an auxiliary procedure \cleanup\ is invoked to fix it before the \ins\ or \del\ returns.

Detailed pseudocode for \func{Get}, \func{Search}, \del\ and \func{Cleanup} is given in Figure~\ref{code-chromatic1} and~\ref{code-chromatic2}.
(The implementation of \ins\ is similar to that of \del, and its pseudocode is omitted due to lack of space.)
Note that an expression of the form $P\ ?\ A : B$ evaluates to $A$ if the predicate $P$ evaluates to true, and $B$ otherwise.
The expression $x.y$, where $x$ is a \rec, denotes field $y$ of $x$, and the expression $\&x.y$ represents a pointer to field $y$.

\del($key$) invokes \trydel\ to search for a leaf containing $key$ and perform the localized update that actually deletes $key$ and its associated value.
The effect of \trydel\ is illustrated in Figure~\ref{fig-dotreeupdate-dorb2}.
There, nodes drawn as squares are leaves, shaded nodes are in $V$, $\otimes$ denotes a node in $R$ to be finalized, and $\oplus$ denotes a new node.
The name of a node appears below it or to its left.
The weight of a node appears to its right.

\trydel\ first invokes \func{Search}$(key)$ to find the grandparent, $n_0$, of the leaf on the search path to $key$.
If the grandparent does not exist, then the tree is empty (and it looks like Figure~\ref{fig-treetop}(a)), so \trydel\ returns successfully at line~\ref{del-no-gp}. 
\trydel\ then performs \llt$(n_0)$, and uses the result to obtain a pointer to the parent, $n_1$, of the leaf to be deleted.
Next, it performs \llt$(n_1)$, and uses the result to obtain a pointer to the leaf, $n_2$, to be deleted.
If $n_2$ does not contain $key$, then the tree does not contain $key$, and \trydel\ returns successfully at line~\ref{del-notin}.
So, suppose that $n_2$ does contain $key$.
Then \trydel\ performs \llt$(n_2)$.
At line~\ref{del-getsibling}, \trydel\ uses the result of its previous \llt$(n_1)$ to obtain a pointer to the sibling, $n_3$, of the leaf to be deleted.
A final \llt\ is then performed on $n_3$.
Over the next few lines, \trydel\ computes \func{\sct-Arguments}.
Line~\ref{del-weight} computes the weight of the node $new$ in the depiction of \del\ in Figure~\ref{fig-treetop}, ensuring that it has weight one if it is taking the place of a sentinel or the root of the chromatic tree.
Line~\ref{del-create-new} creates $new$, reading the key, and value directly from $n_3$ (since they are immutable) and the child pointers from the result of the \llt$(n_3)$ (since they are mutable).
Next, \trydel\ uses locally stored values to construct the sequences $V$ and $R$ that it will use for its \sct, ordering their elements according to a breadth-first traversal, in order to satisfy PC\ref{con-V-sequences-ordered-consistently}.
Finally, \trydel\ invokes \sct\ to perform the modification.
If the \sct\ succeeds, then \trydel\ returns a pair containing the value stored in node $n_2$ (which is immutable) and the result of evaluating the expression $w > 1$.

\del\ can create an overweight violation (but not a red-red violation), so the result of $w > 1$ indicates whether \trydel\ created a violation.
If any \llt\ returns \fail\ or \finalized, or the \sct\ fails, \trydel\ simply returns \fail, and \del\ invokes \trydel\ again.
If \trydel\ creates a new violation, then \del\ invokes \cleanup$(key)$ (described in Section~\ref{sec-chromatic-rebalancing-alg}) to fix it before \del\ returns.

\after{The template says that we read immutable fields of a \rec\ after  we read its mutable fields, but here we are reading them in the opposite order.
Let's try to fix this for the camera ready copy.}

A simple inspection of the pseudocode suffices to prove that \func{\sct-Arguments} satisfies postconditions PC1 to PC\ref{con-R-non-empty-then-GR-a-non-empty-tree}.
\trydel\ follows the template except when it returns at line~\ref{del-no-gp} or line~\ref{del-notin}.
In these cases, not following the template does not impede our efforts to prove correctness or progress, since \trydel\ will not modify the data structure, and returning at either of these lines will cause \del\ to terminate.

We now describe how rebalancing steps are implemented from \llt\ and \sct, using the tree update template.
As an example, we consider one of the 22 rebalancing steps, named \func{RB2} (shown in Figure~\ref{fig-dotreeupdate-dorb2}), which eliminates a red-red violation at node $n_3$.
The other 21 are implemented similarly.
To implement \func{RB2}, a sequence of \llt s are performed,
starting with node $n_0$. A pointer to node $n_1$ is obtained from
the result of \llt($n_0$), pointers to nodes $n_2$ and $f_3$ are
obtained from the result of \llt($n_1$), and a pointer to node $n_3$
is obtained from the result of \llt($n_2$).
Since node $n_3$ is to be removed from the tree, an \llt\ is performed
on it, too.
If any of these \llt s returns \fail\ or \finalized, then 
this update fails.
For \func{RB2}
to be applicable, $n_1$ and $f_3$ must have positive weights
and $n_2$ and $n_3$ must both have weight 0.
Since weight fields are immutable, they can be read
any time after the pointers to $n_1$, $f_3$, $n_2$, and $n_3$
have been obtained.
Next, $new$ and its two children are created.
$N$ consists of these three nodes.
Finally, \sct$(V, R, fld, new)$ is invoked,
where $fld$ is the child pointer of $n_0$
that pointed to $n_1$ in the result of \llt($n_0$).

\begin{figure}[tb]
\vspace{-1mm}
\centering
\includegraphics[scale=1]{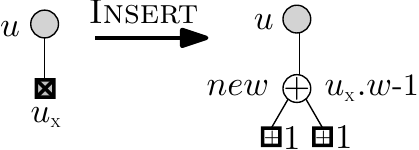}

\vspace{2mm}
\includegraphics[scale=1]{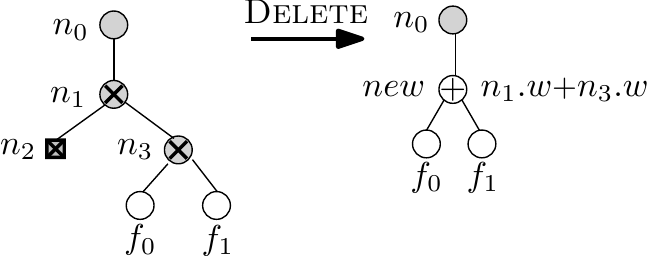}

\vspace{2mm}
\includegraphics[scale=1]{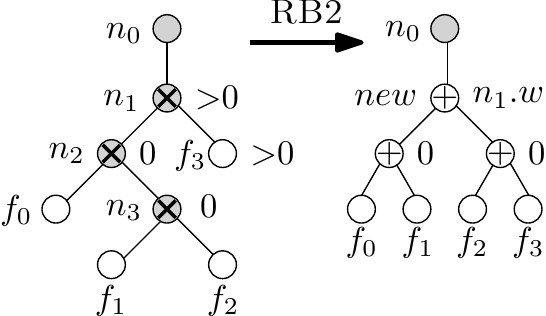}
	\caption{Examples of chromatic tree updates.
}
	\label{fig-dotreeupdate-dorb2}
\vspace{-2mm}
\end{figure}

If the \sct\ modifies the tree, then no node $r \in V$ has changed since the
update performed \llt$(r)$.
In this case, the \sct\ replaces the directed graph $G_R$
by the directed graph $G_N$
and the nodes in $R$ are finalized. This ensures that
other updates cannot erroneously modify these old nodes after
they have been replaced.
The nodes in the set $F_R = F_N = \{f_0, f_1, f_2, f_3\}$
each have the same keys, weights, and child pointers before and after the
rebalancing step, so they can be reused.
$V = \langle n_0, n_1, n_2, n_3 \rangle$ is
simply the sequence of nodes 
on which \llt\ is performed, and $R = \langle n_1, n_2, n_3 \rangle$ is a subsequence 
of $V$, so PC\ref{con-llt-on-all-nodes-in-V}, PC\ref{con-R-subsequence-of-V} and~PC\ref{con-parent-in-V} are satisfied.
Clearly, we satisfy PC\ref{con-GN-non-empty-tree} and~PC\ref{con-new-nodes} when we create $new$ and its two children.
It is easy to verify that PC\ref{con-old-nil-then-R-empty}, PC\ref{con-fringe-of-GN-is-old} and~PC\ref{con-R-non-empty-then-GR-a-non-empty-tree} are satisfied.
If the tree does not change during the update, then the nodes in
$V$ are ordered consistently with a breadth-first traversal of the tree.
Since this is true for all updates,
PC\ref{con-V-sequences-ordered-consistently} is satisfied.

One might wonder why $f_3$ is not 
in $V$, since RB2 should be applied only if $n_1$ has a right child with positive weight.
Since weight fields are immutable, the only way that 
this can change after we check $f_3.w > 0$ is if
the right child field of $n_1$ is altered.
If this happens, the \sct\ will fail.

\subsection{The rebalancing algorithm} \label{sec-chromatic-rebalancing-alg}

Since rebalancing is decoupled from updating,
there must be
a scheme that determines 
when processes should perform rebalancing steps to eliminate violations.
In \cite{Boyar97amortizationresults}, the authors suggest maintaining one or more \textit{problem queues} which contain pointers to nodes that contain violations, and dedicating one or more \textit{rebalancing processes} to simply perform rebalancing steps as quickly as possible.
This approach does not yield a bound on the height of the tree, since rebalancing may lag behind insertions and deletions.
It is possible to obtain a height bound with a different queue based scheme, but we present a way to bound the tree's height without the (significant) overhead of maintaining any auxiliary data structures.
The linchpin of our method is the following 
claim concerning violations.
\begin{compactenum}[\hspace{3mm}\bfseries {VIOL}:]
 \item If a violation is on the search path to $key$ before a rebalancing step, then the violation is still on the search path to $key$ after the rebalancing step, or it has been eliminated.
\end{compactenum}
While studying the rebalancing steps in \cite{Boyar97amortizationresults}, we realized that most of them
satisfy VIOL.
Furthermore, any time a rebalancing step would violate
VIOL
another rebalancing step that
satisfies VIOL
can be applied instead.
Hence, we always choose to perform rebalancing so that each violation created by an \ins$(key)$ or \del$(key)$ stays on the search path to $key$ until it is eliminated.
In our implementation, each 
\func{Insert} or \func{Delete}
that increases the number of violations cleans up after itself. 
It does this by invoking a procedure \func{Cleanup}$(key)$, which behaves like \func{Search}($key$) until it finds the first node $n_3$ on the search path where a violation occurs.
Then, \func{Cleanup}$(key)$ attempts to eliminate or move the violation at $n_3$ by 
invoking another procedure \tryrebalance\, which applies one localized rebalancing step at $n_3$, following the tree update template.
(\tryrebalance\ is similar to \del, and pseudocode is omitted, due to lack of space.)
\func{Cleanup}$(key)$ repeats these actions, searching for $key$ and invoking \tryrebalance\ to perform a rebalancing step, until the search goes all the way to a leaf without finding a violation.

\edited
In order to prove that each \ins\ or \del\ cleans up after itself, we must prove that while an invocation of \func{Cleanup}$(key)$ searches for $key$ by reading child pointers, it does not somehow miss the violation it is responsible for eliminating, even if a concurrent rebalancing step moves the violation upward in the tree, above where \func{Cleanup} is currently searching.
To see why this is true,
consider any rebalancing step that occurs while \cleanup\ is searching.
The rebalancing step is implemented using the tree update template, and looks like Figure~\ref{fig-replace-subtree}.
It takes effect at the point it changes
a child pointer $fld$ of some node $parent$ 
from a node $old$ to a node $new$.
If \func{Cleanup} reads $fld$ while searching, 
we argue that it
does not matter whether $fld$ contains $old$ or $new$.
First, suppose the violation is at a
node that is removed from the tree 
by the rebalancing step, or a child of such a node.
If the search passes through $old$, it will definitely reach the violation,
since nodes do not change after they are removed from the tree.
If the search passes through $new$,
VIOL
implies that the rebalancing step
either eliminated the violation, or moved it to a new node that will still be on the search path through $new$.
Finally, if the violation is further down in the tree, below the section modified by the concurrent rebalancing step, a search through either $old$ or $new$ will 
reach it.

Showing that \tryrebalance\ follows the template (i.e., by defining the procedures in Figure~\ref{code-dotreeupdate}) is complicated by the fact that it must decide which of the chromatic tree's 22 rebalancing steps to perform.
It is more convenient to unroll the loop that performs \llt s, and write \tryrebalance\ using conditional statements.
A helpful technique is to consider each path through the conditional statements in the code, and check that the procedures \func{Condition}, \func{NextNode}, \func{\sct-Arguments} and \func{Result} can be defined to produce this single path.
It is sufficient to show that this can be done for each path through the code, since it is always possible to use conditional statements to combine the procedures for each path into procedures that handle all paths.

\subsection{Proving a bound on the height of the tree} \label{sec-height}

\edited
Since we always perform rebalancing steps that 
satisfy VIOL,
if we reach a leaf without finding the violation that an \ins\ or \del\ created, then the violation has been eliminated.
This allows us to prove that the number of violations
in the tree at any time is bounded above by $c$, the number of insertions and deletions that are currently in progress.
Further, since removing all violations would yield a red-black tree with height $O(\log n)$, and eliminating each violation reduces the height by at most one,
the height of the chromatic tree is $O(c + \log n)$.

\subsection{Correctness and Progress}

As mentioned above, \func{Get}($key$) invokes \func{Search}$(key)$, which traverses a path from $entry$ to a leaf by reading child pointers.
Even though this search can pass through nodes that have been removed by concurrent updates,
we prove by induction that every node visited {\it was} on the search path for $key$ at some time during the search.
\func{Get} 
can thus be linearized when the leaf it reaches is on the search path for $key$ (and, hence, this leaf is the only one in the tree that could contain $key$).

Every \del\ operation that performs an update, and every \ins\ operation, is linearized at the \sct\ that performs the update.
Other \del\ operations (that return at line~\ref{del-no-gp} or \ref{del-notin}) behave like queries, and are linearized in the same way as \func{Get}.
Because no rebalancing step
modifies the set of keys stored in leaves,
the set of leaves always represents the set of dictionary entries.

The fact that our chromatic tree is non-blocking follows from P1 
and the fact that at most $3i+d$ rebalancing steps can be performed after $i$ insertions and $d$ deletions have occurred (proved in \cite{Boyar97amortizationresults}).

\subsection{\func{Successor} queries} \label{sec-chromatic-succ}

\func{Successor}($key$) runs an ordinary BST search algorithm, using \llt s to read the child fields of each node visited, until it reaches a leaf.
If the key of this leaf
is larger than $key$, it is returned and the operation is
linearized at any time during the operation
when this leaf was on the search path for $key$.
Otherwise, \func{Successor} finds the next leaf.
To do this, it remembers the last time it followed a left child pointer and, instead, follows one right child pointer, and then left child pointers until it reaches a leaf, using \llt s to read the child fields of each node visited.
If any \llt\ it performs
returns \fail\ or \finalized, \func{Successor} restarts.
Otherwise, it performs a validate-extended (\vlt), which returns \true\ only if all nodes on the path connecting the two leaves have
not changed.  If the \vlt\ succeeds, the key of the second leaf found is returned and
the query is linearized at the linearization point of the \vlt.
If the \vlt\ fails, \func{Successor} restarts.

\subsection{Allowing more violations}

Forcing insertions and deletions to rebalance the chromatic tree after creating only a single violation can cause unnecessary rebalancing steps to be performed, for example, because an overweight violation created by a deletion might be eliminated by a subsequent insertion.
In practice, we can reduce the total number of rebalancing steps that occur by modifying our \ins\ and \del\ procedures so that \cleanup\ is invoked only once the number of violations on a path from $entry$ to a leaf exceeds some constant $k$.
The resulting data structure has height $O(k + c + \log n)$.
Since searches in the chromatic tree are extremely fast, slightly increasing search costs to reduce update costs can yield significant benefits for update-heavy workloads.

\section{Experimental Results} \label{sec-exp}

\begin{figure*}[tb]
\def\darkness{45}
\def\expscale{0.75}
\def\expleftwidth{0.03\textwidth}
\centering
\hspace{-1cm}
\begin{minipage}{\expleftwidth}
\end{minipage}
\begin{minipage}{0.333\textwidth}
\small
\centering
\hspace{6mm}
\textbf{50\% \func{Ins}, 50\% \func{Del}, 0\% \func{Get}}
\end{minipage}
\begin{minipage}{0.333\textwidth}
\small
\centering
\hspace{4.5mm}
\textbf{20\% \func{Ins}, 10\% \func{Del}, 70\% \func{Get}}
\end{minipage}
\begin{minipage}{0.333\textwidth}
\small
\centering
\hspace{6mm}
\textbf{0\% \func{Ins}, 0\% \func{Del}, 100\% \func{Get}}
\end{minipage}\\
\hspace{-1cm}
\begin{minipage}{\expleftwidth}
\centering
\vspace{-2mm}
\includegraphics[scale=\expscale]{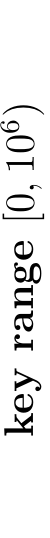}
\vspace{2mm}

\vspace{-3mm}
\includegraphics[scale=\expscale]{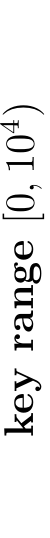}
\vspace{3mm}

\vspace{-6mm}
\includegraphics[scale=\expscale]{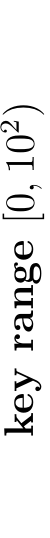}
\vspace{6mm}
\end{minipage}
\begin{minipage}{0.333\textwidth}
\centering
\includegraphics[scale=\expscale]{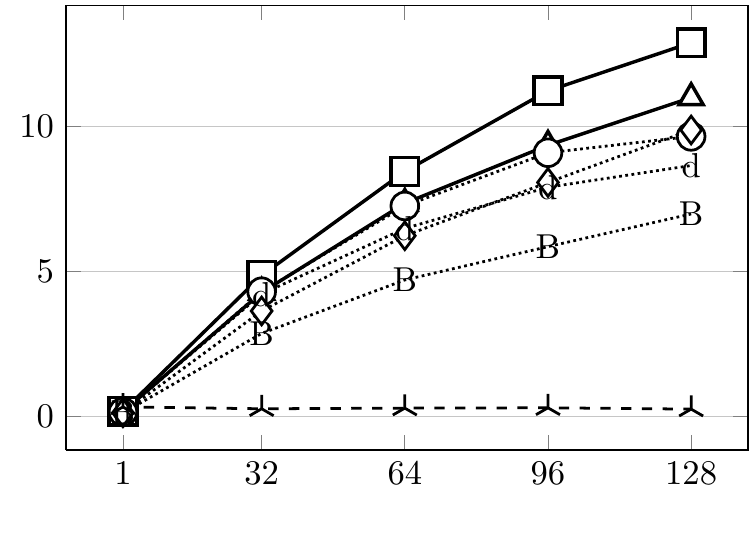}

\vspace{-2mm}
\includegraphics[scale=\expscale]{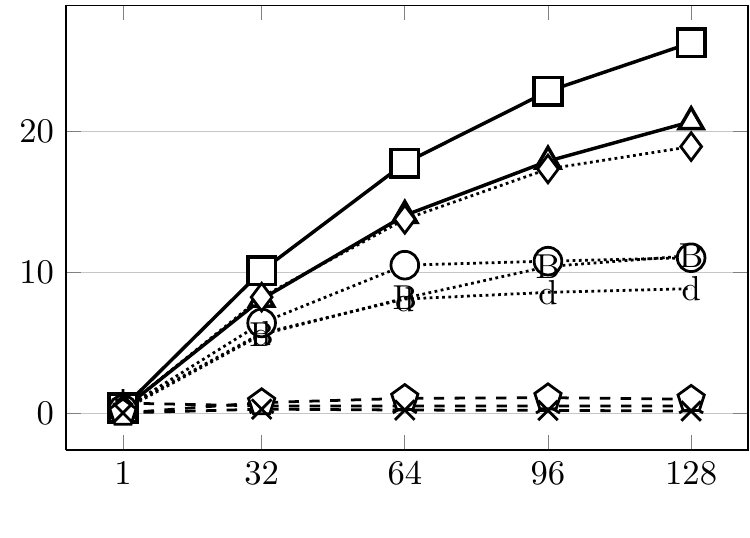}

\vspace{-2mm}
\includegraphics[scale=\expscale]{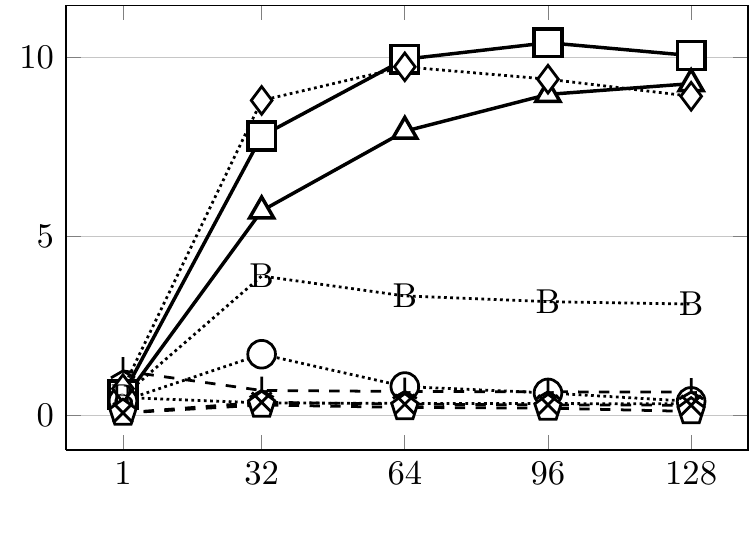}
\end{minipage}
\begin{minipage}{0.333\textwidth}
\centering
\vspace{-1.45mm}
\includegraphics[scale=\expscale]{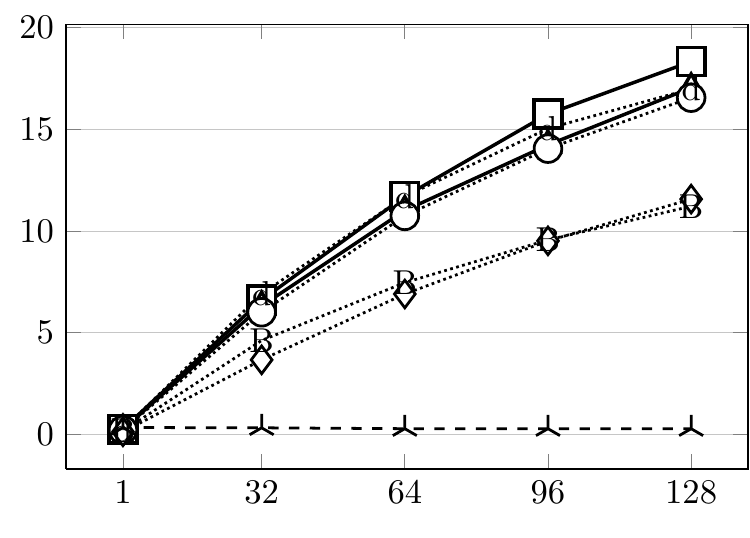}
\vspace{1.45mm}

\vspace{-6mm}
\includegraphics[scale=\expscale]{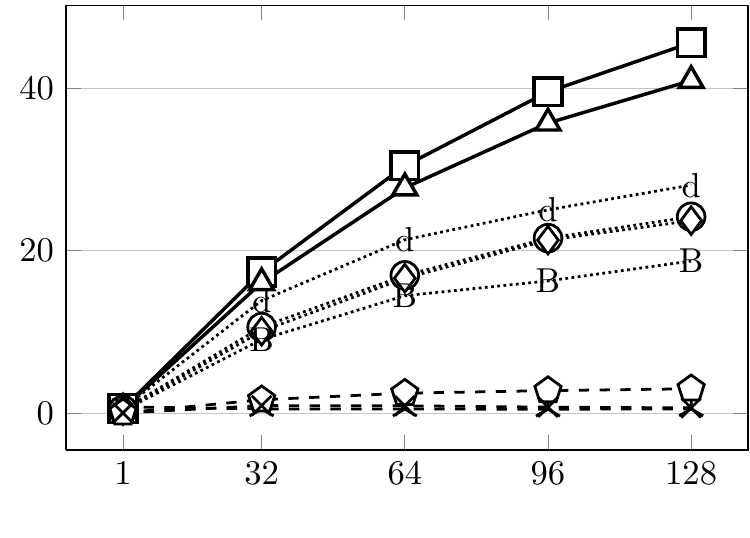}

\vspace{-2mm}
\includegraphics[scale=\expscale]{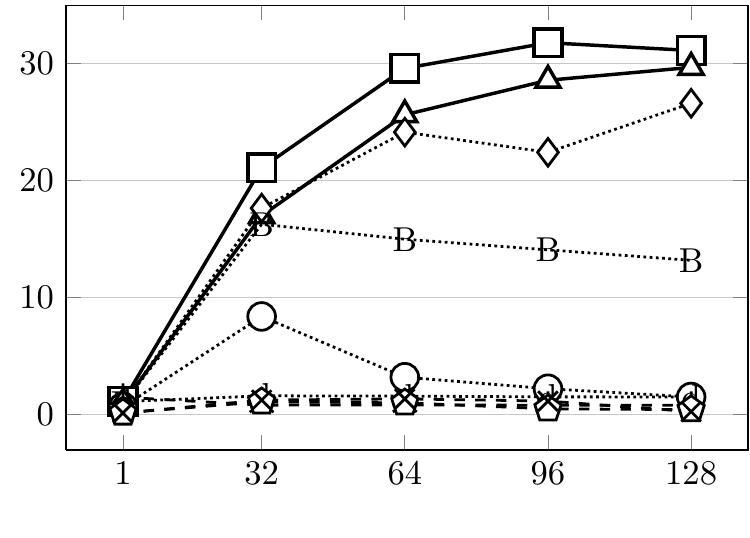}
\end{minipage}
\begin{minipage}{0.333\textwidth}
\centering
\vspace{1mm}
\includegraphics[scale=\expscale]{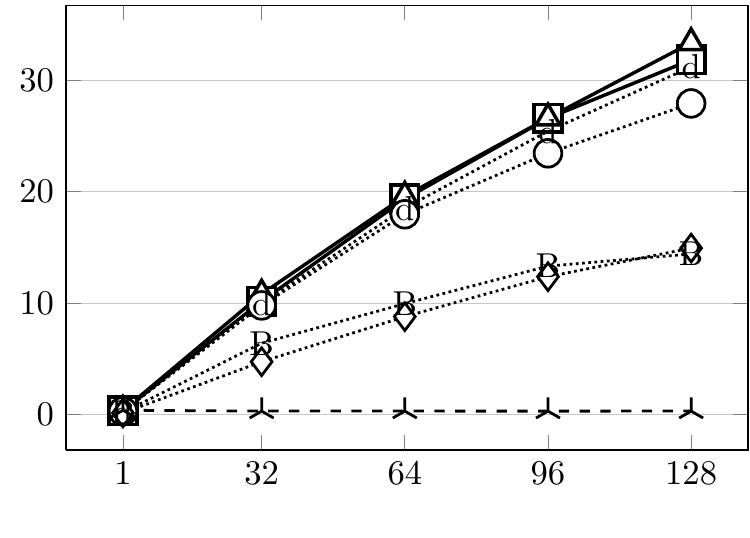}

\vspace{-2mm}
\vspace{-0.88mm}
\includegraphics[scale=\expscale]{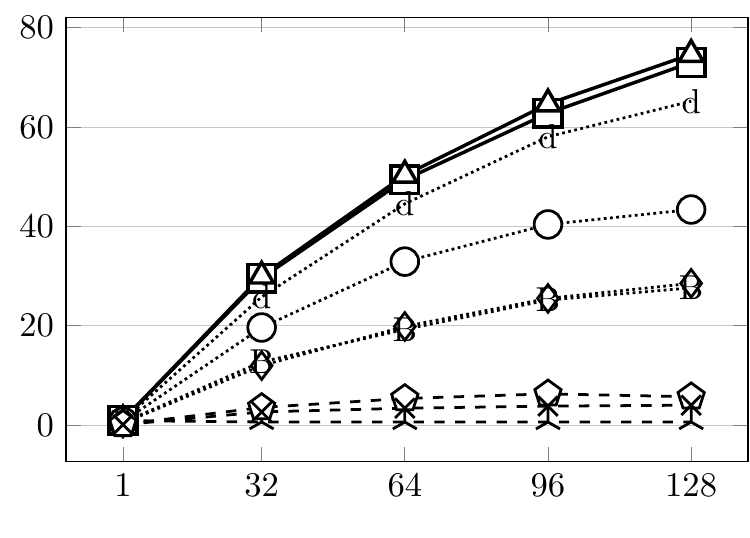}
\vspace{0.88mm}

\vspace{-6mm}
\includegraphics[scale=\expscale]{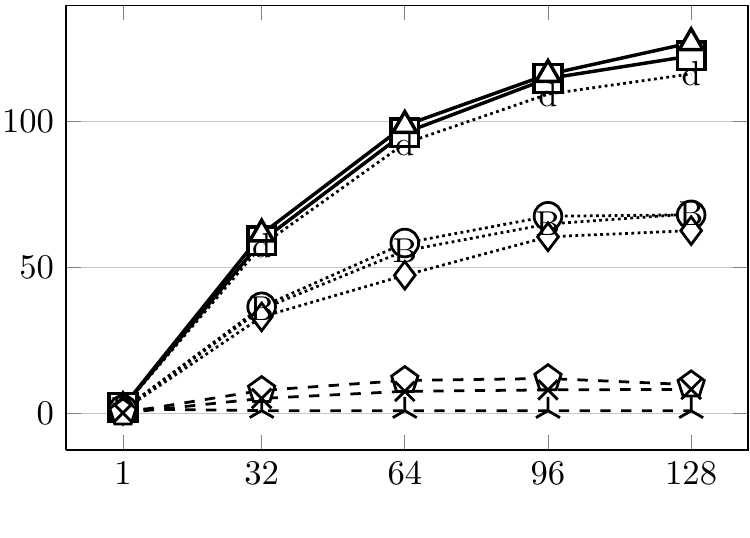}
\end{minipage}\\
\vspace{-4mm}
\includegraphics[scale=0.15]{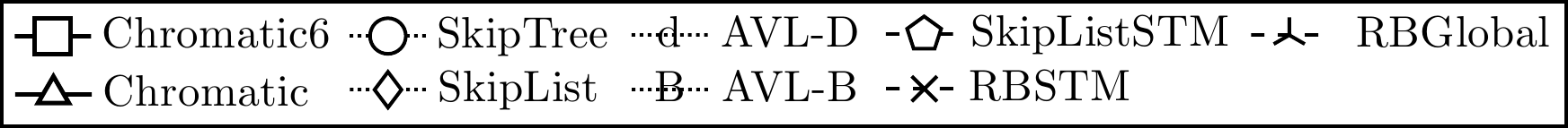}
\vspace{-1mm}
\caption{\textit{Multithreaded} throughput (millions of operations/second) for 2-socket SPARC T2+ (128 hardware threads) on y-axis versus number of threads on x-axis.}
\label{fig-experiments}
\vspace{-1mm}
\end{figure*}

We compared the performance of our chromatic tree (Chromatic) and the variant of our chromatic tree that invokes \cleanup\ only when the number of violations on a path exceeds six (Chromatic6) against several leading data structures that implement ordered dictionaries:
the non-blocking skip-list (SkipList) of the Java Class Library,
the non-blocking multiway search tree (SkipTree) of Spiegel and Reynolds~\cite{SR10},
the lock-based relaxed-balance AVL tree with non-blocking searches (AVL-D) of Drachsler et~al.~\cite{DVY14-incompletecitation}, and
the lock-based relaxed-balance AVL tree (AVL-B) of Bronson et~al.~\cite{BCCO10:ppopp}.
Our comparison also includes an STM-based red-black tree optimized by Oracle engineers (RBSTM) \cite{HLMS03:podc}, an STM-based skip-list (SkipListSTM), and the highly optimized Java red-black tree, \texttt{java.util.TreeMap}, with operations protected by a global lock (RBGlobal).
The STM data structures are implemented using DeuceSTM 1.3.0, which is one of the fastest STM implementations that does not require modifications to the Java virtual machine.
We used DeuceSTM's offline instrumentation capability to eliminate any STM instrumentation at running time that might skew our results.
All of the implementations that we used were made publicly available by their respective authors.
For a fair comparison between data structures, we made slight modifications to RBSTM and SkipListSTM to use generics, instead of hardcoding the type of keys as \texttt{int}, and to store values in addition to keys. 
Java code for Chromatic and Chromatic6 is available from
\mbox{\url{http://implementations.tbrown.pro}}.

We tested the data structures for three different operation mixes, 0i-0d, 20i-10d and 50i-50d, where $x$i-$y$d denotes $x$\% \ins s, $y$\% \del s, and $(100-x-y)$\% \func{Get}s, to represent the cases 
when all of the operations are queries, 
when a moderate proportion of the operations are \ins s and \del s, and 
when all of the operations are \ins s and \del s. 
We used three key ranges, $[0,10^2), [0,10^4)$ and $[0,10^6)$, to test different contention levels.
For example, for key range $[0,10^2)$, data structures will be small, so 
updates are likely to affect overlapping parts of the data structure.

For each data structure, each operation mix, each key range, and each thread count in 
\{1, 32, 64, 96, 128\}, we ran five trials which each measured the total throughput (operations per second) of all threads for five seconds.
Each trial began with an untimed prefilling phase, which continued until the data structure was within 5\% of its expected size in the steady state.
For operation mix 50i-50d, the expected size is half of the key range.
This is because, eventually, each key in the key range has been inserted or deleted at least once, and the last operation on any key is equally likely to be an insertion (in which case it is in the data structure) or a deletion (in which case it is not in the data structure).
Similarly, 20i-10d yields an expected size of two thirds of the key range since, eventually, each key has been inserted or deleted and the last
operation
on it is twice as likely to be an insertion as a deletion.
For 
0i-0d, we prefilled to half of the key range.

We used a Sun SPARC Enterprise T5240 with 32GB of RAM and two UltraSPARC T2+ processors, for a total of 16 1.2GHz cores supporting a total of 128 hardware threads.
The Sun 64-bit JVM version 1.7.0\_03 was run in server mode, with 3GB minimum and maximum heap sizes.
Different experiments run within a single instance of a Java virtual machine (JVM) are not statistically independent, so each batch of five trials was run in its own JVM instance.
Prior to running each batch, a fixed set of three trials was run to cause the Java HotSpot compiler to optimize the running code.
Garbage collection was manually triggered before each trial.
The heap size of 3GB was small enough that garbage collection was performed regularly (approximately ten times) in each trial.
We did not pin threads to cores, since this is unlikely to occur in practice. 

Figure~\ref{fig-experiments} shows our experimental results.
Our algorithms are drawn with solid lines.
Competing handcrafted implementations are drawn with dotted lines.
Implementations with coarse-grained synchronization are drawn with dashed lines.
Error bars are not drawn because they are mostly too small to see: The standard deviation is less than 2\% of the mean for half of the data points, and less than 10\% of the mean for 95\% of the data points.
The STM data structures are not included in the graphs for key range $[0, 10^6)$, because of the enormous length of time needed just to perform prefilling (more than 120 seconds per five second trial).

Chromatic6 nearly always outperforms Chromatic.
The only exception is for an all query workload, where Chromatic performs slightly better.
Chromatic6 is prefilled with the Chromatic6 insertion and deletion algorithms, so it has a slightly larger average leaf depth than Chromatic; this accounts for the performance difference.
In every graph, Chromatic6 rivals or outperforms the other data structures, even the highly optimized implementations of SkipList and SkipTree which were crafted with the help of Doug Lea and 
the Java Community Process JSR-166 Expert Group.
Under high contention (key range $[0,10^2)$), Chromatic6 outperforms every competing data structure except for SkipList in case 50i-50d and AVL-D in case 0i-0d.
In the former case, SkipList 
approaches the performance of Chromatic6 when there are many \ins s and \del s, due to the simplicity of its updates.
In the latter case, the non-blocking searches of AVL-D allow it to perform nearly as well as Chromatic6; this is also evident for the other two key ranges.
SkipTree, AVL-D and AVL-B all experience negative scaling beyond 32 threads when there are updates.
For SkipTree, this is because its nodes contain many child pointers, and processes modify a node by replacing it (severely limiting concurrency when the tree is small).
For AVL-D and AVL-B, this is likely because processes waste time waiting for locks to be released when they perform updates.
%
Under moderate contention (key range $[0,10^4)$), in cases 50i-50d and 20i-10d, Chromatic6 significantly outperforms the other data structures. 
Under low contention, the advantages of a non-blocking approach are less pronounced, but Chromatic6 is still at the top of each graph (likely because of low overhead and searches that ignore updates).

Figure~\ref{fig-single-thread} compares the single-threaded performance of the data structures, relative to the performance of the sequential RBT, \texttt{java.util.TreeMap}.
This demonstrates that the overhead introduced by our technique is relatively small.

\begin{figure}[tb]
\vspace{-3mm}
\centering
\includegraphics[scale=0.33]{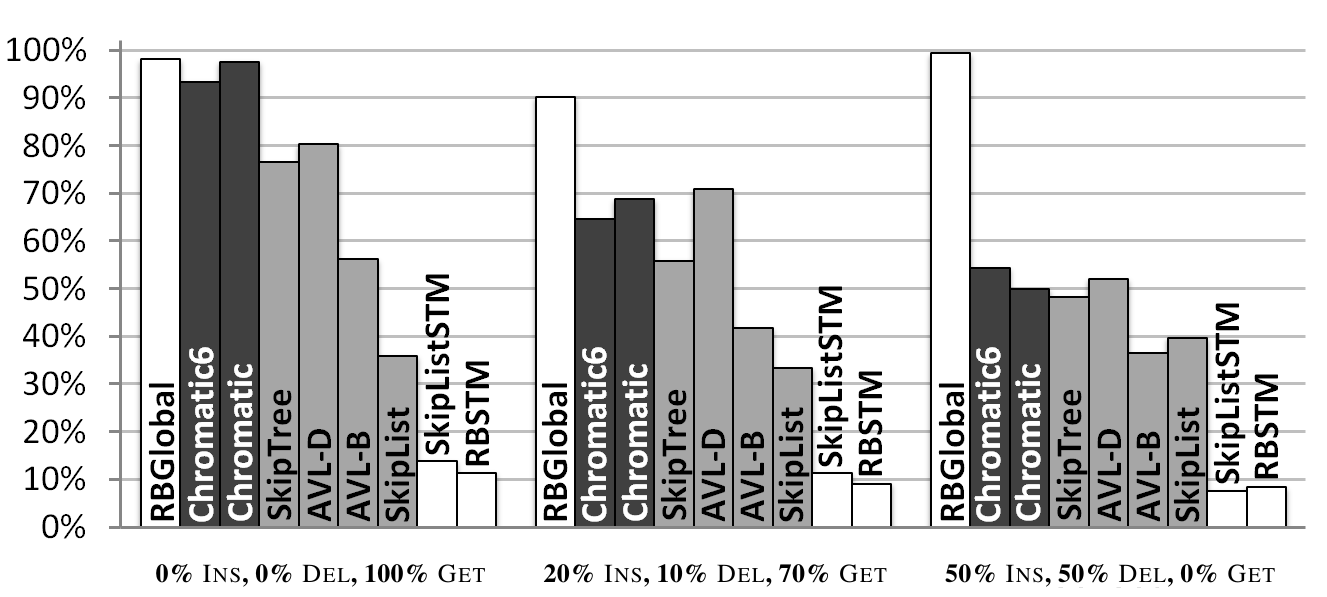}
\caption{\textit{Single threaded} throughput of the data structures relative to Java's sequential RBT for key range $[0,10^6)$.}
\label{fig-single-thread}
\end{figure}
Although balanced BSTs are designed to
give performance guarantees for worst-case sequences of operations,
the experiments are performed using random sequences.
For such sequences, BSTs without rebalancing operations
are balanced with high probability and, hence, will have better performance
because of their lower overhead.
Better experiments are needed to evaluate balanced BSTs.

\section{Conclusion} \label{sec-conclusion}

In this work, we presented a template that can be used to obtain non-blocking implementations of any data structure based on a down-tree, and demonstrated its use by implementing a non-blocking chromatic tree.
To the authors' knowledge, this is the first provably correct, non-blocking balanced BST with fine-grained synchronization.
Proving the correctness of a direct implementation of a chromatic tree from hardware primitives would have been completely intractable.
By developing our template abstraction and our chromatic tree in tandem, we were able to avoid introducing any extra overhead, so our chromatic tree is very efficient.

Given a copy of \cite{Lar00}, and this paper, a first year undergraduate student produced our Java implementation of a relaxed-balance AVL tree in less than a week.
Its performance was slightly lower than that of Chromatic.
After allowing more violations on a path before rebalancing, its performance was indistinguishable from that of Chromatic6.

We hope that this work sparks interest in developing more relaxed-balance sequential versions of data structures, since it is now easy to obtain efficient concurrent implementations of them using our template. 

\subsection*{Acknowledgements}
This work was supported by NSERC.
We thank Oracle for providing the machine used for our experiments.

\bibliographystyle{abbrv}
\bibliography{bibliography}

\newpage
\appendix

\section{Primitives}
\label{primitives-spec}

Since the definition of \llt, \vlt\ and \sct\
has not yet appeared in print, we include
the specification of the primitives here for completeness.
However, a much fuller discussion of the primitives can be found in \cite{paper1}.

The \vlt\ primitive is an extension of the usual validate primitive:  it takes as an argument
a sequence $V$ of \rec s.  Intuitively, a \vlt($V$) by a process $p$
returns true only if no \sct\ has modified
any record $r$ in $V$ since the last \llt\ by $p$.  

\begin{defn} \label{defn-llt-linked-to-sct}
Let $I'$ be an invocation of \sct$(V, R, fld, new)$ or \vlt$(V)$ by a process $p$, and $r$ be a \rec \ in $V$.
We say an invocation $I$ of \llt$(r)$ is \textbf{linked to} $I'$ if and only if:
\begin{enumerate}
\item $I$ returns a value different from \fail\ or \finalized, and
\label{prop-returns-value-different-from-fail-or-finalized}
\item no invocation of \llt$(r)$, \sct$(V', R', fld', new')$, or \vlt$(V')$, where $V'$ contains $r$, is performed by $p$ between $I$ and $I'$.
\label{prop-no-sct-or-vlt-between-linked-llt-and-sct-or-vlt}
\end{enumerate}
\end{defn}

Before calling \sct\ or \vlt,
$p$ must perform a linked \llt($r$) on each \rec\ $r$ in $V$. 

For every execution, there is a linearization of all successful \llt s,
all successful \sct s, a subset of the non-terminating \sct s,
all successful \vlt s, and all reads, such that the following conditions
are satisfied.
\begin{compactitem}
	\item 
		Each read of a field $f$ of a \rec\ $r$ returns
		the last value stored in $f$ by an \sct\ linearized before the read
		(or $f$'s initial value, if no such \sct\ has modified $f$).
	\item 
Each linearized \llt($r$) that does not return \finalized\
		returns	the last value stored in each mutable field $f$ of $r$
		by an \sct\ linearized before the \llt (or $f$'s initial value, if no such \sct\ has modified~$f$).
	\item 
		Each linearized \llt$(r)$ returns \finalized\ if and only if it is linearized
		after an 
		\sct($V, R, fld,$ $new$) with $r$ in $R$.
	\item 
For each linearized invocation $I$ of \sct($V, R,$ $fld, new$) or \vlt$(V)$,
and for each $r$ in $V$, 
no \sct($V'$, $R'$, $fld'$, $new'$) with $r$ in $V'$ is linearized between the
\llt$(r)$ linked to $I$ and $I$.
\end{compactitem}

Moreover, we have the following progress properties.
\begin{compactitem}
	\item 
        Each terminating \llt$(r)$ returns \finalized\ if it begins after the end of a successful \sct($V, R,$ $fld, new$) with $r$ in $R$ or after another \llt$(r)$ has returned \finalized. 
    \item If operations are performed infinitely often, then operations succeed infinitely often.
    \item If \sct\ and \vlt\ operations are performed infinitely often, then \sct\ or \vlt\ operations succeed infinitely often.
    \item If \sct\ operations are performed infinitely often, then \sct\ operations succeed infinitely often.
\end{compactitem}

\section{Proof for Tree Update Template}
\label{app-tree-proof}

\subsection{Additional properties of \llt/\sct/\vlt}

We first import some useful definitions and properties of \llt/\sct/\vlt\ that were proved in the appendix of the companion paper \cite{paper1}.

\begin{defn} \label{defn-set-up-sct}
A process $p$ \textbf{sets up} an invocation of \sct$(V, R, fld, new)$ by invoking \llt$(r)$ for each $r \in V$, and then invoking \sct$(V, R, fld, new)$ if none of these \llt s return \fail\ or \finalized.
\end{defn}

\begin{defn} \label{defn-rec-in-added-removed}
A \rec \ $r$ is \textbf{in the data structure} in some configuration $C$ if and only if $r$ is reachable by following pointers from an entry point.
We say $r$ is \textbf{initiated} if it has ever been in the data structure.
We say $r$ is \textbf{removed (from the data structure) by} some step $s$ if and only if $r$ is in the data structure immediately before $s$, and $r$ is not in the data structure immediately after $s$.
We say $r$ is \textbf{added (to the data structure) by} some step $s$ if and only if $r$ is not in the data structure immediately before $s$, and $r$ is in the data structure immediately after $s$.
\end{defn}

Note that a \rec \ can be removed from or added to the data structure only by a linearized invocation of \sct.
The results of this section holds only if Constraint~3 (from section~\ref{sec-primitives}) is satisfied. 

%


\begin{lem} {\rm\cite{paper1}}\label{lem-rec-in-data-structure-just-before-llt}
If an invocation $I$ of $\llt(r)$ returns a value different from \fail\ or \finalized, then $r$ is in the data structure just before $I$ is linearized.
\end{lem}

\begin{lem} {\rm\cite{paper1}}
\label{lem-rec-in-data-structure-after-linearized-sct}
If $S$ is a linearized invocation of \sct$(V, R, fld, new)$, where $new$ is a \rec, then $new$ is in the data structure just after $S$.
\end{lem}

Let $C_1$ and $C_2$ be configurations in the execution.
We use $C_1 < C_2$ to mean that $C_1$ precedes $C_2$ in the execution.
We say $C_1 \le C_2$ precisely when $C_1 = C_2$ or $C_1 < C_2$.
We denote by $[C_1, C_2]$ the set of configurations $\{C \mid C_1 \le C \le C_2\}$.

\begin{lem} {\rm\cite{paper1}} \label{lem-if-rec-traversed-then-rec-in-data-structure}
Let $r_1,r_2,...,r_l$ be a sequence of \rec s, where $r_1$ is an entry point, and $C_1,C_2,...,C_{l-1}$ be a sequence of configurations satisfying $C_1 < C_2 < ... < C_{l-1}$.
If, for each $i \in \{1, 2, ..., l-1\}$, a field of $r_i$ points to $r_{i+1}$ in configuration $C_i$, then $r_{i+1}$ is in the data structure in some configuration in $[C_1, C_i]$.
Additionally, if a mutable field $f$ of $r_l$ contains a value $v$ in some configuration $C_l$ after $C_{l-1}$ then, in some configuration in $[C_1, C_l]$, $r_l$ is in the data structure and $f$ contains $v$.
\end{lem}

\subsection{Correctness of the tree update template}\label{sec-dotreeup-correctness}

In this section, we refer to an operation that follows the tree update template simply as a {\it tree update operation}.
The results of this section apply only if the following constraint is satisfied.

\begin{con} \label{con-dotreeup-exclusively-does-sct}
Every invocation of \sct\ is performed by a tree update operation.
\end{con}

In the proofs that follow, we sometimes argue about when invocations of \llt, \sct\ and \vlt\ are linearized.
However, the arguments we make do not require any knowledge of the linearization points chosen by any particular implementation of these primitives.
We do this because the behavior of \llt, \sct\ and \vlt\ is defined in terms of when operations are linearized, relative to one another.
Similarly, we frequently refer to \textit{linearized} invocations of \sct, rather than \textit{successful} invocations, because it is possible for a non-terminating invocation of \sct\ to modify the data structure, and we linearize such invocations.

Since we refer to the preconditions of \llt\ and \sct\ in the following, we reproduce them here.
\begin{compactitem}
\item \llt$(r)$: $r$ has been initiated
\item \sct$(V, R, fld, new)$:
\begin{compactenum}
\item for each $r \in V$, $p$ has performed an invocation $I_r$ of $\llt(r)$ linked to this \sct
\item $new$ is not the initial value of $fld$
\item for each $r \in V$, no \sct$(V', R', fld, new)$ was linearized before $I_r$ was linearized
\end{compactenum}
\end{compactitem}

The following lemma establishes Constraint~\conmarkallremovedrecs, and some other properties that will be useful when proving linearizability.


\begin{lem} \label{lem-dotreeup-constraints-invariants}
The following properties hold in any execution of tree update operations.
\begin{enumerate}
\item    Let $S$ be a linearized invocation of \sct$(V, R, fld, new)$, 
         and $G$ be the directed graph induced by the edges read by the \llt s linked to $S$.
         $G$ is a sub-graph of the data structure at all times after the last \llt\ linked to $S$ and before $S$ is linearized, and no node in the $N$ set of $S$ is in the data structure before $S$ is linearized.
\label{claim-dotreeup-G-in-data-structure}
\item    Every invocation of \llt\ or \sct\ performed by a tree update operation has valid arguments, and satisfies its preconditions.
\label{claim-dotreeup-llt-sct-preconditions}
\item    Let $S$ be a linearized invocation of \sct$(V, R, fld, new)$, where $fld$ is a field of $parent$, and $old$ is the value read from $fld$ by the $\llt(parent)$ linked to $S$. 
         $S$ changes $fld$ from $old$ to $new$, replacing a connected subgraph containing nodes $R \cup F_N$ with another connected subgraph containing nodes $N \cup F_N$.
         Further, the \rec s added by $S$ are precisely those in $N$, and the \rec s removed by $S$ are precisely those in $R$.
\label{claim-dotreeup-finalized-before-removed}
\item    At all times, $root$ is the root of a tree of \node s.
         (We interpret $\bot$ as the empty tree.)
\label{claim-dotreeup-tree}
\end{enumerate}
\end{lem}
\begin{proof}
We prove these claims by induction on the sequence of steps taken in the execution.
Clearly, these claims hold initially.
Suppose they hold before some step $s$.
We prove they hold after $s$.
Let $O$ be the operation that performs $s$.

\textbf{Proof of Claim~\ref{claim-dotreeup-G-in-data-structure}.}
To affect this claim, $s$ must be a linearized invocation of \sct$(V, R, fld, new)$.
Since $s$ is linearized, the semantics of \sct\ imply that, for each $r \in V $, no \sct$(V', R', fld', new')$ with $r \in V'$ is linearized between the linearization point of the invocation $I$ of $\llt(r)$ linked to $s$ and the linearization point of $s$.
Thus, for each $r \in V $, no mutable field of $r$ changes between when $I$ and $s$ are linearized.
We now show that all nodes and edges of $G$ are in the data structure at all times after the last \llt\ linked to $s$ is linearized, and before $s$ is linearized.
Fix any arbitrary $r \in V$.
By inductive Claim~\ref{claim-dotreeup-llt-sct-preconditions}, $I$ satisfies its precondition, so $r$ was initiated when $I$ started and, hence, was in the data structure before $I$ was linearized.
By the semantics of \sct, since $I$ returns a value different from \fail\ or \finalized, no invocation of \sct$(V'', R'', fld'', new'')$ with $r \in R''$ is linearized before $I$ is linearized.
By 
inductive Claim~\ref{claim-dotreeup-finalized-before-removed}, Constraint~\conmarkallremovedrecs\ is satisfied at all times before $s$.
Thus, $r$ is not removed before $I$.
Since $s$ is linearized, no invocation of \sct$(V'', R'', fld'', new'')$ with $r \in R''$ is linearized between the linearization points of $I$ and $s$.
(If such an invocation were to occur then, since $r$ would also be in $V''$, the semantics of \sct\ would imply that $s$ could not be linearized.)
Since the linearization point of $I$ is before that of $s$, 
$r$ is not removed before $s$ is linearized.
When $I$ is linearized, since $r$ is in the data structure, all of its children are also in the data structure.
Since no mutable field of $r$ changes between the linearization points of $I$ and $s$, all of $r$'s children read by $I$ are in the data structure throughout this time.
Thus, each node and edge in $G$ is in the data structure at all times after the last \llt\ linked to $s$, and before $s$.

Finally, we prove that no node in $N$ is in the data structure before $s$ is linearized.
Since $O$ follows the tree update template, $s$ is its only modification to shared memory.
Since each $r' \in N$ is newly created by $O$, it is clear that $r'$ can only be in the data structure after $s$ is linearized.

\textbf{Proof of Claim~\ref{claim-dotreeup-llt-sct-preconditions}.}
Suppose $s$ is invocation of $\llt(r)$.
Then, $r \neq \nil$ (by the discussion in Sec.~\ref{sec-dotreeupdate}).
By the code in Figure~\ref{code-dotreeupdate}, either $r = top$, or $r$ was obtained from the return value of some invocation $L$ of \llt$(r')$ previously performed by $O$.
If $r$ was obtained from the return value of $L$, 
then Lemma~\ref{lem-rec-in-data-structure-just-before-llt} implies that $r'$ is in the data structure when $L$ is linearized.
Hence, $r$ is in the data structure when $L$ is linearized, which implies that $r$ is initiated when $s$ occurs.
Now, suppose $r = top$.
By the precondition of $O$, $r$ was reached by following child pointers from $root$ since the last operation by $p$.
By 
inductive Claim~\ref{claim-dotreeup-finalized-before-removed}, Constraint~\conmarkallremovedrecs\ is satisfied at all times before $s$.
Therefore, we can apply Lemma~\ref{lem-if-rec-traversed-then-rec-in-data-structure}, which implies that $r$ was in the data structure at some point before the start of $O$ (and, hence, before $s$).
By Definition~\ref{defn-rec-in-added-removed}, $r$ is initiated when $s$ begins.

Suppose $s$ is an invocation of \sct$(V, R, fld, new)$.
By Condition~C\ref{con-parent-in-V}, $fld$ is a mutable (child) field of some node $parent \in V$.
By Condition~C\ref{con-R-subsequence-of-V}, $R$ is a subsequence of $V$.
Therefore, the arguments to $s$ are valid.
By Condition~C\ref{con-llt-on-all-nodes-in-V} and the definition of $\sigma$, for each $r \in V$, $O$ performs an invocation $I$ of $\llt(r)$ before $s$ that returns a value different from \fail\ or \finalized\ (and, hence, is linked to $s$), 
so $s$ satisfies Precondition~\presctlinked\ of \sct.

We now prove that $s$ satisfies \sct\ Precondition~\presctabainit.
Let $parent.c_i$ be the field pointed to by $fld$.
If $parent.c_i$ initially contains $\bot$ then, by Condition~C\ref{con-GN-non-empty-tree}, 
$new$ is a \node, and we are done.  
Suppose $parent.c_i$ initially points to some \node\ $r$.
We argued in the previous paragraph that $O$ performs an \llt$(parent)$ linked to $s$ before $s$ is linearized.
By 
inductive Claim~\ref{claim-dotreeup-finalized-before-removed}, Constraint~\conmarkallremovedrecs\ is satisfied at all times before $s$.
Therefore, we can apply Lemma~\ref{lem-if-rec-traversed-then-rec-in-data-structure} to show that $r$ was in the data structure at some point before the start of $O$ (and, hence, before $s$).
However, by inductive Claim~\ref{claim-dotreeup-G-in-data-structure} (which we have proved for $s$), 
$new$ cannot be initiated before $s$, so $new \neq r$.

Finally, we show $s$ satisfies \sct\ Precondition~\presctaba.
Fix any $r' \in V$, and let $L$ be the \llt$(r')$ linked to $s$ performed by $O$.
To derive a contradiction, suppose an invocation $S'$ of \sct$(V', R', fld, new)$ is linearized before $L$ (which is before $s$).
By Lemma~\ref{lem-rec-in-data-structure-after-linearized-sct} (which we can apply since Constraint~\conmarkallremovedrecs\ is satisfied at all times before $s$), $new$ would be in the data structure (and, hence, initiated) before $s$ is linearized.
However, this contradicts our argument that $new$ cannot be initiated before $s$ occurs.

\textbf{Proof of Claim~\ref{claim-dotreeup-finalized-before-removed} and Claim~\ref{claim-dotreeup-tree}.}
To affect these claims, $s$ must be a linearized invocation of \sct$(V,$ $R,$ $fld,$ $new)$.
Let $t$ be when $s$ is linearized.
The semantics of \sct\ and the fact that $s$ is linearized imply that, for each $r \in V$, no \sct$(V', R', fld', new')$ with $r \in V'$ is linearized after the invocation $I$ of $\llt(r)$ linked to $s$ is linearized and before $t$.
Thus, $O$ satisfies tree update template Condition~C\ref{con-old-nil-then-R-empty}, C\ref{con-R-non-empty-then-GR-a-non-empty-tree}, C\ref{con-GN-non-empty-tree} and C\ref{con-fringe-of-GN-is-old}.
By inductive Claim~\ref{claim-dotreeup-G-in-data-structure} (which we have proved for $s$), all nodes and edges in $G$ are in the data structure just before $t$, and no node in $N$ is in the data structure before $t$.
Let $parent.c_i$ be the mutable (child) field changed by $s$, and $old$ be the value read from $parent.c_i$ by the $\llt(parent)$ linked to $s$.
%

Suppose $R \neq \emptyset$ (as in Fig.~\ref{fig-replace-subtree}).
Then, by Condition~C\ref{con-R-non-empty-then-GR-a-non-empty-tree}, $G_R$ is a tree rooted at $old$ and $F_N = F_R$.
Since $G_R$ is a sub-graph of $G$, inductive Claim~\ref{claim-dotreeup-G-in-data-structure} implies that each node and edge of $G_R$ is in the data structure just before $t$.
Further, since $O$ performs an \llt$(r)$ linked to $s$ for each $r \in R$, and no child pointer changes between this \llt\ and time $t$, $G_R$ contains every node that was a child of a node in $R$ just before $t$.
Thus, $F_R$ contains every node $r \notin R$ that was a child of a node in $R$ just before $t$.
This implies that, just before $t$, for each node $r \notin R$ in the sub-tree rooted at $old$, $F_R$ contains $r$ or an ancestor of $r$.
By inductive Claim~\ref{claim-dotreeup-tree}, just before $t$, every path from $root$ to a descendent of $old$ passes through $old$.
Therefore, just before $t$, every path from $root$ to a node in $\{$descendents of $old\} - R$ passes through a node in $F_R$. 
Just before $t$, by the definition of $F_R$, and the fact that the nodes in $R$ form a tree, $R \cap F_R$ is empty and no node in $R$ is a descendent of a node in $F$.
By Condition~C\ref{con-GN-non-empty-tree}, $G_N$ is a non-empty tree rooted at $new$ with node set $N \cup F_N = N \cup F_R$, where $N$ contains nodes that have not been in the data structure before $t$.
Since $parent.c_i$ is the only field changed by $s$, $s$ replaces a connected sub-graph with node set $R \cup F_R$ by a connected sub-graph with node set $N \cup F_R$.
We prove that $parent$ was in the data structure just before $t$.
Since $s$ modifies $parent.c_i$, just before $t$, $parent$ must not have been finalized.
Thus, no \sct$(V', R', fld', new')$ with $parent \in R'$ can be linearized before $t$.
By inductive Claim~\ref{claim-dotreeup-finalized-before-removed}, Constraint~\conmarkallremovedrecs\ is satisfied at all times before $t$, so $parent$ cannot be removed from the data structure before $t$.
By inductive Claim~\ref{claim-dotreeup-llt-sct-preconditions}, the precondition of the \llt$(parent)$ linked to $s$ implies that $parent$ was initiated, so $parent$ was in the data structure just before $t$.
Since no node in $N$ is in the data structure before $t$, the \rec s added by $s$ are precisely those in $N$.
Since, just before $t$, no node in $R$ is in $F$, or a descendent of a node in $F$, and every $r \in \{$descendents of $old\}-R$ is reachable from a node in $F_R$, the \rec s removed by $s$ are precisely those in $R$.

Now, to prove Claim~\ref{claim-dotreeup-tree}, we need only show that $parent.c_i$ is the root of a sub-tree just after $t$.
We have argued that $old$ is the root of $G_R$ just before $t$.
Since $parent.c_i$ points to $old$ just before $t$, the inductive hypothesis implies that $old$ is the root of a subtree, and $parent$ is not a descendent of $old$.
%
Therefore, just before $t$, $parent$ is not in any sub-tree rooted at a node in $F_R$.
This implies that no descendent of $old$ is changed by $s$.
By inductive Claim~\ref{claim-dotreeup-tree}, each $r \in F$ is the root of a sub-tree just before $t$, so each $r \in F$ is the root of a sub-tree just after $t$.
%
%
Finally, since Condition~C\ref{con-GN-non-empty-tree} states that $G_N$ is a non-empty down-tree rooted at $new$, and we have argued that $G_N$ has node set $N \cup F_R$, $parent$ is the root of a sub-tree just after $t$.

The two other cases, where $R = \emptyset$ and $old = \nil$ (as in Fig.~\ref{fig-replace-subtree2}(a)), and where $R = \emptyset$ and $old \neq \nil$ (as in Fig.~\ref{fig-replace-subtree2}(b)), are similar (and substantially easier to prove).
\end{proof}

\begin{obs} \label{obs-dotreeup-satisfies-con-mark-all-removed-recs}
Constraint~\conmarkallremovedrecs\ is implied by Lemma~\ref{lem-dotreeup-constraints-invariants}.\ref{claim-dotreeup-finalized-before-removed}.
\end{obs}

We call a tree update operation \textbf{successful} if it performs a linearized invocation of \sct\ (which either returns \true, or does not terminate).
We linearize each successful tree update operation at its linearized invocation of \sct.
Clearly, each successful tree update operation is linearized during that operation.

\begin{thm} \label{thm-dotreeup-linearizable}
At every time $t$, the tree $T$ rooted at $root$ is the same as the tree $T_L$ that would result from the atomic execution of all tree update operations linearized up until $t$, at their linearization points.
Moreover, the return value of each tree update operation is the same as it would be if it were performed atomically at its linearization point.
\end{thm}
\begin{proof}
The steps that affect $T$ are linearized invocations of \sct.
The steps that affect $T_L$ are successful tree update operations, each of which is linearized at the linearized invocation of \sct\ that it performs.
Thus, the steps that affect $T$ and $T_L$ are the same.
We prove this claim by induction on the sequence of linearized invocations of \sct.
Clearly, the claim holds before any linearized invocation of \sct\ has occurred.
We suppose it holds just before some linearized invocation $S$ of \sct$(V, R, fld, new)$, and prove it holds just after $S$.
Consider the tree update operation $OP(top, args)$ that performs $S$.
Let $O_L$ be the tree update operation $OP(top, args)$ in the linearized execution that occurs at $S$.
By the inductive hypothesis, before $S$, $T = T_L$.
To show that $T = T_L$ holds just after $S$, we must show that $O$ and $O_L$ perform invocations of \sct\ with exactly the same arguments.
In order for the arguments of these \sct s to be equal, the \func{SCX-Arguments}$(s_0, ..., s_i, args)$ computations performed by $O$ and $O_L$ must have the same inputs.
Similarly, in order for the return values of $O$ and $O_L$ to be equal, their \func{Result}$(s_0, ..., s_i, args)$ computations must have the same inputs.
Observe that $O$ and $O_L$ have the same $args$, and $s_0, ..., s_i$ consist of the return values of the \llt s performed by the operation, and the immutable fields of the nodes passed to these \llt s.
Therefore, it suffices to prove that the arguments and return values of the \llt s performed by $O$ and $O_L$ are the same.

We prove that each \llt\ performed by $O$ returns the same value as it would if it were performed atomically at $S$ (when $O_L$ is atomically performed).
Since $S$ is linearized, for each $\llt(r)$ performed by $O$, no invocation of \sct$(V', R', fld', new')$ with $r \in V'$ is linearized between when this $\llt$ is linearized and when $S$ is linearized.
Thus, the $\llt(parent)$ performed by $O$ returns the same result that it would if it were performed atomically when $S$ occurs.

Let $I^k$ and $I_L^k$ be the $k$th \llt s by $O$ and $O_L$, respectively, $a^k$ and $a_L^k$ be the respective arguments to $I^k$ and $I_L^k$, and $v^k$ and $v_L^k$ be the respective return values of $I^k$ and $I_L^k$.
We prove by induction that $a^k = a_L^k$ and $v^k = v_L^k$ for all $k \ge 1$.

\textbf{Base case:}
Since $O$ and $O_L$ have the same $top$ argument, $a^1 = a_L^1 = top$.
Since each \llt\ performed by $O$ returns the same value as it would if it were performed atomically at $S$, $v^1 = v_L^1$.

\textbf{Inductive step:}
Suppose the inductive hypothesis holds for $k-1$ ($k > 1$). 
The \func{NextNode} computation from which $O$ obtains $a^k$ depends only on $v^1, ..., v^{k-1}$ and the immutable fields of nodes $a^1, ..., a^{k-1}$.
Similarly, the \func{NextNode} computation from which $O_L$ obtains $a_L^k$ depends only on $v_L^1, ..., v_L^{k-1}$ and the immutable fields of nodes $a_L^1, ..., a_L^{k-1}$.
Thus, by the inductive hypothesis, $a^k = a_L^k$.
Since each \llt\ performed by $O$ returns the same value as it would if it were performed atomically at $S$, $v^k = v_L^k$.
Therefore, the inductive hypothesis holds for $k$, and the claim is proved.
\end{proof}

\begin{lem} \label{lem-dotreeupdate-rec-cannot-be-added-after-removal}
After a \rec\ $r$ is removed from the data structure, it cannot be added back into the data structure.
\end{lem}
\begin{proof}
Suppose $r$ is removed from the data structure.
The only thing that can add $r$ back into the data structure is a linearized invocation $S$ of \sct$(V, R, fld, new)$.
By Constraint~\ref{con-dotreeup-exclusively-does-sct}, such an \sct\ must occur in a tree update operation $O$.
By Lemma~\ref{lem-dotreeup-constraints-invariants}.\ref{claim-dotreeup-finalized-before-removed}, every \rec\ added by $S$ is in $O$'s $N$ set.
By Lemma~\ref{lem-dotreeup-constraints-invariants}.\ref{claim-dotreeup-G-in-data-structure}, no node in $N$ is in the data structure at any time before $S$ is linearized.
Thus, $r$ cannot be added to the data structure by $S$.
\end{proof}

Recall that a process $p$ \textit{sets up} an invocation $S$ of
\sct$(V, R, fld, new)$ if it performs a \llt$(r)$ for each $r \in V$, which each return a snapshot, and then invokes $S$.
In the companion paper \cite{paper1}, we make a general assumption that there is a bound on the number of times that any process will perform an \llt$(r)$ that returns \finalized, for any \rec\ $r$.
Under this assumption, we prove that
if a process sets up 
invocations of \sct \ infinitely often, then invocations of \sct \ will succeed infinitely often.
We would like to use this fact to prove progress.
Thus, we must show that our algorithm respects this assumption.


%
%

\begin{lem} \label{lem-dotreeup-only-one-finalized-llt}
No process performs more than one invocation of $\llt(r)$ that returns \finalized, for any $r$, during tree update operations.
\end{lem}
\begin{proof}
Fix any \rec\ $r$.
Suppose, to derive a contradiction, that a process $p$ performs two different invocations $L$ and $L'$ of $\llt(r)$ that return \finalized, during tree update operations.
Then, since each tree update operation returns \fail\ immediately after performing an $\llt(r)$ that returns \finalized, $L$ and $L'$ must occur in different tree update operations $O$ and $O'$.
Without loss of generality, suppose $O$ occurs before $O'$.
Since $L$ returns \finalized, it occurs after an invocation $S$ of \sct$(\sigma, R, fld, new)$ with $r \in R$ (by the semantics of \sct).
By Lemma~\ref{lem-dotreeup-constraints-invariants}.\ref{claim-dotreeup-finalized-before-removed}, $r$ is removed from the data structure by $S$.
By Lemma~\ref{lem-dotreeupdate-rec-cannot-be-added-after-removal}, $r$ cannot be added back into the data structure.
Thus, $r$ is not in the data structure at any time after $S$ occurs.
By the precondition of the tree update template, the argument $top$ to $O'$ was obtained by following child pointers from the $root$ entry point since $O$.
By Observation~\ref{obs-dotreeup-satisfies-con-mark-all-removed-recs}, Constraint~\conmarkallremovedrecs\ is satisfied.
Thus, Lemma~\ref{lem-if-rec-traversed-then-rec-in-data-structure} implies that $top$ is in the data structure at some point between $O$ and $O'$.
This implies that $r \neq top$.
Since $O'$ performs $L'$, and $r \neq top$, $O'$ must have obtained $r$ from the return value of an invocation of \llt$(r')$ that $O'$ performed before $L'$.
%
By Lemma~\ref{lem-rec-in-data-structure-just-before-llt}, $r'$ must be in the data structure just before $L''$ is linearized.
Since $L''$ returns $r$, $r'$ also points to $r$ just before $L''$ is linearized.
Therefore, $r$ is in the data structure just before $L''$, which contradicts the fact that $r$ is not in the data structure at any time after $S$ occurs.
\end{proof}

\begin{lem} \label{lem-dotreeup-wait-free}
Tree update operations are wait-free.
\end{lem}
\begin{proof}
A tree update operation consists of a loop in which \llt, \func{Condition} and \func{NextNode} are invoked (see Fig.~\ref{code-dotreeupdate}), an invocation of \sct, and some other constant, local work.
The implementations of \llt\ and \sct\ given in \cite{paper1} are wait-free.
Similarly, \func{NextNode} and \func{Condition} must perform finite computation, and \func{Condition} must eventually return \true\ in every tree update operation, causing the operation to exit the loop.
\end{proof}

\begin{thm} \label{thm-dotreeup-progress}
If tree update operations are performed infinitely often, then tree update operations succeed infinitely often.
\end{thm}
\begin{proof}
Suppose, to derive a contradiction, that tree update operations are performed infinitely often but, after some time $t_0$, no tree update operation succeeds.
Then, after some time $t_1 \ge t_0$, no tree update operation is successful.
Since a tree update operation is successful if and only if it performs a linearized \sct, no invocation of \sct\ is linearized after $t_1$.
Therefore, after $t_1$, the data structure does not change, so only a finite number of \rec s are ever initiated in the execution.
By Lemma~\ref{lem-dotreeup-only-one-finalized-llt}, after some time $t_2 \ge t_1$, no invocation of \llt\ performed by a tree update operation will return \finalized.
Consider any tree update operation $O$ (see Figure~\ref{code-dotreeupdate}).
First, \llt s are performed on a sequence of \rec s.
If these \llt s all return values different from \fail\ or \finalized, then an invocation of \sct\ is performed.
If this invocation of \sct\ is successful, then $O$ will be successful. 
Since tree update operations are performed infinitely often after $t_2$, Definition~\ref{defn-set-up-sct} implies that invocations of \sct\ are set up infinitely often.
Thus, invocations of \sct\ succeed infinitely often.
Finally, Constraint~\ref{con-dotreeup-exclusively-does-sct} implies that tree update operations succeed infinitely often, which is a contradiction.
\end{proof}


\section{Chromatic search trees}
\label{chromatic}

\subsection{Pseudocode}
\label{sec-pseudocode}

\begin{wrapfigure}{R}{0.4\textwidth}
\vspace{-2mm}
\raisebox{8.5mm}{\includegraphics[scale=.9]{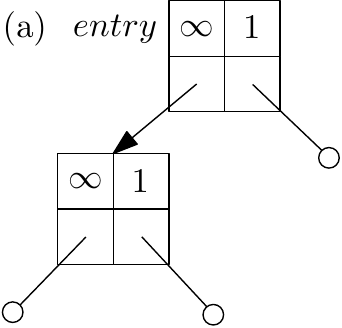}}\hspace*{3mm}
\includegraphics[scale=.9]{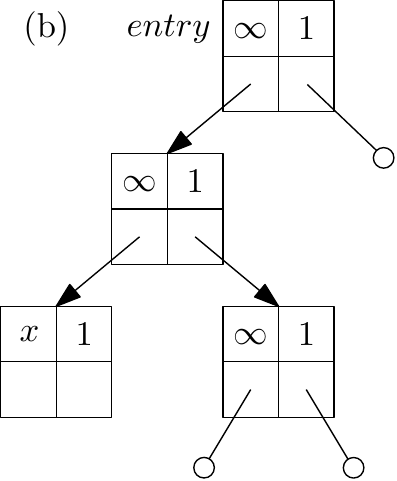}
\caption{(a) initial tree; (b) non-empty tree.
Weights are on the right.}
\label{fig-treetop}
\end{wrapfigure}

To avoid special cases when applying operations near the root, we add sentinel nodes with a special key $\infty$ that is larger than any key  in the dictionary.
When the dictionary is empty, it looks like Fig.~\ref{fig-treetop}(a).
When non-empty, it looks like Fig.~\ref{fig-treetop}(b), with all dictionary keys stored in
a chromatic tree rooted at $x$.
The nodes shown in Fig.~\ref{fig-treetop} always have weight one.

Pseudocode for the dictionary operations on chromatic search trees
are given in Figure~\ref{code-chromatic1} and \ref{code-chromatic2}.
As described in Section~\ref{sec-chromatic}, 
a search for a key simply uses reads of child pointers to locate a leaf, as in an
ordinary (non-concurrent) BST.
\ins\ and \del\ search for the location to perform the update and then call \tryins\ or \trydel,
which each follow the tree update template to swing a pointer that accomplishes the update.  
This is repeated until the update is successful, and the
update then calls \cleanup, if necessary, to remove the violation the update created.
The tree transformations that accomplish the updates are the first three shown in Figure~\ref{fig-chromatic-rotations}.
\func{Insert1} adds a key that was not previously present in the tree; one of the two new leaves contains the new key, and the other contains the key (and value) that were stored in $\ux$.
\func{Insert2} updates the value associated with a key that is already present in the tree by
creating a new copy of the leaf containing the key.
The \func{Delete} transformation
deletes the key stored in the left child of $\ux$.  (There is a symmetric version
for deleting keys stored in the right child of a node.)

\begin{figure}[tbp]
\def\figsize{1}
\centering
%
\begin{tabularx}{\textwidth}{YY}
\includegraphics[scale=\figsize]{figures/INS1.pdf} & \hspace{-0.4cm}\includegraphics[scale=\figsize]{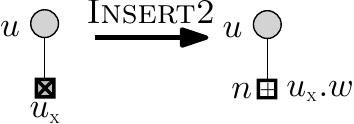} \\
\hspace{-0.2cm}\includegraphics[scale=\figsize]{figures/DEL.pdf} & \hspace{-0.4cm}\includegraphics[scale=\figsize]{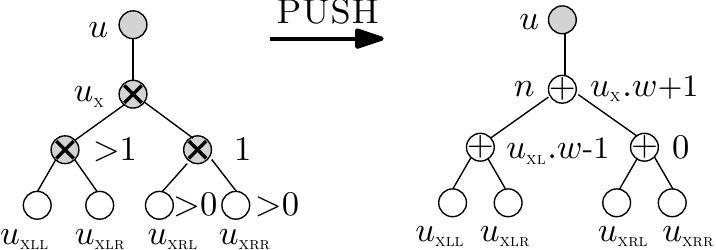} \\
\hspace{-1cm}\includegraphics[scale=\figsize]{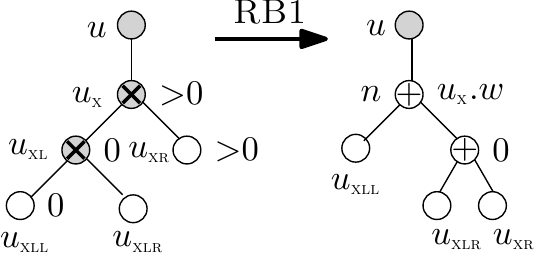} & \hspace{-1.6cm}\includegraphics[scale=\figsize]{figures/RB2.pdf} \\
\hspace{-1.2cm}\includegraphics[scale=\figsize]{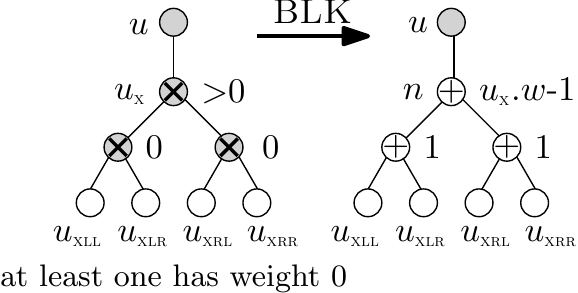} & \hspace{-0.4cm}\includegraphics[scale=\figsize]{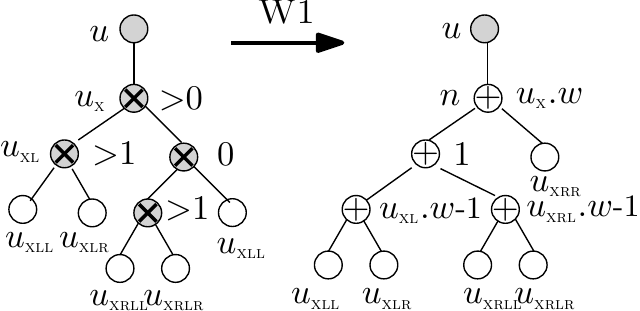} \\
\hspace{-0.4cm}\includegraphics[scale=\figsize]{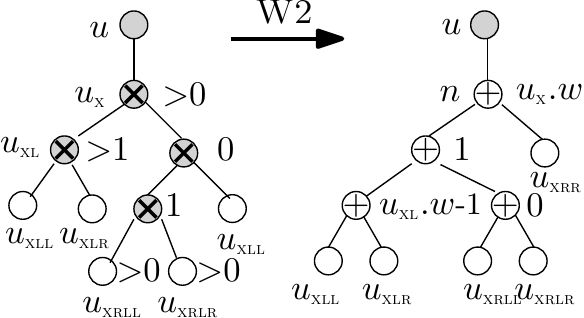} & \hspace{-0.8cm}\includegraphics[scale=\figsize]{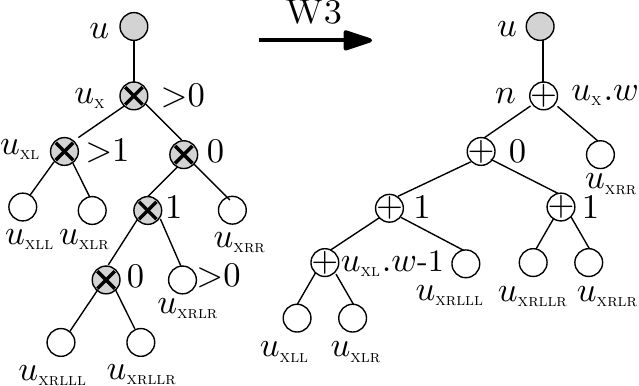} \\
\includegraphics[scale=\figsize]{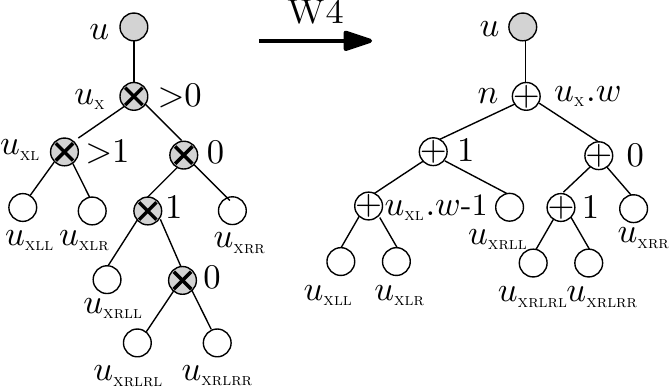} & \hspace{-1cm}\includegraphics[scale=\figsize]{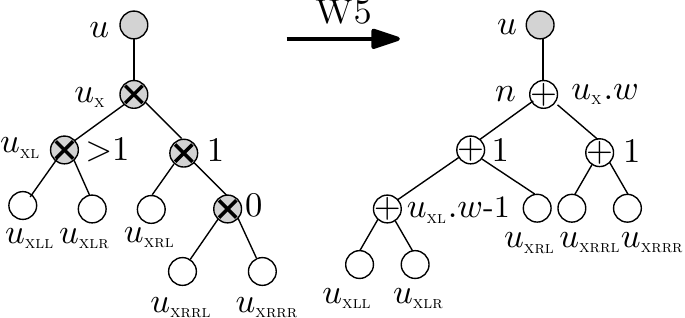} \\
\hspace{-0.5cm}\includegraphics[scale=\figsize]{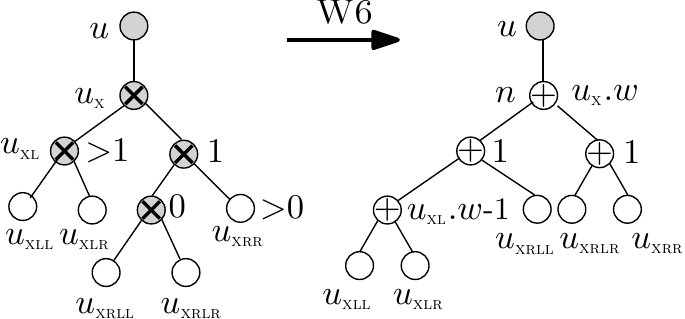} & \hspace{0.2cm}\includegraphics[scale=\figsize]{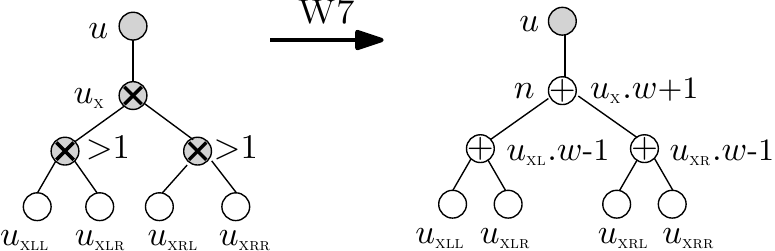} \\
\end{tabularx}

\caption{Transformations for chromatic search trees.  Each transformation also has a mirror image.}
\label{fig-chromatic-rotations}
\end{figure}


\begin{figure}[tbp]
\newcommand{\wcnarrow}[2]{\parbox{\namewidth}{#1} \com \mbox{#2}}
\hspace*{-7mm}
\def\namewidth{18mm}
\preplisting
\begin{lstlisting}[mathescape=true,style=nonumbers]
    type// \node
        //\com User-defined fields
        //\wcnarrow{$left, right$}{child pointers (mutable)}
        //\wcnarrow{$k, v, w$}{key, value, weight (immutable)}
        //\com Fields used by \llt/\sct\ algorithm
        //\wcnarrow{$\info$}{pointer to \op}
        //\wcnarrow{$marked$}{Boolean}
%\end{lstlisting}

\prepnewlisting
\hrule
\vspace{-2mm}
\begin{lstlisting}[mathescape=true]
    //\func{Get}$(key)$
      //\com Returns the value associated with $key$, or $\bot$ if no value is associated with $key$
      do// a standard BST search for $key$ using reads, ending at a leaf $l$
      if $l.k = key$ then return $l.v$
      else return $\bot$// \\ \hrule %
    
    //\ins$(key, value)$
      //\com Associates $value$ with $key$ in the dictionary and returns the old associated value, or $\bot$ if none existed %Replaces $\langle key, old \rangle$ with $\langle key, value \rangle$ in the dictionary, and returns $old$, or $\bot$ if $key$ was not in the dictionary
      
      do
        do// a standard BST search for $key$ using reads, ending at a leaf $l$ with parent $p$
        $result := \tryins(p, l, key, value)$
      while $result = \fail$
      $\langle createdViolation, value \rangle := result$
      if $createdViolation$ then $\cleanup(key)$
      return $value$// \\ \hrule %
    
    //\tryins$(p, l, key, value)$ 
      //\tline{\com Returns $\langle \true, \bot \rangle$ if $key$ was not in the dictionary %
               and inserting it caused a violation,} %
              {$\langle \false, \bot \rangle$ if $key$ was not in the dictionary %
               and inserting it did not cause a violation,} %
              {$\langle \false, oldValue \rangle$ if $\langle key, oldValue \rangle$ was in the dictionary, and %} %
              %{
              \fail\ if we should try again} \vspace{2mm}%

      if $(result := \llt(p)) \in \{\fail, \finalized\}$ then return $\fail$ else $\langle p_{L}, p_{R} \rangle := result$ 
      if $p_L = l$ then $ptr := \&p.left$
      else if $p_R = l$ then $ptr := \&p.right$
      else return $ \fail $ //\vspace{2mm}%
      
      if $\llt(l) \in \{\fail, \finalized\}$ then return $\fail$ //\vspace{2mm}%

      $newLeaf :=$ //pointer to a new \node$\langle \nil, \nil, key, value, 1 \rangle$
      if $l.k = key$ then
        $oldValue := l.v$
        $new := newLeaf$
      else
        $oldValue := \bot$
        if $l \mbox{ is a sentinel node}$ then $newWeight := 1$ else $newWeight := l.w - 1$
        if $key < l.k$ then $new :=$ //pointer to a new \node%
                                       $\langle newLeaf, l, l.k, \nil, newWeight \rangle$
        else $new :=$ //pointer to a new \node%
                        $\langle l, newLeaf, key, \nil, newWeight \rangle$ \vspace{2mm}%

      if $\sct(\langle p, l \rangle, \langle l \rangle, ptr, new)$ then
        return $\langle (new.w = p.w = 0), oldValue \rangle$
      else return $\fail$
\end{lstlisting}
	\caption{Data structure, and pseudocode for \func{Get}, \func{Insert} and \tryins\ (which follows the tree update template).}
	\label{code-chromatic1}
\end{figure}

\begin{figure}[tbp]
\preplisting
\hrule
\vspace{-2mm}
\begin{lstlisting}[mathescape=true]
    //\del$(key)$ 
      //\com Deletes $key$ and returns its associated value, or returns $\bot$ if $key$ was not in the dictionary
      do
        do// a standard BST search for $key$ using reads, ending at a leaf $l$ with parent $p$ and grandparent $gp$ \label{del-search-line}
        $result := \trydel(gp, p, l, key)$
      while $result = \fail$
      $\langle createdViolation, value \rangle := result$
      if $createdViolation$ then $\cleanup(key)$
      return $value$// \\ \hrule %
    
    //\trydel$(gp, p, l, key)$ 
      //\tline{\com Returns $\langle \true, value \rangle$ if $\langle key, value \rangle$ was in the dictionary %
               and deleting $key$ caused a violation,} %
              {$\langle \false, value \rangle$ if $\langle key, value \rangle$ was in the dictionary %
               and deleting $key$ did not cause a violation,} %
              {$\langle \false, \bot \rangle$ if $key$ was not in the dictionary, and %} %
              %{
              \fail\ if we should try again} \vspace{2mm}%
              
      if $l.k \neq key$ then return $\langle \false, \bot \rangle$ //\vspace{2mm}%
      
      if $(result := \llt(gp)) \in \{\fail, \finalized\}$ then return $\fail$ else $\langle gp_{L}, gp_{R} \rangle := result$
      if $gp_L = p$ then $ptr := \&gp.left$
      else if $gp_R = p$ then $ptr := \&gp.right$
      else return $\fail$ //\vspace{2mm}%

      if $(result := \llt(p)) \in \{\fail, \finalized\}$ then return $\fail$ else $\langle p_{L}, p_{R} \rangle := result$
      if $p_L = l$ then $s := p_R$
      else if $p_R = l$ then $s := p_L$
      else return $\fail$ //\vspace{2mm}%
      
      if $\llt(l) \in \{\fail, \finalized\}$ then return $\fail$
      if $\llt(s) \in \{\fail, \finalized\}$ then return $\fail$ //\vspace{2mm}%

      if $p \mbox{ is a sentinel node }$ then $newWeight := 1$ else $newWeight := p.w + s.w$
      if $\sct(\langle gp, p, l \rangle, \langle p, l \rangle, ptr, \mbox{new copy of } s \mbox{ with weight } newWeight)$ then
        return $\langle (newWeight > 1), l.v \rangle$
      else return $\fail$// \\ \hrule %

    //$\func{Successor}(key)$
      //\com Returns the successor of $key$ and its associated value (or $\langle \bot, \bot\rangle$ if there is no such successor)
      $l := root$
      loop until $l$// is a leaf
        if $\llt(l) \in \{\fail, \finalized\}$ then $\mbox{retry }\func{Successor}(key)\mbox{ from scratch}$
        if $key < l.key$ then
          $lastLeft := l$
          $l := l.left$
          $V := \langle lastLeft \rangle$
        else
          $l := l.right$
          //add $l$ to end of $V$ \vspace{2mm}%
      
      if $lastLeft = root$ then return $\langle \bot, \bot \rangle$ //\hfill \com Dictionary is empty \label{succ-return-empty}
      else if $key < l.k$ then return $\langle l.k, l.v \rangle$ //\label{succ-return-l}
      else //\hfill \com Find next leaf after $l$ in in-order traversal
        $succ := lastLeft.right$
        loop until $succ$// is a leaf
          if $\llt(succ) \in \{\fail, \finalized\}$ then $\mbox{retry }\func{Successor}(key)\mbox{ from scratch}$
          //add $succ$ to end of $V$
          $succ := succ.left$
        if $succ.key = \infty$ then $result:=\langle \bot,\bot\rangle$ else $result:=\langle succ.k,succ.v\rangle$
        if $\vlt(V)$ then return $result$ //\label{succ-return-succ}
        else $\mbox{retry }\func{Successor}(key)\mbox{ from scratch}$
\end{lstlisting}
\vspace*{-3mm}
	\caption{Code for \del, \trydel\ (which follows the tree update template), and \func{Successor}.}
	\label{code-chromatic2}
\end{figure}

\begin{figure}[tbp]
\centering
\includegraphics[scale=0.42]{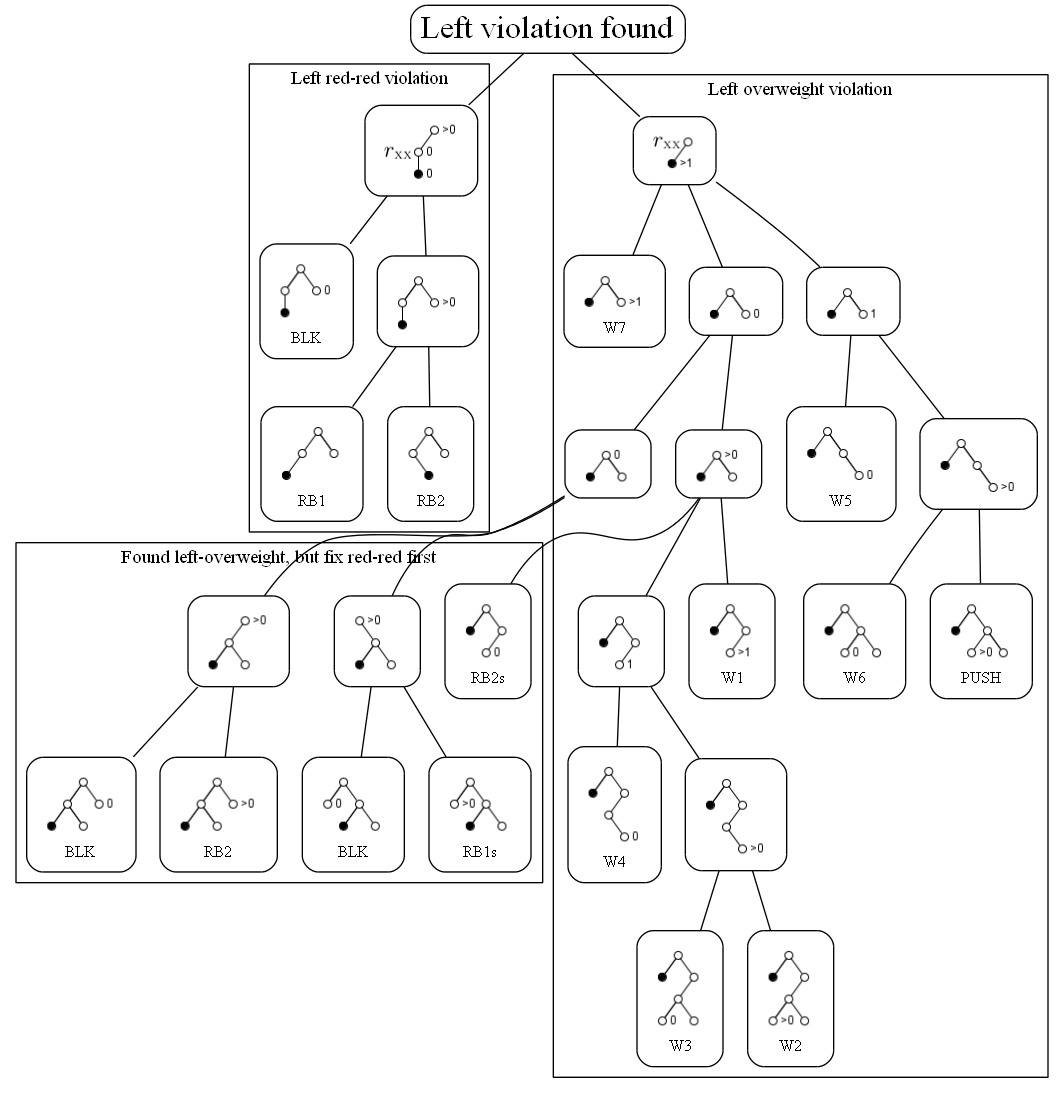}
\caption{Decision tree used by the algorithm to determine which rebalancing operation to apply when a violation is encountered at a node (highlighted in black). 
The corresponding diagram to cover right violations can be obtained by: horizontally flipping each miniature tree diagram, changing each rebalancing step \func{DoX}$(\cdot)$ to its symmetric version \func{DoXs}$(\cdot)$, and changing each symmetric rebalancing step \func{DoXs}$(\cdot)$ back to its original version \func{DoX}$(\cdot)$.}
\label{fig-decision-tree}
\end{figure}

\begin{figure}[tbp]
\preplisting
\hrule
\vspace{-2mm}
\begin{lstlisting}[mathescape=true]
    //\cleanup$(key)$
      //\com Ensures the violation created by an \ins\ or \del\ of $key$ gets eliminated
      while $\true$
        $ggp := \nil$;  $gp:=\nil$;  $p:=\nil$;  $l := root$ //\hfill\com Save four last nodes traversed\label{cleanup-start}
        while //\true 
          if //$l$ is a leaf then \textbf{return} \hfill \com Arrived at leaf without finding a violation\label{cleanup-terminate}
          if $key < l.key$ then {$ggp := gp$;  $gp := p$;  $p := l$; $l:=l.left$} //\label{move-l-left}
          else {$ggp := gp$;  $gp := p$;  $p := l$;  $l:=l.right$} //\label{move-l-right}
          if $l.w > 1$ or ($p.w = 0$ and $l.w = 0$) then //\hfill \com Found a violation\label{find-violation}
            $\tryrebalance(ggp,gp,p,l)$ //\hfill \com Try to fix it
            exit loop //\hfill \com Go back to $root$ and traverse again\label{cleanup-end} \\ \hrule %

    //\tryrebalance$(ggp, gp, p, l)$
    //\com Precondition: $l.w > 1$ or $l.w = p.w = 0 \neq gp.w$%, } %
            {%     $l$ has been a child of $p$, $p$ has been a child of $gp$, and %
             %     $gp$ has been a child of $ggp$}
             
      $r := ggp$
      if $(result := \llt(r)) \in \{\fail, \finalized\}$ then return else $\langle \rl, \rr \rangle := result$
        
      $\rx := gp$
      if $\rx \notin\{\rl, \rr\}$ then return
      if $(result := \llt(\rx)) \in \{\fail, \finalized\}$ then return else $\langle \rxl, \rxr \rangle := result$

      $\rxx := p$
      if $\rxx \notin\{\rxl, \rxr\}$ then return
      if $(result := \llt(\rxx)) \in \{\fail, \finalized\}$ then return else $\langle \rxxl, \rxxr \rangle := result$
        
      if $l.w >1 $ then        //\tabto{8cm}\com Overweight violation at $l$
        if $l = \rxxl$ then //\tabto{8cm}\com Left overweight violation ($l$ is a left child)
          if $(result := \llt(\rxxl)) \in \{\fail, \finalized\}$ then return
          //$\func{OverweightLeft}($all $r$ variables$)$
        else if $l = \rxxr $ then //\tabto{8cm}\com Right overweight violation ($l$ is a right child)
          if $(result := \llt(\rxxr)) \in \{\fail, \finalized\}$ then return
          //$\func{OverweightRight}($all $r$ variables$)$
      else        //\tabto{8cm}\com Red-red violation at $l$
        if $\rxx = \rxl$ then //\tabto{8cm}\com Left red-red violation ($p$ is a left child)
          if $\rxr.w = 0$ then
            if $(result := \llt(\rxr)) \in \{\fail, \finalized\}$ then return else $\langle \rxrl, \rxrr \rangle := result$
            //$\func{DoBLK}(\langle r,\rx,\rxx,\rxr \rangle,$ all $r$ variables$)$
          else if $l = \rxxl$ then return //$\func{DoRB1}(\langle r,\rx,\rxx\rangle,$ all $r$ variables$)$
          else if $l = \rxxr$ then
            if $(result := \llt(\rxxr)) \in \{\fail, \finalized\}$ then return
            //$\func{DoRB2}(\langle r,\rx,\rxx,\rxxr \rangle,$ all $r$ variables$)$
        else //\com $\rxx = \rxr $ \tabto{8cm}\com Right red-red violation ($p$ is a right child)
          if $\rxl.w = 0$ then
            if $(result := \llt(\rxl)) \in \{\fail, \finalized\}$ then return else $\langle \rxll, \rxlr \rangle := result$
            //$\func{DoBLK}(\langle r,\rx,\rxl,\rxx \rangle,$ all $r$ variables$)$
          else if $l = \rxxr$ then return //$\func{DoRB1s}(\langle r,\rx,\rxx \rangle,$ all $r$ variables$)$
          else if $l = \rxxl$ then
            if $(result := \llt(\rxxl)) \in \{\fail, \finalized\}$ then return
            //$\func{DoRB2s}(\langle r,\rx,\rxx,\rxxl \rangle,$ all $r$ variables$)$
\end{lstlisting}
	\caption{Pseudocode for \tryrebalance\ (which follows the tree update template) and \cleanup.
	}
	\label{code-chromatic-tryrebalance}
\end{figure}

\begin{figure}[tbp]
\hspace{-2mm}
\begin{minipage}{1.02\textwidth}
\preplisting
\hrule
\vspace{-2mm}
\begin{lstlisting}[mathescape=true]
    //\func{OverweightLeft}$(r,\rx,\rxx,\rxxl,\rl,\rr,\rxl,\rxr,\rxxr)$
      if $\rxxr.w = 0 $ then
        if $\rxx.w = 0 $ then
          if $\rxx = \rxl$ then
            if $\rxr.w = 0$ then
              if $(result := \llt(\rxr)) \in \{\fail, \finalized\}$ then return else $\langle \rxrl, \rxrr \rangle := result$
              //$\func{DoBLK}(\langle r,\rx,\rxx,\rxr\rangle,$ all $r$ variables$)$
            else //\com $\rxr.w >0 $
              if $(result := \llt(\rxxr)) \in \{\fail, \finalized\}$ then return else $\langle \rxxrl, \rxxrr \rangle := result$
              //$\func{DoRB2}(\langle r,\rx,\rxx,\rxxr\rangle,$ all $r$ variables$)$
          else //\com $\rxx = \rxr $
            if $\rxl.w = 0$ then
              if $(result := \llt(\rxl)) \in \{\fail, \finalized\}$ then return else $\langle \rxll, \rxlr \rangle := result$
              //$\func{DoBLK}(\langle r,\rx,\rxl,\rxx\rangle,$ all $r$ variables$)$
            else //$\func{DoRB1s}(\langle r,\rx,\rxx\rangle,$ all $r$ variables$)$
        else //\com $\rxx.w >0 $
          if $(result := \llt(\rxxr)) \in \{\fail, \finalized\}$ then return else $\langle \rxxrl, \rxxrr \rangle := result$
          if $(result := \llt(\rxxrl)) \in \{\fail, \finalized\}$ then return
          if $\rxxrl.w > 1$ then //$\func{DoW1}(\langle\rx,\rxx,\rxxl,\rxxr,\rxxrl\rangle,$ result, all $r$ variables$)$
          else if $\rxxrl.w = 0$ then
            $\langle \rxxrll, \rxxrlr \rangle := result$
            //$\func{DoRB2s}(\langle\rx,\rxx,\rxxr,\rxxrl\rangle,$ all $r$ variables$)$
          else //\com $\rxxrl.w = 1$
            $\langle \rxxrll, \rxxrlr \rangle := result$
            if $\rxxrlr = \nil$ then return//\label{overweightleft-check-nil-rxxrlr} \tabto{10cm}\com a node we performed \llt\ on was modified
            if $\rxxrlr.w = 0$ then
              if $(res := \llt(\rxxrlr)) \in \{\fail, \finalized\}$ then return else $\langle \rxxrlrl, \rxxrlrr \rangle := res$//\label{overweightleft-bad-llt-rxxrlr}
              //$\func{DoW4}(\langle\rx,\rxx,\rxxl,\rxxr,\rxxrl,\rxxrlr\rangle,$ all $r$ variables$)$
            else //\com $\rxxrlr.w > 0$
              if $\rxxrll.w = 0$ then
                if $(res := \llt(\rxxrll)) \in \{\fail, \finalized\}$ then return else $\langle \rxxrlll, \rxxrllr \rangle := res$//\label{overweightleft-bad-llt-rxxrll}
                //$\func{DoW3}(\langle\rx,\rxx,\rxxl,\rxxr,\rxxrl,\rxxrll\rangle,$ all $r$ variables$)$
              else //$\func{DoW2}(\langle\rx,\rxx,\rxxl,\rxxr,\rxxrl\rangle,$ all $r$ variables$)$
      else if $\rxxr.w = 1 $ then
        if $(result := \llt(\rxxr)) \in \{\fail, \finalized\}$ then return else $\langle \rxxrl, \rxxrr \rangle := result$
        if $\rxxrr = \nil$ then return//\label{overweightleft-check-nil-rxxrr} \tabto{10cm}\com a node we performed \llt\ on was modified
        if $\rxxrr.w = 0 $ then
          if $(result := \llt(\rxxrr)) \in \{\fail, \finalized\}$ then return else $\langle \rxxrrl, \rxxrrr \rangle := result$//\label{overweightleft-bad-llt-rxxrr}
          //$\func{DoW5}(\langle\rx,\rxx,\rxxl,\rxxr,\rxxrr\rangle,$ all $r$ variables$)$
        else if $\rxxrl.w = 0 $ then
          if $(result := \llt(\rxxrl)) \in \{\fail, \finalized\}$ then return else $\langle \rxxrll, \rxxrlr \rangle := result$//\label{overweightleft-bad-llt-rxxrl}
          //$\func{DoW6}(\langle\rx,\rxx,\rxxl,\rxxr,\rxxrl\rangle,$ all $r$ variables$)$
        else //$\func{DoPush}(\langle\rx,\rxx,\rxxl,\rxxr\rangle,$ all $r$ variables$)$
      else
        if $(result := \llt(\rxxr)) \in \{\fail, \finalized\}$ then return else $\langle \rxxrl, \rxxrr \rangle := result$
        //$\func{DoW7}(\langle\rx,\rxx,\rxxl,\rxxr\rangle,$ all $r$ variables$)$ \\ \hrule %

    //\func{OverweightRight}$(r,\rx,\rxx,\rxxr,\rl,\rr,\rxl,\rxr,\rxxl)$
    //\com Obtained from \func{OverweightLeft} by flipping each $R$ in the subscript of an $r$ variable to an $L$ (and vice versa), and by flipping each rebalancing step \func{DoX}$(\cdot)$ to its symmetric version \func{DoXs}$(\cdot)$ (and vice versa).
\end{lstlisting}
\end{minipage}
	\caption{Pseudocode for \func{OverweightLeft}.}
	\label{code-chromatic-overweight-left}
\end{figure}

\begin{figure}[tb]
\begin{minipage}{0.38\textwidth}
\vspace{-1.6cm}
\hspace{-1.6cm}
\includegraphics[scale=1]{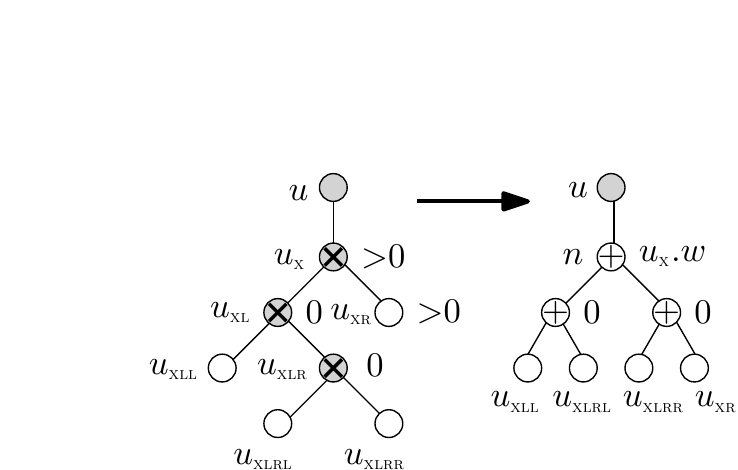}
\end{minipage}
\begin{minipage}{0.62\textwidth}
\preplisting
\hrule
\vspace{-2mm}
\begin{lstlisting}[mathescape=true]
    //$\func{DoRB2}(u,\ux,\uxl,\uxr,\uxll,\uxlr,\uxlrl,\uxlrr)$
      //\com Create new nodes according to the right-hand diagram
      //create node $\nL$ with $k=\uxl.k$, $w=0$, {\it left} $ = \uxll$, $right=\uxlrl$
      //create node $\nR$ with $k=\ux.k$, $w=0$, {\it left} $=\uxlrr$, $right=\uxr$
      //create node $n$ with $k=\uxlr.k$, $w=\ux.w$, {\it left} $=\nL$, $right=\nR$
      //\com Perform the \sct\ to swing the child pointer of node $u$
      if $\ux = \ul$ then $ptr := \&u.left$ else $ptr := \&u.right$
      $\sct(\langle u,\ux,\uxx,\uxlr \rangle, \langle \ux,\uxl,\uxlr \rangle, ptr, n)$ // \vspace{2mm} \hrule %
\end{lstlisting}
\end{minipage}
	\caption{
	Implementing rebalancing step \func{RB2}.
	Other rebalancing steps are handled similarly using the diagrams shown in Figure~\ref{fig-chromatic-rotations}.
	}
	\label{code-chromatic-dorb2}
\end{figure}

The \cleanup\ repeatedly searches for the key that was inserted or deleted,
fixing any violation it encounters, until it can traverse the entire path from
$root$ to a leaf without seeing any violations.
When a violation is found, it calls \tryrebalance, which uses the decision tree
given in Figure \ref{fig-decision-tree} to decide which rebalancing step should
be applied to fix the violation.
At each node of the decision tree,
the process decides which child to proceed to by looking at the weight of one node,
as indicated in the child.  The leaves of the decision tree are labelled by the 
rebalancing operation to apply.  
The code to implement the decision tree is given in Figure \ref{code-chromatic-tryrebalance} 
and \ref{code-chromatic-overweight-left}.
Note that this decision tree is a component of the sequential chromatic tree algorithm that was left to the implementer in \cite{Boyar97amortizationresults}.

The rebalancing steps, which are shown in Figure~\ref{fig-chromatic-rotations}, are a slight modification of those in \cite{Boyar97amortizationresults}.\footnote{Specifically, we do not allow \func{W1}, \func{W2}, \func{W3} or \func{W4} to be applied when the node labeled $\ux$ has weight 0.  Under this restriction, this set of rebalancing steps has the desirable property that when a violation moves,  it remains on the search path to the key whose insertion or deletion originally caused the violation.  It is easy to verify that an alternative rebalancing step can always be performed when $\ux.w = 0$, so this modification does not affect the chromatic tree's convergence to a RBT.}
Each also has a symmetric mirror-image version, denoted by an S after the name, except \func{BLK}, which is its own mirror image.
We use a simple naming scheme for the nodes in the diagram.
Consider the node $\ux$.
We denote its left child by $\uxl$, and its right child by $\uxr$.
Similarly, we denote the left child of $\uxl$ by $\uxll$, and so on.
(The subscript {\sc x} indicates that we do not care whether it is a left or right child.)
For each transformation shown in Figure \ref{fig-chromatic-rotations}, 
the transformation is achieved by an \sct\ 
that swings a child pointer of $u$ and depends on  \llt s of all of the shaded nodes.
The nodes marked with $\times$ are finalized (and removed from the data structure).
The nodes marked by a $+$ are newly created nodes.
The nodes with no marking may be internal nodes, leaves or \nil.
Weights of all newly created nodes are shown.
The keys stored in newly created nodes are the same as in the removed nodes (so that an in-order
traversal encounters them in the same order).
Figure~\ref{code-chromatic-dorb2} implements
one of the rebalancing steps.  The others can be generated from their diagrams in a similar way.

The \func{Successor}($key$) function uses an ordinary BST search for $key$ to find a leaf.
If this leaf's key is bigger than $key$, it is returned.
Otherwise, \func{Successor}
finds the leaf that would be reached next by an in-order traversal of the BST
and then performs a \vlt\
that verifies the path connecting the two leaves has not changed.
The \func{Predecessor} function can be implemented similarly.
 
 
\subsection{Correctness of chromatic trees}
\label{chromatic-correctness}

\after{
(1) Proving successor's calls to LLX/VLX satisfy their preconditions
(2) Proving if an insert or delete creates a violation, then it performs cleanup
(3) Describing how to check whether a node is a sentinel (to check node r, evaluate (r.key = \infty or r = root.left.left))
(4) Proving that the way we check whether something is a sentinel plays nicely with the rest of the algorithm
}




The following lemma proves that \tryins, \trydel\ and \tryrebalance\ follow the tree update template, and that \sct s are performed only by tree update operations, so that we can invoke results from Appendix~\ref{sec-dotreeup-correctness} to argue that the transformations in Fig.~\ref{fig-chromatic-rotations} are performed atomically.
Our proof that \tryrebalance\ follows the tree update template is complicated slightly by the fact that its subroutine, \func{OverweightLeft} (or \func{OverweightRight}), can effectively follow the template for a while, and then return early at line~\ref{overweightleft-check-nil-rxxrlr} or line~\ref{overweightleft-check-nil-rxxrr} if it sees a \nil\ child pointer.
(We later prove this can happen only if a tree update has changed a node since we last performed \llt\ on it, so that our \sct\ would be unsuccessful, anyway.)
We thus divide the invocations of \tryrebalance\ into those that follow the template, and those that follow the template until just before they return (at one of these two lines), and show that the latter do not invoke \sct\ (which we must show because they are not technically tree update operations).
We prove these claims together inductively with an invariant that the top of the tree is as shown in Fig.~\ref{fig-treetop}.

\begin{lem} \label{lem-chromatic-invariants}
Our implementation of a chromatic search tree
satisfies the following.
\begin{enumerate}
\item \tryins\ and \trydel\ follow the tree update template and satisfy all constraints specified by the template.
If an invocation of \tryrebalance\ does not return at line~\ref{overweightleft-check-nil-rxxrlr} or line~\ref{overweightleft-check-nil-rxxrr}, then it follows the tree update template and satisfies all constraints specified by the template. 
Otherwise, it follows the tree update template up until it returns without performing an \sct, and it satisfies all constraints specified by the template.
\label{claim-chromatic-invariants-follow-template}
\item The tree rooted at $root$ always looks like Fig.~\ref{fig-treetop}(a) if it is empty, and Fig.~\ref{fig-treetop}(b) otherwise.
\label{claim-chromatic-invariants-top-of-tree}
\end{enumerate}
\end{lem}
\begin{proof}
We proceed by induction on the sequence of steps in the execution.

Claim~\ref{claim-chromatic-invariants-follow-template} follows almost immediately from inspection of the code.
The only non-trivial part of the proof is showing that these algorithms never invoke $\llt(r)$ where $r = \nil$, and the only step that can affect this sub-claim is an invocation of \llt.
Suppose the inductive hypothesis holds just before an invocation $\llt(r)$.
For \ins 1, \ins 2\ and \del, $r \neq \nil$ follows from inspection of the code, inductive Claim~\ref{claim-chromatic-invariants-top-of-tree}, and the fact that every key inserted or deleted from the dictionary is less than $\infty$ (so every key inserted or deleted minimally has a parent and a grandparent).
For rebalancing steps, $r \neq \nil$ follows from inspection of the code and the decision tree in Fig.~\ref{fig-decision-tree}, using a few facts about the data structure.
\tryrebalance\ performs \llt s on its arguments $ggp, gp, p, l$, and then on a sequence of nodes reachable from $l$, as it follows the decision tree.
From Fig.~\ref{fig-treetop}(b), it is easy to see that any node with weight $w \neq 1$ minimally has a parent, grandparent, and great-grandparent.
Thus, the arguments to \tryrebalance\ are all non-\nil.
By inspection of the transformations in Fig.~\ref{fig-chromatic-rotations}, each leaf has weight $w \ge 1$, every node has zero or two children, and the child pointers of a leaf do not change. 
This is enough to argue that all \llt s performed by \tryrebalance, and nearly all \llt s performed by \func{OverweightLeft} and \func{OverweightRight}, are passed non-\nil\ arguments.
Without loss of generality, we restrict our attention to \llt s performed by \func{OverweightLeft}.
The argument for \func{OverweightRight} is symmetric.
The only \llt s that require different reasoning are performed at lines~\ref{overweightleft-bad-llt-rxxrlr}, \ref{overweightleft-bad-llt-rxxrll}, \ref{overweightleft-bad-llt-rxxrr} and \ref{overweightleft-bad-llt-rxxrl}.
For lines~\ref{overweightleft-bad-llt-rxxrlr} and \ref{overweightleft-bad-llt-rxxrr}, the claim follows immediately from lines~\ref{overweightleft-check-nil-rxxrlr} and line~\ref{overweightleft-check-nil-rxxrr}, respectively.
Consider line~\ref{overweightleft-bad-llt-rxxrll}.
If $\rxxrll = \nil$ then, since every node has zero or two children, and the child pointers of a leaf do not change, $\rxxrl$ is a leaf, so $\rxxrlr = \nil$.
Therefore, \func{OverweightLeft} will return before it reaches line~\ref{overweightleft-bad-llt-rxxrll}.
By the same argument, $\rxxrl \neq \nil$ when line~\ref{overweightleft-bad-llt-rxxrl} is performed.
Thus, $r \neq \nil$ no matter where the \llt\ occurs in the code.

We now prove Claim~\ref{claim-chromatic-invariants-top-of-tree}.
The only step that can modify the tree is an \sct.
Suppose the inductive hypothesis holds just before an invocation $S$ of \sct.
By inductive Claim~\ref{claim-chromatic-invariants-follow-template}, the algorithm that performed $S$ followed the tree update template up until it performed $S$.
Therefore, Theorem~\ref{thm-dotreeup-linearizable} implies that $S$ atomically performs one of the transformations in Fig.~\ref{fig-chromatic-rotations}.
By inspection of these transformations, when the tree is empty, \ins 1 at the left child of $root$ changes the tree from looking like Fig.~\ref{fig-treetop}(a) to looking like Fig.~\ref{fig-treetop}(b) and, otherwise, does not affect the claim.
When the tree has only one node with $key \neq \infty$, \del\ at the leftmost grandchild of $root$ changes the tree back to looking like Fig.~\ref{fig-treetop}(a) and, otherwise, does not affect the claim.
Clearly, \ins 2 does not affect the claim.
Without loss of generality, let $S$ be a left rebalancing step.
Each rebalancing step in \{\func{BLK}, \func{RB1}, \func{RB2}\} applies only if $\uxl.w = 0$, and every other rebalancing step applies only if $\uxl.w > 1$.
Therefore, $S$ must change a child pointer of a descendent of the left child of $root$ (by the inductive hypothesis).
Since the child pointer changed by $S$ was traversed while a process was searching for a key that it inserted or deleted, and $\infty$ is greater than any such key, $S$ can replace only nodes with $key < \infty$.
\end{proof}

Using the same reasoning as in the proof of Lemma~\ref{lem-chromatic-invariants}.\ref{claim-chromatic-invariants-follow-template}, it is easy to verify that, each time the chromatic tree algorithm accesses a field of $r$, $r \neq \nil$.

\begin{defn} \label{defn-searchpath}
The \textbf{search path to} $key$ \textbf{starting at} a node $r$ is the path that an ordinary BST search starting at $r$ would follow.
If $r = root$, then we simply call this the search path to $key$.
\end{defn}

Note that this search is well-defined even if the data structure is not a BST.
Moreover, the search path starting at a node is well-defined, even if the
node has been removed from the tree.
In any case, we simply look at the path that an ordinary BST search would follow, if it were performed on the data structure.  

The following few lemmas establish that the tree remains a BST at all times and that 
searches are linearizable.  They are proved in a way similar to \cite{EFRB10:podc}, although
the proofs here must deal with the additional complication of rebalancing operations occuring
while a search traverses the tree.

\begin{lem}
\label{still-on-path}
If a node $v$ is in the data structure in some configuration $C$ and
$v$ was on the search path for key $k$ in some earlier configuration $C'$, 
then $v$ is on the search path for $k$ in $C$.
\end{lem}
\begin{proof}
Since $v$ is in the data structure 
in both $C'$ and $C$, it must be in the data structure at all times between $C'$ and $C$
by Lemma \ref{lem-dotreeupdate-rec-cannot-be-added-after-removal}.
Consider any successful \sct\ $S$ that occurs between $C'$ and $C$.  We show that it preserves
the property that $v$ is on the search path for $k$.
Suppose that the \sct\ changes a child pointer of a node $u$ from $old$ to $new$.
If $v$ is not a descendant of $old$ immediately before $S$, then this change cannot
remove $v$ from the search path for $k$.
So, suppose $v$ is a descendant of $old$ immediately prior to $S$.

Since $S$ does not remove $v$ from the data structure, $v$ must be a descendant of a node
$f$ in the fringe set $F$ of $S$ (see Figure~\ref{fig-replace-subtree} and Figure~\ref{fig-replace-subtree2}).
Moreover, $f$ must be on the search path for $k$ before $S$.
It is easy to check by inspection of each possible tree modification 
in Figure~\ref{fig-chromatic-rotations} that if the 
node $f\in F$ is on the search path for $k$ from $u$ prior to the modification,
then it is still on the search path for $k$ from $u$ after the modification.
So $f$ and $v$ are still on the search path for $k$ after S.
\end{proof}

Since a search only reads child pointers, and the tree may change as the search traverses the tree,
we must show that it still ends up at the correct leaf.
In other words, we must show that the search is linearizable even if it traverses some
nodes that are no longer in the tree.

\begin{lem}
\label{was-on-path}
If an ordinary BST search for key $k$ starting from the $root$ reaches a node $v$, then
there was some earlier configuration during the search when $v$ was on the search path for $k$.
\end{lem}
\begin{proof}
We prove this by induction on the number of nodes visited so far by the search.

{\bf Base case}:  $root$ is always on the search path for $k$.

{\bf Inductive step}:  Suppose that some node $v$ that is visited by the search was on the search path for $k$ in some configuration $C$ between the beginning of the search and the time that the search reached $v$.
Let $v'$ be the next node visited by the search after $v$.
We prove that there is a configuration $C'$ between $C$ and the time the search reaches $v'$ when
$v'$ is on the search path for $k$.  
Without loss of generality, assume $k < v.key$.  (The argument when $k\geq v.key$ is symmetric.)
Then, when the search reaches $v'$, $v'$ is the left child of $v$.
We consider two cases.

{\bf Case 1} When the search reaches $v'$, $v$ is in the data structure:
Let $C'$ be the configuration immediately before the search reads $v'$ from $v.left$.
Then, by Lemma \ref{still-on-path}, $v$ is still on the search path for $k$ in $C'$.
Since $k<v.key$, $v' = v.left$ is also on the search path for $k$ in $C'$.

{\bf Case 2} 
When the search reaches $v'$, $v$ is not in the data structure:  
If $v'=v.left$ at $C$, then $v'$ is in the data structure at $C$.
Otherwise, $v.left$ is changed to $v'$ some time after $C$, and when that change occurs,
$v$ has not been finalized, and therefore $v'$ is in the data structure after that change
(by Constraint~\conmarkallremovedrecs).
Either way, $v'$ is in the data structure some time at or after $C$.
 
Let $C'$ be the last configuration at or after $C$ when $v$ is in the data structure.
By Lemma \ref{still-on-path}, $v$ is on the search path for $k$ in $C'$.
Since the only steps that can modify child pointers are successful \sct s, the next step after $C'$
must be a successful \sct\ $S$.  
Since all updates to the tree satisfy Constraint~\conmarkallremovedrecs,
$v$ must be in the $R$-sequence of $S$.  
Thus, $v$ is finalized by $S$ and its left child pointer never changes again after $C'$.
So, in $C'$, $v.left = v'$.
Since $v$ is on the search path for $k$ in $C'$ and $k<v.key$, $v.left = v'$ is also on the search
path for $k$ in~$C'$.
\end{proof}

\begin{lem}
\label{BST-property}
At all times, the tree rooted at the left child of $root$ is a BST.
\end{lem}
\begin{proof}
Initially, the left child of $root$ is a leaf node, which is a BST.

Keys of nodes are immutable, and the only operation that can change child pointers are successful \sct s, so we prove that 
every successful \sct\ preserves the invariant.
Each \sct\ atomically implements one of the changes shown in Figure~\ref{fig-chromatic-rotations}.

The only \sct\ that can change a child of the $root$ is \ins 1 (when the
tree contains no keys) and \del\ (when the tree contains exactly one key).
By inspection, both of these changes preserve the invariant.

So, for the remainder of the proof
consider an \sct\ that changes a child pointer of 
some descendant of the left child of the $root$.
By inspection, the invariant is preserved by each rebalancing step, \del\
and \ins 2.

It remains to consider an \sct\ that performs \ins 1 to insert a new key $k$. 
Let $u$ be the node whose child pointer is changed by the insertion.
The node $u$ was reached by a search for $k$, so $u$ was on the search path for $k$ at some earlier
time, by Lemma \ref{was-on-path}.
Since $u$ cannot have been finalized prior to the \sct\ that changes its child pointer, 
it is still in the data structure, by Constraint~\conmarkallremovedrecs.
Thus, by Lemma \ref{still-on-path}, $u$ is still on the search
path for $k$ when the \sct\ occurs.
Hence, the \sct\ preserves the BST property.
\end{proof}

We define the linearization points for chromatic tree operations as follows.
\begin{compactitem}
\item \func{Get}($key$) is linearized at the time during the operation when the leaf reached was on the search path for $key$.  (This time exists, by Lemma \ref{was-on-path}.)
\item An \ins\ is linearized at its successful execution of \sct\ inside \tryins\ (if such an \sct\ exists).
\item A \del\ that returns $\bot$ is linearized at the time the leaf reached during the last execution of line \ref{del-search-line} was on the search path for $key$.  (This time exists, by Lemma \ref{was-on-path}.)
\item A \del\ that does not return $\bot$ is linearized at its successful execution of \sct\ inside \trydel\ (if such an \sct\ exists).
\item A \func{Successor} query that returns at line \ref{succ-return-empty} is linearized when it performs \llt($root$).
\item A \func{Successor} query that returns at line \ref{succ-return-l} at the time during the operation that $l$ was on the search path for $key$.  (This time exists, by Lemma \ref{was-on-path}.)

\item A \func{Successor} query that returns at line \ref{succ-return-succ} is linearized when it performs  its successful \vlt.
\end{compactitem}
It is easy to verify that every operation that terminates is assigned a linearization point during
the operation.

\begin{thm}
The chromatic search tree is a linearizable implementation of an ordered dictionary
with the operations \func{Get}, \ins, \del, \func{Successor}.
\end{thm}
\begin{proof}
Theorem~\ref{thm-dotreeup-linearizable} asserts that the \sct s implement atomic changes to the tree as shown in Figure~\ref{fig-chromatic-rotations}.  
By inspection of these transformations, the set of keys and associated values stored in leaves are not altered by any rebalancing steps.  Moreover, 
the transformations performed by each linearized \ins\ and \del\ maintain the invariant
that the set of keys and associated values stored in leaves of the tree is exactly the
set that should be in the dictionary.

When a \func{Get}$(key)$ is linearized, the search path for $key$ ends at the leaf found by
the traversal of the tree.  If that leaf contains $key$, \func{Get} returns the associated value,
which is correct.  If that leaf does not contain $key$, then, by Lemma \ref{BST-property}, it is nowhere else in the tree, so \func{Get} is correct to return $\bot$.

If \func{Successor}($key$) returns $\langle \bot,\bot\rangle$ at line \ref{succ-return-empty}, then at its linearization point,
the left child of the $root$ is a leaf.  By Lemma \ref{lem-chromatic-invariants}.\ref{claim-chromatic-invariants-top-of-tree}, the dictionary is empty.

If \func{Successor}($key$) returns $\langle l.k,l.v\rangle$ at line \ref{succ-return-empty}, then
at its linearization point, $l$ is the leaf on the search path for $key$.  So, $l$ contains either $key$ or its predecessor or successor at the linearization point.  Since $key<l.k$, $l$ is $key$'s successor.

Finally, suppose \func{Successor}($key$) returns $\langle succ.k,succ.v\rangle$ at line \ref{succ-return-succ}.
Then $l$ was on the search path for $key$ at some time during the search.
Since $l$ is among the nodes validated by the $\vlt$, it is not finalized, so it
is still on the search path for $key$ at the linearization point, by Lemma \ref{still-on-path}.
Since $key \geq l.k$, the successor of $key$ is the next leaf after $l$ in an in-order traversal of the tree.
Leaf $l$ is the rightmost leaf in the subtree rooted at the left child of $lastLeft$
and the key returned is the leftmost leaf in the subtree rooted at the right child of $lastLeft$.
The paths from $lastLeft$ to these two leaves are not finalized and therefore are in the tree.
Thus, the correct result is returned.
\end{proof}

\subsection{Progress}

\begin{lem} \label{lem-chromatic-tryrebalance-follows-template-infinitely-often}
If \tryrebalance\ is invoked infinitely often, then it follows the tree update template infinitely often.
\end{lem}
\begin{proof}
We first prove that each invocation $I'$ of \tryrebalance\ that returns at line~\ref{overweightleft-check-nil-rxxrlr} or line~\ref{overweightleft-check-nil-rxxrr} is concurrent with a tree update operation that changes the tree during $I'$.
Suppose not, to derive a contradiction.
Then, since the tree is only changed by \sct s, which are only performed by tree update operations, there is some invocation $I$ of \tryrebalance\ during which the tree does not change.
Suppose $I$ returns at line~\ref{overweightleft-check-nil-rxxrlr}.
Then, by inspection of \func{OverweightLeft}, $\rxxrl.w = 1$, and $\rxxrlr = \nil$, so $\rxxrl$ is a leaf.
Furthermore, since the tree does not change during $I$, $\rxxr$ is the parent of $\rxxrl$ and $\rxxr.w = 0$, and $\rxxl$ is the sibling of $\rxxr$ and $\rxxl.w > 1$.
Therefore, the sum of weights on a path from $root$ to a leaf in the sub-tree rooted at $\rxxr$ is different from $root$ to a leaf in the sub-tree rooted at $\rxxl$, so the tree is not a chromatic search tree, which is a contradiction.
The proof when $I$ returns at line~\ref{overweightleft-check-nil-rxxrr} is similar.

We now prove the stated claim.
To derive a contradiction, suppose \tryrebalance\ is invoked infinitely often, but only finitely many invocations of \tryrebalance\ follow the tree update template.
Then, Lemma~\ref{lem-chromatic-invariants}.\ref{claim-chromatic-invariants-follow-template} implies that, after some time $t$, every invocation of \tryrebalance\ returns at line~\ref{overweightleft-check-nil-rxxrlr} or line~\ref{overweightleft-check-nil-rxxrr}.
Let $I$ be an invocation of \tryrebalance\ that returns at line~\ref{overweightleft-check-nil-rxxrlr} or line~\ref{overweightleft-check-nil-rxxrr}, and that starts after $t$.
However, we have just argued that $I$ must be concurrent with a tree update operation that changes the tree during $I$ and, hence, returns at a line different from line~\ref{overweightleft-check-nil-rxxrlr} or line~\ref{overweightleft-check-nil-rxxrr} (by inspection of the code), which is a contradiction.
\end{proof}

\begin{thm}
The chromatic tree operations are non-blocking.
\end{thm}
\begin{proof}
To derive a contradiction, suppose there is some time $T$ after which some 
processes continue to take steps but none complete an operation.
We consider two cases.

If the operations that take steps forever do not include \ins\ or \del\ operations,
then eventually no process performs any \sct s (since the queries 
\func{get}, \func{Successor} and \func{Predecessor} do not
perform \sct s).  Thus, eventually all \llt s and \vlt s succeed, and therefore all queries
terminate, a contradiction.

If the operations that take steps forever include \ins\ or \del\ operations, then they repeatedly invoke
\tryins, \trydel\ or \tryrebalance.
By Lemma~\ref{lem-chromatic-invariants}.\ref{claim-chromatic-invariants-follow-template} and Lemma~\ref{lem-chromatic-tryrebalance-follows-template-infinitely-often}, infinitely many invocations of \tryins, \trydel\ or \tryrebalance\ follow the tree update template.
By Theorem~\ref{thm-dotreeup-progress}, infinitely many of these calls will succeed.
There is only one successful \tryins\ or \trydel\ performed by each process after $T$.
So, there must be infinitely many successful calls to \tryrebalance.
Boyar, Fagerberg and Larsen proved \cite{Boyar97amortizationresults} 
proved that after a bounded number of
rebalancing steps, the tree becomes a RBT, and then no further rebalancing
steps can be applied, a contradiction.
\end{proof}

\subsection{Bounding the chromatic tree's height}
\label{height-bound}

We now show that the height of the chromatic search tree is $O(\log n + c)$ where
$n$ is the number of keys stored in the tree and $c$ is the number of incomplete
\ins\ and \del\ operations.
Each \ins\ or \del\ can create one new violation.
We prove that no \ins\ or \del\ terminates until the violation it created is destroyed.
Thus, the number of violations in the tree at any time is bounded by $c$, and the required 
bound follows.

\begin{defn}
Let $x$ be a node that is in the data structure.  We say that $x.w-1$ \textbf{overweight violations occur at $x$} if $x.w >1$.
We say that a \textbf{red-red violation occurs at $x$} if $x$ and its parent in the data structure both have weight 0.
\end{defn}

The following lemma says that red-red violations can never be created at a node, except when the node is first added to the data structure.

\begin{lem}
\label{no-new-red-red}
Let $v$ be a node with weight 0.
Suppose that when $v$ is added to the data structure, its (unique) parent has non-zero weight. 
Then $v$ is never the child of a node with weight 0.
\end{lem}
\begin{proof}
Node weights are immutable.
It is easy to check by inspection of each transformation in Figure~\ref{fig-chromatic-rotations}
that if $v$ is not a newly created node and it
acquires a new parent in the transformation with weight 0,
then $v$ had a parent of weight 0 prior to the transformation.
\end{proof}

\begin{defn}
A process $P$ is \textbf{in a cleanup phase for} $k$ if it is executing an \ins$(k)$ or a \del$(k)$ and it has performed a successful \sct\ inside a \tryins\ or \trydel\ that returns $createdViolation=\true$. 
If  $P$ is between line \ref{cleanup-start} and \ref{cleanup-end}, 
$location(P)$ and $parent(P)$ are the values of $P$'s local variables $l$ and $p$; otherwise $location(P)$ is the $root$ node and $parent(P)$ is \nil.
\end{defn}

We use the following invariant to show that each violation in the data structure has a
pending update operation that is responsible for removing it before terminating:
either that process is on the way towards the violation, or it will find another violation and
restart from the top of the tree, heading towards the violation.

\begin{lem}
In every configuration, there exists an injective mapping $\rho$ from violations to processes such that, for every violation $x$, 
\begin{compactitem}
\item
{\rm (A)} process $\rho(x)$ is in a cleanup phase for some key $k_x$ and 
\item
{\rm (B)} $x$ is on the search path from $root$ for $k_x$ and
\item {\rm (C)} either\\ 
{\rm (C1)} the search path for $k_x$ from $location(\rho(x))$ contains the violation $x$, or\\
{\rm (C2)} $location(\rho(x)).w=0$ and $parent(\rho(x)).w=0$, or\\
{\rm (C3)} in the prefix of the search path for $k_x$ from $location(\rho(x))$ up to and including
the first non-finalized node (or the entire search path if all nodes are finalized), there
is a node with weight greater than 1 or two nodes in a row with weight~0.
\end{compactitem}
\end{lem}
\begin{proof}
In the initial configuration, there are no violations, so the invariant is trivially satisfied.
We show that any step $S$ by any process $P$ preserves the invariant.  
We assume there is a function $\rho$ satisfying the claim for the configuration $C$ immediately before $S$ and show that there is a function $\rho'$ satisfying the claim for the configuration $C'$ immediately after $S$.
The only step that can cause a process to leave its cleanup phase is
the termination of an \ins\ or \del\ that is in its cleanup phase.
The only
steps that can change $location(P)$ and $parent(P)$ are $P$'s execution of line \ref{cleanup-end} or the read of the child pointer on line \ref{move-l-left} or \ref{move-l-right}.  (We think of all
of the updates to local variables in the braces on those lines as happening atomically
with the read of the child pointer.)
The only steps that can change child pointers or finalize nodes  
are successful \sct s.   No other steps $S$ can cause the invariant to become false.

{\bf Case 1} $S$ is the termination of an \ins\ or \del\ that is in its cleanup phase:
We choose $\rho'=\rho$.
$S$ happens when the test in line \ref{cleanup-terminate} is true, meaning that $location(P)$ is a leaf.
Leaves always have weight greater than 0.  The weight of the leaf cannot be greater than 1, because then
the process would have exited the loop in the previous iteration after the test at line  \ref{find-violation} returned true (since weights of nodes never change).
Thus, $location(P)$ is a leaf with weight 1.
So, $P$ cannot be $\rho(x)$ for any violation $x$, so $S$ cannot make the invariant become false.

{\bf Case 2} $S$ is an execution of line \ref{cleanup-end}:
We choose $\rho'=\rho$.
Step $S$ changes $location(P)$ to $root$.  
If $P\neq \rho(x)$ for any violation $x$, then this step cannot affect the truth of the invariant.  
Now suppose $P=\rho(x_0)$ for some violation $x_0$.
The truth of properties (A) and (B) are not affected by a change in $location(P)$
and property (C) is not affected for any violation $x\neq x_0$.
Since $\rho$ satisfies property (B) for violation $x_0$ before $S$, it will satisfy 
property (C1) for $x_0$ after $S$.

{\bf Case 3} $S$ is a read of the left child pointer on line \ref{move-l-left}:
We choose $\rho'=\rho$.
Step $S$ changes $location(P)$ from some node $v$ to node $v_L$, which is $v$'s left child when $S$ is performed.
If $P\neq \rho(x)$ for any violation $x$, then this step cannot affect the truth of the invariant.  
So, suppose $P=\rho(x_0)$ for some violation $x_0$.
By (A), $P$ is in a cleanup phase for $k_{x_0}$.
The truth of (A) and (B) are not affected by a change in $location(P)$
and property (C) is not affected for any violation $x\neq x_0$.
So it 
remains to prove that (C) is true for violation $x_0$ in $C'$.

First, we prove $v.w\leq 1$, and hence there is never an overweight violation at $v$.
If $v$ is the $root$, then $v.w=1$.
Otherwise, $S$ does not occur during the first iteration of \cleanup's inner loop.
In the previous iteration, $v.w\leq 1$ at line \ref{find-violation} (otherwise, the loop would
have terminated).

Next, we prove that there is no red-red violation at $v$ when $S$ occurs.
If $v$ is the root or is not in the data structure when $S$ occurs, 
then there cannot be a red-red violation at $v$ when $S$ occurs, by definition.
Otherwise, node $v$ was read as the child of some other node $u$ in the previous iteration
of \cleanup's inner loop
and line \ref{find-violation} found that $u.w\neq 0$ or $v.w\neq 0$ (otherwise the loop would have terminated).
So, at some time before $S$ (and when $v$ was in the data structure), 
there was no red-red violation at $v$.
By Lemma \ref{no-new-red-red}, there is no red-red violation at $v$ when $S$ is performed.

Next, we prove that (C2) cannot be true for $x_0$ in configuration $C$.
If $S$ is in the first iteration of \cleanup's inner loop, then $location(P)=root$, which has weight 1.
If $S$ is not in the first iteration of \cleanup's inner loop, then the previous iteration found
$parent(P).w\neq 0$ or $location(P).w\neq 0$ (otherwise the loop would have terminated).

So we consider two cases, depending on whether (C1) or (C3) is true in configuration $C$.

{\bf Case 3a} (C1) is true in configuration $C$:
Thus, when $S$ is performed, 
the violation $x_0$ is on the search path for $k_{x_0}$ from $v$, but it is not at $v$ (as argued above).
$S$ reads the {\it left} child of $v$, so $k_{x_0} < v.k$ (since the key of node $v$ never changes).
So, $x_0$ must be on the search path for $k_{x_0}$ from $v_L$.  This means (C1) is satisfied for $x_0$ in configuration $C'$.

{\bf Case 3b} (C3) is true in configuration $C$:
We argued above that $v.w \le 1$, so the prefix must contain two nodes in a row with weight 0.
If they are the first two nodes,
$v$ and $v_L$, then (C2) is true after $S$.  Otherwise, (C3) is still true after $S$.

{\bf Case 4} $S$ is a read of the right child pointer on line \ref{move-l-right}:
The argument is symmetric to Case 3.

{\bf Case 5} $S$ is a successful \sct:
We must define the mapping $\rho'$ for each violation $x$ in configuration $C'$.
By Lemma \ref{no-new-red-red} and the fact that node weights are immutable, no transformation in Figure~\ref{fig-chromatic-rotations} can create a new violation
at a node that was already in the data structure in configuration $C$.
So, if $x$ is at a node that was in the data structure in configuration $C$,
$x$ was a violation in configuration $C$, and $\rho(x)$ is well-defined.
In this case, we let $\rho'(x)=\rho(x)$.

If $x$ is at a node that was added to the data structure by $S$, then we must define
$\rho(x)$ on a case-by-case basis for all transformations described in Figure~\ref{fig-chromatic-rotations}.
(The symmetric operations are handled symmetrically.)

If $x$ is a red-red violation at a newly added node, 
we define $\rho'(x)$ according to the following table.

\medskip

\noindent
\begin{tabular}{|l|l|l|}\hline
Transformation & Red-red violations $x$ created by $S$ & $\rho'(x)$\\\hline
RB1 & none created & --\\\hline
RB2 & none created & --\\\hline
BLK & at $n$ (if $\ux.w=1$ and $u.w=0$) & $\rho$(red-red violation at one of $\uxll,\uxlr,\uxrl,\uxrr$)\footnotemark\\\hline
PUSH & none created & --\\\hline
W1,W2,W3,W4 & none created\footnotemark & -- \\\hline
W5 & at $n$ (if $\ux.w=u.w=0$) & $\rho$(red-red violation at $\ux$) \\\hline
W6 & at $n$ (if $\ux.w=u.w=0$) & $\rho$(red-red violation at $\ux$)\\\hline
W7 & none created & --\\\hline
INSERT1 & at $n$ (if $\ux.w=1$ and $u.w=0$) & process performing the \ins\\\hline
INSERT2 & none created & -- \\\hline
DELETE & at $n$ (if $\ux.w=\uxr.w=u.w=0$)&$\rho$(red-red violation at $\ux$)\\\hline
\end{tabular}

\footnotetext[1]{By inspection of the decision tree in Figure \ref{fig-decision-tree}, BLK is only applied if one of $\uxll,\uxlr,\uxrl$ or $\uxrr$ has weight 0, and therefore a red-red violation, in configuration $C$, and this red-red violation is eliminated by the transformation.}
\footnotetext[2]{By inspection of the decision tree in Figure \ref{fig-decision-tree}, W1, W2, W3 and W4 are applied only if $\ux.w>0$.}
\eric{Should the second footnote be removed now that we have changed the rebalancing steps to make this a precondition?}

\medskip

For each newly added node that has $k$ overweight violations after $S$,
$\rho'$ maps them to the $k$ distinct processes $\{\rho(q) : q\in Q\}$, where
$Q$ is given by the following table. 

\medskip

\noindent
\begin{tabular}{|l|ll|l|}\hline
Transformation & \multicolumn{2}{c|}{Overweight violations created by $S$} & Set $Q$ of overweight violations before $S$ \\\hline
RB1 & $\ux.w-1$ at $n$ &(if $\ux.w>1$) & $\ux.w-1$ at $\ux$\\\hline
RB2 & $\ux.w-1$ at $n$ &(if $\ux.w>1$) & $\ux.w-1$ at $\ux$\\\hline
BLK & $\ux.w-2$ at $n$ &(if $\ux.w>2$) & $\ux.w-2$ of the $\ux.w-1$ at $\ux$\\\hline
PUSH & $\ux.w$ at $n$  &(if $\ux.w>0$)  & $\ux.w-1$ at $\ux$, and 1 at $\uxl$\\
PUSH & $\uxl.w-2$ at $\nL$ &(if $\uxl.w>2$) & $\uxl.w-2$ of the $\uxl.w-1$ at $\uxl$\\\hline
W1 & $\ux.w-1$ at $n$  &(if $\ux.w>1$)  & $\ux.w-1$ at $\ux$\\
W1 & $\uxl.w-2$ at $\nLL$ &(if $\uxl.w>2$) & $\uxl.w-2$ of the $\uxl.w-1$ at $\uxl$\\
W1 & $\uxrl.w-2$ at $\nLR$ &(if $\uxrl.w>2$) & $\uxrl.w-2$ of the $\uxrl.w-1$ at $\uxl$\\\hline
W2 & $\ux.w-1$ at $n$ & (if $\ux.w>1$) & $\ux.w-1$ at $\ux$\\
W2 & $\uxl.w-2$ at $\nLL$ & (if $\uxl.w>2$) & $\uxl.w-2$ of the $\uxl.w-1$ at $\uxl$\\\hline
W3 & $\ux.w-1$ at $n$  &(if $\ux.w>1$)  & $\ux.w-1$ at $\ux$\\
W3 & $\uxl.w-2$ at $\nLLL$ &(if $\uxl.w>2$) & $\uxl.w-2$ of the $\uxl.w-1$ at $\uxl$\\\hline
W4 & $\ux.w-1$ at $n$ &(if $\ux.w>1$) & $\ux.w-1$ at $\ux$\\
W4 & $\uxl.w-2$ at $\nLL$ &(if $\uxl.w>2$) & $\uxl.w-2$ of the $\uxl.w-1$ at $\uxl$\\\hline
W5 & $\ux.w-1$ at $n$ &(if $\ux.w>1$) & $\ux.w-1$ at $\ux$\\
W5 & $\uxl.w-2$ at $\nLL$ &(if $\uxl.w>2$) & $\uxl.w-2$ of the $\uxl.w-1$ at $\uxl$\\\hline
W6 & $\ux.w-1$ at $n$ &(if $\ux.w>1$) & $\ux.w-1$ at $\ux$\\
W6 & $\uxl.w-2$ at $\nLL$ &(if $\uxl.w>2$) & $\uxl.w-2$ of the $\uxl.w-1$ at $\uxl$\\\hline
W7 & $\ux.w$ at $n$  &(if $\ux.w>0$)  & $\ux.w-1$ at $\ux$, and 1 at $\uxl$\\
W7 & $\uxl.w-2$ at $\nL$ &(if $\uxl.w>2$) & $\uxl.w-2$ of the $\uxl.w-1$ at $\uxl$\\
W7 & $\uxr.w-2$ at $\nR$ &(if $\uxr.w>2$) & $\uxr.w-2$ of the $\uxr.w-1$ at $\uxr$\\\hline
INSERT1 & $\ux.w-2$ at $n$ &(if $\ux.w>2$) & $\ux.w-2$ of the $\ux.w-1$ at $\ux$\\\hline
INSERT2 & $\ux.w-1$ at $n$ &(if $\ux.w>1$) & $\ux.w-1$ at $\ux$\\\hline
DELETE & $\ux.w+\uxr.w-1$ at $n$ & (if $\ux.w+\uxr.w>1$) & $\max(0,\ux.w-1)$ at $\ux$ and\\
	&							&					  & $\max(0,\uxr.w-1)$ at $\uxr$\footnotemark\\
\hline
\end{tabular}

\footnotetext{In this case, the number of violations in $Q$ is one too small if both $\ux.w$ and $\uxr.w$ are greater than 0, so the remaining violation is assigned to the process that performed the \del's \sct.}

\medskip

The function $\rho'$ is injective, since $\rho'$ maps each violation created by $S$ to
a distinct process that $\rho$ assigned to a violation that has been removed by $S$, with only two exceptions:  for red-red violations caused by INSERT1 and one overweight violation
caused by DELETE, $\rho'$ maps 
the red-red violation to the process that has just begun its cleanup phase (and therefore
was not assigned any violation by $\rho$).

Let $x$ be any violation in the tree in configuration $C'$.  We show that $\rho'$ satisfies
properties (A), (B) and (C) for $x$ in configuration $C'$.

{\bf Property (A)}:
Every process in the image of $\rho'$ was either in the image of $\rho$ or a process that just
entered its cleanup phase at step $S$, so every process in the image of $\rho'$ is in its
cleanup phase.  

{\bf Property (B) and (C)}: We consider several subcases.

{\bf Subcase 5a}
Suppose $S$ is an INSERT1's \sct, and $x$ is the red-red violation  created by $S$.
Then, $P$ is in its cleanup phase for the inserted key, which is one of the children of the node containing the red-red violation $x$.
Since the tree is a BST, $x$ is on the search path for this key, so (B) holds.

In this subcase, $location(\rho'(x)) = root$ since $P=\rho'(x)$ has just entered its cleanup phase.
So property (B) implies property (C1).

{\bf Subcase 5b}
Suppose $S$ is a DELETE's \sct, and $x$ is the overweight violation assigned to $P$ by $\rho'$. 
Then, $P$ is in a cleanup phase for the deleted key, which was in one of the children of $\ux$ before $S$.
Therefore, $x$ (at the root of the newly inserted subtree) is on the search path for this key, so (B) holds.

As in the previous subcase, $location(\rho'(x)) = root$ since $P=\rho'(x)$ 
has just entered its cleanup phase.
So property (B) implies property (C1).

{\bf Subcase 5c}
If $x$ is at a node that was added to the data structure by $S$ (and is not covered by the above
two cases), then $\rho'(x)$ is
$\rho(y)$ for some violation $y$ that has been removed from the tree by $S$, as described
in the above two tables.  
Let $k$ be the key such that process $\rho(y)=\rho'(x)$ is in the cleanup phase for $k$.
By property (B), $y$ was on the search path for $k$ before $S$.
It is easy to check by inspection of the tables and Figure~\ref{fig-chromatic-rotations} that 
any search path that went through $y$'s node in configuration $C$ goes through $x$'s node in configuration $C'$.
(We designed the tables to have this property.)
Thus, since $y$ was on the search path for $k$ in configuration $C$, 
$x$ is on the search path for $k$ in configuration $C'$, satisfying property (B).

If (C2) is true for violation $y$ in configuration $C$, then (C2) is true for $x$
in configuration $C'$ (since
$S$ does not affect $location()$ or $parent()$ and $\rho(y)=\rho'(x)$).
If (C3) is true for violation $y$ in configuration $C$, then (C3) is true for $x$ in 
configuration $C'$ (since any node that is finalized remains finalized forever, and its child pointers do not change).

So, for the remainder of the proof of subcase 5c, suppose (C1) is true for $y$ in configuration $C$.
Let $l=location(\rho(y))$ in configuration $C$.
Then $y$ is on the search path for $k$ from $l$ in configuration $C$.

First, suppose $S$ removes $l$ from the data structure.
\begin{compactitem}
\item
If $y$ is a red-red violation at node $l$ in configuration $C$, then
the red-red violation was already there when process $\rho(y)$ read $l$ as the child 
of some other node (by Lemma \ref{no-new-red-red}) and (C2) is true for $x$ in configuration $C'$.
\item
If $y$ is an overweight violation at node $l$ in configuration $C$, then it makes (C3) true for $x$ in configuration $C'$.
\item
Otherwise, 
since both $l$ and its descendant, the parent of the node that contains $y$, are removed by $S$,
the entire path between these two nodes is removed from the data structure by $S$.
So, all nodes along this path are finalized by $S$ because Constraint~\conmarkallremovedrecs\ is satisfied.  Thus, the violation $y$ makes
(C3) true for $x$ in configuration $C'$.
\end{compactitem}

Now, suppose $S$ does not remove $l$ from the data structure.
In  configuration $C$, the search path from $l$ for $k$ contains $y$.
It is easy to check by inspection of the tables defining $\rho'$ and Figure~\ref{fig-chromatic-rotations} that 
any search path from $l$ that went through $y$'s node in configuration $C$ 
goes through $x$'s node in  configuration  $C'$.
So, (C1) is true in configuration $C'$.

{\bf Subcase 5d}
If $x$ is at a node that was in the data structure in configuration $C$, 
then $\rho'(x)=\rho(x)$.
Let $k$ be the key such that this process is in the cleanup phase for $k$.
Since $x$ was on the search path for $k$ in configuration $C$ 
and $S$ did not remove $x$ from the data structure,
$x$ is still on the search path for $k$ in configuration $C'$ (by inspection of Figure~\ref{fig-chromatic-rotations}).
This establishes property (B).

If (C2) or (C3) is true for $x$ in configuration $C$, 
then it is also true for $x$ in configuration $C'$, for
the same reason as in Subcase 5c. 

So, suppose (C1) is true for $x$ in configuration $C$.
Let $l=location(\rho(x))$ in configuration $C$.  
Then, (C1) says that $x$ is on the search path for $k$ from $l$ in configuration $C$. 
If $S$ does not change any of the child pointers on this path between $l$ and $x$, then
$x$ is still on the search path from $location(\rho'(x)) = l$ in configuration $C'$, 
so property (C1) holds for $x$ in $C'$.
So, suppose $S$ does change the child pointer of some node on this path from $old$ to $new$.
Then the search path from $l$ for $k$ in configuration $C$
goes through $old$ to some node $f$ in the Fringe set 
$F$ of $S$ and then onward to the node containing violation $x$.
By inspection of the transformations in Figure~\ref{fig-chromatic-rotations},
the search path for $k$ from $l$ in configuration $C'$ 
goes through $new$ to the same node $f$, and then
onward to the node containing the violation $x$.
Thus, property (C1) is true for $x$ in configuration $C'$.
\end{proof}

\begin{cor}
\label{violation-bound}
The number of violations in the data structure is bounded by the number of incomplete \ins\ and \del\ operations.
\end{cor}

In the following discussion, we are discussing ``pure'' chromatic trees, without the dummy
nodes with key $\infty$ that appear at the top of our tree.
The sum of weights on a path from the root to a leaf of a chromatic tree
is called the \textit{path weight} of that leaf.
The {\it height} of a node $v$, denoted $h(v)$ is the maximum number of nodes on a path from $v$ to a leaf descendant of $v$.
We also define the \textit{weighted height} of a node $v$ as follows.
 \begin{displaymath}
   wh(v) = \left\{
     \begin{array}{ll}
       v.w & \mbox{if } v \mbox{ is a leaf} \\
       \max(wh(v.left), wh(v.right))+v.w & \mbox{otherwise}
     \end{array}
   \right.
  \end{displaymath}

\begin{lem} \label{lem-chromatic-claims}
Consider a chromatic tree rooted at $root$ that contains $n$ nodes and $c$ violations.
Suppose $T$ is any red black tree rooted at $root_T$ that results from performing a sequence 
of rebalancing steps on the tree rooted at $root$ to eliminate all violations.
Then, the following claims hold.
\begin{enumerate}
\item $h(root) \le 2wh(root) + c$
\label{claim-chromatic-h-wh}
\item $wh(root) \le wh(root_T) + c$
\label{claim-chromatic-wh-whT}
\item $wh(root_T) \le h(root_T)$
\label{claim-chromatic-whT-hT}
\end{enumerate}
\end{lem}
\begin{proof}
\textbf{Claim~\ref{claim-chromatic-h-wh}:}
Consider any path from $root$ to a leaf.
It has at most $wh(root)$ non-red nodes.
So, there can be at most $wh(root)$ red nodes that do not have red parents on the path (since $root$ has weight 1).
There are at most $c$ red nodes on the path that have red parents.
So the total number of nodes on the path is at most $2wh(root)+c$.

\textbf{Claim~\ref{claim-chromatic-wh-whT}:}
Consider any rebalancing step that is performed by replacing some node 
$\ux$ by $n$ (using the notation of Figure~\ref{fig-chromatic-rotations}).
If $\ux$ is not the root of the chromatic tree, then $wh(\ux)=wh(n)$,
since the path weights of all leaves in a chromatic tree must be equal. 
(Otherwise, the path weight to a leaf in the subtree rooted at $n$ would become
different from the path weight to a leaf outside this subtree.)

Thus, the only rebalancing steps that can change the weighted height of the root are those
where $\ux$ is the root of the tree.
Recall that the weight of the root is always one.
If $\ux$ and $n$ are supposed to have different weights according to Figure~\ref{fig-chromatic-rotations},
then blindly setting the weight of the 
$n$ to one will have the effect of changing the weighted height of the root.
By inspection of Figure~\ref{fig-chromatic-rotations}, the only transformation that increases
the weighted height of the root is BLK, because it is the only transformation where
the weight of $n$ is supposed to be less than the weight of $\ux$.
Thus, each application of BLK at the root increases the weighted height of the root by one, but also
eliminates at least one red-red violation at a grandchild of the root (without introducing any new violations).  Since none of the rebalancing
steps increases the number of violations in the tree,
performing any sequence of steps that eliminates $c$ violations will change the weighted height of the root by at most $c$. 
The claim then follows from the fact that $T$ is produced by eliminating $c$ violations from the chromatic tree rooted at $root$.

\textbf{Claim~\ref{claim-chromatic-whT-hT}:}
Since $T$ is a RBT, it contains no overweight violations. 
Thus, the weighted height of the tree is a sum of zeros and ones. 
It follows that $wh(root_T) \le h(root_T)$.
\end{proof}

\begin{cor} \label{cor-chromatic-height-and-violations}
If there are $c$ incomplete \ins\ and \del\ operations and the data structure contains
$n$ keys, then its height is $O(log\ n + c)$.
\end{cor}
\begin{proof}
Let $root$, $T$, $root_T$ be defined as in Lemma~\ref{lem-chromatic-claims}.
We immediately obtain $h(root) \le 2h(root_T) + 3c$  from Corollary \ref{violation-bound} and Lemma~\ref{lem-chromatic-claims}.  Since the height of a RBT is $O(\log n)$, it follows that the height of our data structure is $O(\log n + c)$ (including the two dummy nodes at the top of the tree with key $\infty$).
\end{proof}

\end{document}